\documentclass[11pt]{amsart}
\usepackage{amsfonts,amssymb,amsmath}
\usepackage{epsfig}
\newcommand{\RR}{\mathbb R}

\newcommand\mycom[2]{\genfrac{}{}{0pt}{}{#1}{#2}}
\newtheorem{theorem}{Theorem}

\newtheorem{lemma}{Lemma}
\newtheorem{prop}{Proposition}
\newtheorem{corollary}{Corollary}
\newtheorem{remark}{Remark}
\newtheorem{definition}{Definition}
\newcounter{two}
\usepackage[usenames,dvipsnames]{color}

\begin{document}
\title [$\mathtt M$-curves and totally positive Grassmannians]{Rational degenerations of $\mathtt M$-curves, totally positive Grassmannians {and KP2--solitons}.}
%\subjclass{AMS Subject Classification {37K40, 37K20, 14H50, 14H70}}

\author{Simonetta Abenda}
\address{Dipartimento di Matematica,
Universit\`a degli Studi di Bologna, Italy,
simonetta.abenda@unibo.it}
\author{Petr G. Grinevich}
\address{L.D.Landau Institute for Theoretical Physics,
pr. Ak Semenova 1a, Chernogolovka, 142432, Russia,
{\footnotesize pgg@landau.ac.ru}\\
Lomonosov Moscow State University,
 Faculty of Mechanics and Mathematics, 
Russia, 119991, Moscow, GSP-1, 1 Leninskiye Gory, Main Building,\\
Moscow Institute of Physics and Technology, 
9 Institutskiy per., Dolgoprudny,
Moscow Region, 141700, Russia.}
\thanks{
This work has been partially supported by PRIN ``Teorie geometriche e analitiche dei sistemi Hamiltoniani in dimensioni finite e infinite'', by the project ``EQUATIONS'', by the Russian Foundation for Basic Research, grant 17-01-00366, by the program ``Fundamental problems of nonlinear dynamics'', Presidium of RAS }

\begin{abstract}
We establish a new connection between the theory of totally positive Grassmannians and the 
theory of $\mathtt M$-curves using the finite--gap theory for solitons of the KP equation.
Here and in the following KP equation denotes the Kadomtsev-Petviashvili 2 equation (see (\ref{eq:KP})), which 
is the first flow from the KP hierarchy. We also assume that all KP times are \textbf{real}.
We associate to any point of the real totally positive Grassmannian $Gr^{\mbox{\tiny TP}} (N,M)$ a reducible curve which is
a rational degeneration of an $\mathtt M$--curve of minimal genus $g=N(M-N)$, and we reconstruct the real 
algebraic-geometric data \'a la Krichever for the underlying real bounded multiline KP soliton solutions. 
From this construction it follows that these multiline solitons can be explicitly obtained by degenerating 
regular real finite-gap solutions corresponding to smooth $\mathtt M$-curves. In our approach 
we rule the addition of each new rational component to the spectral curve via an elementary Darboux transformation which corresponds to a section of a specific projection $Gr^{\mbox{\tiny TP}} (r+1,M-N+r+1)\mapsto Gr^{\mbox{\tiny TP}} (r,M-N+r)$.

\end{abstract}

\maketitle

\tableofcontents
\section{Introduction}

Since the seminal works of Schoenberg, Gantmacher and Krein \cite{Sch, GK}, totally positive matrices have appeared in connection with problems from different areas of pure and applied mathematics, including small vibrations of mechanical systems, statistical mechanics, approximation theory, combinatorics, graph theory (for more details see \cite{Lu,Pinkus}). A generalization of total positivity for generic reductive groups and their flag varieties was introduced by Lusztig in \cite{Lu1}. In \cite{Lu2} this construction was extended to the partial flag manifolds, including Grassmannians associated with arbitrary reductive groups. A parametrization of the totally non-negative part of a flag variety was obtained by Marsh and Rietsch in \cite{MR}. The development of criteria to test total positivity is also deeply related to the foundation of the cluster algebra theory of Fomin and Zelevinsky \cite{FZ1,FZ2}. Relevant recent applications include the development of Poisson geometry in cluster algebras\cite{GSV1}, and the computation of scattering amplitudes for on--shell diagrams for the planar limit of $\mathcal N=4$ super Yang Mills theory\cite{AGP1,AGP2}. In recent years, total positivity has proven to be an efficient tool to investigate the asymptotic properties of multi--line solitons for the KP equation \cite{BK,BPPP,BPP,  CK1, CK2,  DMH, KW1, KW2, KW3}
which we study in this paper in the finite--gap setting. 

The topological classification of real algebraic curves is also deeply connected with the theory of integrable systems and statistical models. In particular $\mathtt M$-curves \cite{Har}, {\sl i.e.} real algebraic curves with maximal number of 
components, naturally arise in the theory of integrable systems, such as the real finite-gap theory of the KP equation \cite{DN} the theory of finite-gap at one energy two-dimensional 
Schr\"odinger operators at the energies below the ground state \cite{VN1, VN2}, quantum cluster systems \cite {KG} and in statistical models such as dimer models \cite{KSO}.  Let us point out that in these papers the $\mathtt M$-curves (Harnack's curves) arise as the spectral curves for two-dimensional models with double-periodic boundary 
conditions or their quasiperiodic generalizations, and, generically, they are smooth.
An investigation of the relative positions of the branches of real algebraic curves 
of given degree (and similarly for algebraic surfaces) is the first part of the Hilbert's 16th problem. The term $\mathtt M$-curve was first 
introduced by Petrovsky \cite{Petr} (``$\mathtt M$'' means ``maximal''). Additional information about this topic can be found in the review papers \cite{Nat, Vi}.

In this paper we establish a connection of different nature between the theory of totally positive 
Grassmannians, and rational degenerations of $\mathtt M$ curves,
in the framework of the solitonic limit of real finite gap solutions for the Kadomtsev-Petviashvili 2 (KP) equation
\begin{equation}\label{eq:KP}
\partial_x \left (-4 \partial_t u + 6 u \partial_x u +\partial_{x}^3 u \right) + 3 \partial_y^2 u=0,
\end{equation}
where $\partial_z$ denotes the usual partial derivative with respect to the variable $z$.

We recall that \textbf{a relevant question is:} can all \textbf{real regular} KP multisoliton solutions be obtained by degenerating \textbf{real regular} finite-gap KP solutions. The importance of this problem was pointed out by S.P. Novikov. In this paper we provide a positive answer to the above question for real bounded regular multiline KP solitons, associated to the principal cell $Gr^{\mbox{\tiny TP}} (N,M)$.  

We remark that, even if the same soliton solution can be obtained from different degenerate algebraic--geometric data, Baker--Akhiezer functions are still an important tool of integration in such limiting case \cite{DKN}. For instance, in the case of the open Toda lattice, such approach has been used by Krichever and Vaninsky \cite{KV} to construct a Baker--Akhiezer function on a reducible singular curve which is the limit of the hyperelliptic spectral curve associated to the periodic Toda lattice. 

\smallskip

Multiline KP solitons can be obtained from the Wronskian method \cite{ZS,Mat1,Mat,H} starting from $N$ independent solutions of the heat hierarchy depending on $M$ phases and are naturally identified with a certain finite--dimensional reduction of the Sato Grassmannian \cite{S}. The KP multiline soliton solutions are real bounded and regular for all $(x,y,t)$ if and only if they are parametrized by points in the totally nonnegative part of the Grassmannian, $Gr^{\mbox{\tiny TNN}} (N,M)$ \cite{KW2}. The asymptotic behavior in space--time and the tropicalization for this class of solutions has been thoroughly investigated in \cite{CK1, CK2, KW1,KW2,KW3} using the combinatorial classification of Postnikov \cite{Pos} of $Gr^{\mbox{\tiny TNN}} (N,M)$. According to Krichever scheme of KP finite gap theory \cite{Kr1,Kr2}, this family of solutions also originates from regular quasi--periodic solutions in the solitonic limit.
 
In this paper we show that real bounded regular multiline KP soliton solutions originate from 
\textbf{regular real quasi--periodic finite-gap} solutions in the solitonic limit.

Regular real quasi--periodic solutions are parametrized by degree $g$ non--special divisors on genus $g$ Riemann surfaces which possess an antiholomorphic involution of decomposing type which fixes the maximum number $g+1$ of ovals \cite{DN}\footnote{In the following, with a slight abuse of notation, we call regular $\mathtt M$--curves the Riemann surfaces satisfying Dubrovin and Natanzon's hypotheses.}. In such a case, there exists a homological basis of cycles such that the $\alpha$ cycles correspond to $g$ (finite) ovals and the essential singularity of the wavefunction is in the remaining (infinite) oval. Finally, for a fixed spectral curve, such solutions are parametrized by degree $g$ non--special divisors of the Baker--Akhiezer function  with exactly one pole in each finite oval. 
In the solitonic limit a certain number of cycles shrinks and the non--singular curve degenerates to a reducible curve of rational type.

Let us fix the $M$ phases, $\kappa_1< \cdots <\kappa_M$. Then generic regular bounded multi--line KP solitons are parametrized by points in $Gr^{\mbox{\tiny TP}} (N,M)\subset Gr^{\mbox{\tiny TNN}} (N,M)$, {\sl i.e.} points in the Grassmannian with all Pl\"ucker coordinates strictly positive \cite{Pos}. These  solutions then depend on $N(M-N)$ parameters (the dimension of the Grassmannian).
We show that for any compact subset in $Gr^{\mbox{\tiny TP}} (N,M)$ we may fix a spectral curve, which is a rational degeneration  
of regular $\mathtt M$--curves of genus $N(M-N)$, and the solutions are parametrized by $N(M-N)$ point divisors on it.

The starting point of the construction is the following observation: for any soliton data in $Gr^{\mbox{\tiny TP}} (N,M)$, the Sato dressed wave function is defined on $\mathbb{CP}^1$, which we denote $\Gamma_0$, with $M+1$ marked points (the phases $\kappa_1,\dots,\kappa_M$ and the essential singularity $P_0$).  On $\Gamma_0$, the normalized Sato dressed wave function is a Baker--Akhiezer function for the given soliton data with a real  $N$--point divisor. To obtain a degree $N(M-N)$ divisor, in our construction we attach additional components to $\Gamma_0$ in such a way that the resulting reducible curve possesses the $N(M-N)+1$ real components (ovals), and we extend the Baker--Akhiezer function to these new components so that the degree of the divisor matches the number of finite ovals and the reality constraints of Dubrovin--Natanzon theorem. To make effective the relation between Sato's and finite--gap approaches for this family of KP soliton solutions in $Gr^{\mbox{\tiny TP}} (N,M)$, we use total positivity in classical sense. 
Our construction gives a new relation between the theory of integrable systems, the theory of algebraic curves and total positivity.
In \cite{AG}, we modify our approach and extend the present construction to all positroid cells in the stratification of $Gr^{\mbox{\tiny TNN}}(N,M)$ characterized in \cite{Pos}.
  
A relevant open question is the characterization of the regular $\mathtt M$--curves whose rational limit has been constructed here and in \cite{AG}.
Another important open question is also the thorough investigation of the following problem: start from a given reducible curve and a divisor compatible with the reality and regularity conditions in \cite{DN} and identify the soliton data in $Gr^{\mbox{\tiny TNN}} (N,M)$. In \cite{A}, one of us has identified the $(M-1)$--dimensional variety of soliton data in $Gr^{\mbox{\tiny TP}} (N,M)$ associated to a specific rational degeneration of a hyperelliptic curve and proved that the vacuum KP divisor for such soliton data coincides with the Toda divisor found in \cite{KV}. 

In a forthcoming paper we also plan to investigate explicitly the relations between our construction and the characterization of the asymptotic behaviour of the multi--line solitons in \cite{CK2,KW3}, {\sl i.e.} to characterize the asymptotics of the zero divisor $\mathcal{D} (x,y,t,\vec 0)$  on $\Gamma(\xi)$.
Finally it would be relevant to investigate possible connections between KP soliton theory and statistical models such as dimers
in the disk \cite{Lam} or with field theoretical models in the framework of the planar limit of ${\mathcal N}=4$--SYM \cite{AGP1,AGP2} through their common combinatorial characterization in $Gr^{\mbox{\tiny TNN}} (N,M)$.  

\subsection{Main results and plan of the paper}
The main result of this paper is that, for soliton solutions associated with the principal cell  $Gr^{\mbox{\tiny TP}} (N,M)$, we can always fix algebraic--geometric data on {\em reducible} curves, generating these solutions, with additional requirements that these curves are rational degenerations of some family of regular $\mathtt M$--curves and that the divisor points satisfy the reality and regularity conditions of \cite{DN}.

The first nontrivial case we study is $Gr^{\mbox{\tiny TP}} (1,M)$ (see also \cite{A}). In this case the curve $\Gamma$ is obtained by attaching a second copy of $\mathbb{CP}^1$ to $\Gamma_0$ at the phases $\kappa_1,\dots,\kappa_M$, and it is the rational degeneration of a hyperelliptic $\mathtt M$-curve of genus $M-1$. The divisor has degree $M-1$ and exactly one point in each real finite oval. We do the construction of the KP wavefunction on $\Gamma$ in two steps. We first extend the Sato vacuum wavefunction to a degree $M-1$ meromorphic function, $\Psi(P,\vec t)$\footnote{Here and in the following, unless differently specified, $\vec t$ means the whole sequence of times $\vec t= (x,y,t,t_4, t_5, \dots)$ associated to the KP hierarchy.}, on $\Gamma\backslash \{ P_0\}$ with exactly one simple pole in each finite oval (see (\ref{eq:vacHyp})) and such that at the Darboux marked point $Q_1$ belonging to the infinite oval,  $\Psi(Q_1, \vec t)$ coincides with the generator of the dressing transformation for the KP solution (see (\ref{eq:DarHyp})). Then, after normalization, the dressing of $\Psi$, ${\tilde \Psi} (P,\vec t)$  has a degree $M-1$ pole divisor with exactly one simple pole in each finite oval and its restriction to $\Gamma_0$ is the normalized Sato dressed wavefunction (see Lemma \ref{lemma:hyper}).

For soliton data in $Gr^{\mbox{\tiny TP}} (N,M)$, we use a similar scheme: we first 
glue $(N+1)$ copies of $\mathbb{CP}^1$, $\Gamma= \Gamma_0 \sqcup \Gamma_1 \sqcup \cdots \sqcup \Gamma_N$ in pairs at some real ordered points creating $N(M-N)+1$ ovals. We then extend Sato vacuum (zero potential) wavefunction on $\Gamma_0 \backslash \{ P_0\}$ to a meromorphic vacuum wavefunction on $\Gamma\backslash \{ P_0 \}$ with a degree $N(M-N)$ divisor on $\Gamma\backslash \Gamma_0$ such that exactly one divisor point belongs to each finite oval and, at each marked Darboux point $Q_r$, $r\in [N]$, we require that $\Psi$ coincides for all time with one of the $N$ generators of the dressing transformation associated to the given soliton data. Then, after normalization, the dressing of $\Psi$, ${\tilde \Psi} (P,\vec t)$  has a degree $N(M-N)$ pole divisor with exactly one simple pole in each finite oval and its restriction to $\Gamma_0$ is the normalized Sato dressed wavefunction (see Theorem \ref{theo:divisor}).
The advantage of this multistep procedure is that it is easier to control indirectly the position of the effective KP divisor imposing conditions on the vacuum wavefunction than to directly extend the Baker Akhiezer function to $\Gamma\backslash \{ P_0\}$.

In the general case $Gr^{\mbox{\tiny TP}} (N,M)$, the main technical problem is to coherently glue more copies of $\mathbb{CP}^1$ in pairs creating the necessary number of ovals, in such a way to control the position of the effective divisor. 

To solve this problem we use the following two-step  construction:
\begin{enumerate}
\item We provide an algebraic model of the vacuum wave function using total positivity to control the values of the vacuum 
wavefunction at all marked points. 
\item We introduce a family of curves depending on a parameter $\xi\gg1$, and we show that, for any point in 
$Gr^{\mbox{\tiny TP}}(N,M)$ there exists $\xi_0$ such that for any $\xi>\xi_0$ the analytic extension of the Sato vacuum wave function 
to the curve $\Gamma(\xi)$ coincides with the algebraic model.
\end{enumerate}

The topological type of the resulting reducible rational curve $\Gamma=\Gamma(\xi)$ is independent of $\xi\gg1$ and it is the same for all points in $Gr^{\mbox{\tiny TP}}(N,M)$. In particular the real part of $\Gamma$, $\Gamma_{\mathbb R}$ has $N(M-N)+1$ real ovals, and  $\Gamma$ is the rational degeneration of a regular $\mathtt M$-curve of minimal genus $g=N(M-N)$. In section \ref{sec:example} we construct explicitly the reducible nodal curve associated to generic points in $Gr^{\mbox{\tiny TP}} (2,4)$ and the underlying regular $\mathtt M$--curve of genus 4, which is a 3-sheeted covering of the Riemann sphere.

More precisely, the leading order behavior of the \textbf{vacuum} wavefunction $\Psi(P,\vec t)$ at all marked points of $\Gamma(\xi)$, for $\xi\gg1$, is ruled explicitly by an algebraic recursion (see Theorems \ref{lemma:vectors} and \ref{theo:main0}) associated to a conveniently normalized totally positive band matrix $\hat A$ representing the soliton datum in $Gr^{\mbox{\tiny TP}}(N,M)$. The exact behavior of $\Psi(P,\vec t)$ at the marked points on $\Gamma (\xi)$ is ruled via an upper triangular matrix $\hat A(\xi)$ explicitly computed in the proof of Theorem \ref{theo:main0}.  

Thanks to the recursive relations in Theorem \ref{lemma:vectors}, we control the sign of $\Psi$ at the double points and, therefore, the position of the vacuum divisor in the ovals of $\Gamma(\xi)$. In section \ref{sec:example} we explicitly compute such vacuum divisor for the algebraic curve associated to generic soliton data in $Gr^{\mbox{\tiny TP}} (2,4)$.
The characterization of the vacuum wave--function is contained in Theorem \ref{theo:main} where we prove that $\Psi(P, \vec t)$
\begin{enumerate}
\item coincides with Sato's vacuum wavefunction on $\Gamma_0$ ;
\item is meromorphic of degree $N(M-N)$ on $\Gamma\backslash \{P_0\}$;
\item possesses $M-N$ simple poles on each $\Gamma_r$, $r\in [N]$;
\item\label{it:div} possesses exactly one simple pole in each finite oval of $\Gamma$;
\item\label{it:dar} coincides with an explicit basis of heat hierarchy solution $f^{(r)}_{\xi} (	\vec t)$, which generate the Darboux transformation, at the marked Darboux point $Q_r\in \Gamma_r$, for any $r\in [N]$ and for all $\vec t$.
\end{enumerate}
Properties  (\ref{it:div}) and (\ref{it:dar}) are sufficient to guarantee the complete control of the position of the effective divisor of the normalized KP wavefunction $\tilde \Psi (P, \vec t)= \tilde \Psi_{\xi} (P, \vec t)$, which is obtained applying the dressing (Darboux) transformation to the vacuum wavefunction:
the effective divisor $\mathcal{D}=\mathcal{D} (\xi)$ satisfies both the conditions imposed by Dubrovin--Natanzon theorem and the constraints imposed by Sato dressing (see Theorem \ref{theo:divisor}). More precisely, we prove that $\tilde \Psi (P,\vec t)$ has the following properties:
\begin{enumerate}
\item its restriction to $\Gamma_0$ coincides with Sato's KP dressed wavefunction;
\item it is meromorphic of degree $N(M-N)$ on $\Gamma\backslash \{P_0\}$;
\item it possesses $N$ poles on $\Gamma_0$;
\item it possesses $M-N-1$ simple poles on each $\Gamma_r$, $r\in [N]$;
\item it possesses exactly one simple pole in each finite oval of $\Gamma$.
\end{enumerate}
We also provide explicit estimates for the position of the divisor $\mathcal D$ in Theorem \ref{theo:t0} in the given local coordinates. We stress that such local coordinates are associated to a totally positive basis in Fomin--Zelevinsky sense \cite{FZ1} and, in this sense, we establish a natural correspondence between points in $Gr^{\mbox{\tiny TP}} (N,M)$ and the algebraic--geometric data associated to the corresponding soliton solutions. 

%In \cite{AG}, we pay a price to extend the above construction to any given soliton datum in $Gr^{\mbox{\tiny TNN}} (N,M)$, since we are able to control the sign of the vacuum wavefunction only in a neighborhood of the initial condition $\vec t_0$ where only a finite number of times changes of value.

\smallskip

{\bf Plan of the paper:}
\begin{itemize}
\item In section \ref{sec:solitons}, we recall some known facts about regular finite gap and multi--soliton solutions of the KP equation.
\item In section \ref{sec:3}, we first associate the rational degeneration of a genus $(M-1)$ hyperelliptic $\mathtt M$-- curve and construct the effective divisor for soliton data in $Gr^{\mbox{\tiny TP}} (1,M)$. Then we present the main ideas of the algebraic--geometric construction for soliton data in $Gr^{\mbox{\tiny TP}} (N,M)$.
\item Section \ref{sec:mainres} contains the principal statements. In Theorem \ref{theo:divisor}, we extend the Baker--Akhiezer function for generic soliton data in $Gr^{\mbox{\tiny TP}} (N,M)$ to the curve $\Gamma(\xi)$. The proof follows immediately from Theorem \ref{theo:main}, in which we extend Sato's vacuum wavefunction to $\Gamma(\xi)$ so that it takes the desired values at all marked points (double points and Darboux points).  
\item The proof of Theorem \ref{theo:main} is carried out in detail in Section \ref{sec:proof}. The first part of the proof is algebraic and fixes the leading order behavior of the vacuum wavefunction at the double points and at the Darboux points (Lemma \ref{lemma:PAL} and Theorem \ref{lemma:vectors}). The matrix associated to such leading order behavior is in band form and we characterize its properties in Appendix~\ref{sec:totpos}. The second part of the proof (Theorem \ref{theo:main0}) is analytic and contains the explicit construction of the vacuum wavefunction. Some Lemmas necessary to its proof are in Appendix \ref{sec:lemmas}.
\item In section \ref{sec:divest} we characterize the pole divisor $\mathcal{D}$ and the zero divisor $\mathcal{D} (\vec t)$ of the KP--eigenfunction $\tilde \Psi (P, \vec t)$.
\item In section \ref{sec:example} we apply our construction to soliton data in $Gr^{\mbox{\tiny TP}} (2,4)$. In particular we explicitly  construct the curve $\Gamma(\xi)$ and show that it is the rational degeneration of an $\mathtt M$--curve of genus 4. We also construct the vacuum divisor and check the algebraic identities for the vacuum wavefunction in such case.
\end{itemize} 

\smallskip

{\bf Notation:} We use the following notations throughout the paper:
\begin{enumerate}
\item $N$ and $M$ are positive integers such that $N<M$;
\item  For $s\in {\mathbb N}$ let $[s] =\{ 1,2,\dots, s\}$; if $s,j \in {\mathbb N}$, $s<j$, then
$[s,j] =\{ s, s+1, s+2,\dots, j-1,j\}$;
\item For a given matrix $A$ we denote by  $A^{[i_1,\ldots,i_p]}_{[j_1,\ldots,j_q]}$ the 
$p\times q$ submatrix of $A$ formed by the elements $A^{i_m}_{j_l}$, $m\in [p]$,  
$l\in[q]$;
\item  If $p=q$, $\Delta^{[i_1,\ldots,i_p]}_{[j_1,\ldots,j_p]}(A)$ denotes the determinant of 
the submatrix  $A^{[i_1,\ldots,i_p]}_{[j_1,\ldots,j_p]}$.
\item For a given matrix $A$, $\Delta_{[j_1,\ldots,j_n]}(A)$ denotes the determinant of the $n\times n$ matrix, combined from the last $n$ rows of the columns $j_1$,\ldots, $j_n$;
\item  ${\vec t} = (t_1,t_2,t_3,\dots)$, where
$t_1=x$, $t_2=y$, $t_3=t$;
\item $\theta(\zeta,\vec t)= \sum\limits_{n=1}^{\infty} \zeta^n t_n,$
\item We denote the real phases 
$\kappa_1< \kappa_2 < \cdots < \kappa_M$ and
$\theta_j \equiv \theta (\kappa_j, \vec t)$.
\end{enumerate}

\medskip

\section{Multi--soliton KP solutions }\label{sec:solitons}

The KP equation (\ref{eq:KP}) was originally introduced by Kadomtsev and Petviashvili \cite{KP} to study the stability of the Korteweg de Vries equation under weak transverse perturbation in the $y$ direction. It has remarkable properties coming from the fact that it is the first non trivial flow \cite{ZS} of the so--called KP Hierarchy (see the monographs
\cite{D,DKN,H,MJD,NMPZ}). The family of regular real bounded KP multi--line soliton solutions may be characterized with different approaches: the Wronskian method, a special reduction of the Sato Grassmannian and in the solitonic limit of real finite--gap theory.
Below we briefly recall how these solutions may be obtained via these different approaches.

\subsection{The heat hierarchy and the dressing transformation}

Let $A =( A^i_j )$, be an $N\times M$ real matrix and fix $\kappa_1<\cdots<\kappa_M$.
{In the following $\vec t =(x,y,t,t_4,t_5,\dots)$ indicates an infinite number of real times unless specified differently.}
Following \cite{Mat1}, let  us consider $N$ linear independent solutions
\begin{equation}\label{eq:heatsol}
f^{(i)}(\vec t) = \sum_{j=1}^M A^i_j e^{\theta_j}, \quad i\in [N],\quad\quad
\end{equation}
to the heat hierarchy\footnote{We remark that the class of solutions to (\ref{eq:heat}) that we consider in this paper is not the general one.}
\begin{equation}\label{eq:heat}
\partial_y f =\partial_x^2 f, \quad\quad
\partial_{t_l} f = \partial_x^l f,\quad l=2,3,\dots,
\end{equation}
and define their Wronskian 
\begin{equation}
\tau (\vec t) = Wr (f^{(1)},\dots, f^{(N)}) \equiv \sum\limits_{I} \Delta_I (A)\prod_{\mycom{i_1<i_2}{ i_1,i_2 \in I}} (\kappa_{i_2}-\kappa_{i_1} ) e^{ \sum\limits_{i\in I} \theta_i }
\end{equation}
where the sum is other all $N$--element ordered subsets $I$ in $[M]$, {\it i.e.} $I=\{ 1\le i_1<i_2 < \cdots < i_N < M\}$ and $\Delta_I (A)$ are the maximal minors of the matrix $A$, {\it i.e.} the Pl\"ucker coordinates for the corresponding point in the finite dimensional Grassmannian $Gr (N,M)$.

The multi--line soliton solution to the KP equation (\ref{eq:KP}) is defined by the following formula:
\begin{equation}\label{eq:KPsol}
u( {\vec t} ) = 2\partial_{x}^2 \log(\tau ( {\vec t})).
\end{equation}
For the full KP hierarchy it is rather easy to check that the condition $\Delta_I (A) \ge 0$ for all $I$ is necessary and
sufficient to have (\ref{eq:KPsol}) regular and real bounded for all times. A very nontrivial result of \cite{KW2} states
that this condition is necessary and sufficient also in the case of the first flow from the KP hierarchy.
In such case, the equivalence class of $A$ , $[A]$ is a point in the totally non--negative Grassmannian \cite{Pos}
\[
Gr^{\mbox{\tiny TNN}} (N,M) = GL_N^+ \backslash Mat^{\mbox{\tiny TNN}}_{N,M}, 
\]
where $Mat^{\mbox{\tiny TNN}}_{N,M}$ is the set of real $N\times M$ matrices of maximal rank $N$ with nonnegative maximal minors $\Delta_I (A)$ and $GL_N^+$ is the group of $N\times N$ matrices with positive determinants.

Since left multiplication by $N\times N$ matrices with positive determinants
preserves the KP multisoliton solution $u({\vec t})$ in (\ref{eq:KPsol}),
there is a natural bijection between KP {regular real bounded} multi--line solitons (\ref{eq:KPsol})  and points in $Gr^{\mbox{\tiny TNN}}(N,M)$.

\medskip

Let us recall the construction of the wave function for the  multi--line soliton solutions.  According to Sato theory \cite{S} all KP soliton solutions may be obtained from the dressing (inverse gauge) transformation of the vacuum eigenfunction $\displaystyle \Psi^{(0)} (\zeta, \vec t) =\exp ( \theta(\zeta, {\vec t}))$, which solves
\[
\partial_x \Psi^{(0)} (\zeta, \vec t)=\zeta \Psi^{(0)} (\zeta, \vec t), \quad\quad
\partial_{t_l}\Psi^{(0)} (\zeta, \vec t) = \zeta^l \Psi^{(0)} (\zeta, \vec t),\quad l\ge 2,
\]
via the dressing ({\it i.e. } gauge) operator
\[
W ( {\vec t})= 1-\sum_{j=1}^{\infty} \chi_j({\vec t})\partial_x^{-N} ,
\]
under the condition that $W$ satisfies Sato equations
\[
\partial_{t_n} W = B_n W - W \partial_x^n, \quad\quad n\ge 1,
\]
where $B_n = (W \partial_x^n W^{-1} )_+$ is the differential part of the operator $W \partial_x^n W^{-1}$. Then
\[
L= W \partial_x W^{-1} = \partial_x + \frac{u(\vec t)}{2}\partial_x^{-1} +\cdots,  \quad\quad
u(\vec t) = 2\partial_x \chi_1 (\vec t),\]
and
\[
\hat\Psi^{(0)} (\zeta; \vec t)= W\Psi^{(0)} (\zeta; \vec t)
\]
are respectively the KP-Lax operator, the KP--potential (KP solution) and the KP-eigenfunction, {\sl i.e.}
\[
L \hat\Psi^{(0)} (\zeta; \vec t) =\zeta \hat\Psi^{(0)} (\zeta; \vec t), \quad\quad
\partial_{t_l} \hat\Psi^{(0)} (\zeta; \vec t)= B_l \hat\Psi^{(0)} (\zeta; \vec t),\quad l\ge 2,
\]
where $B_l = (W \partial_x^l W^{-1} )_+ =(L^l)_+ $.

\medskip

The dressing transformation associated to the line solitons 
(\ref{eq:KPsol}) corresponds to the following choice of the dressing operator
\[
W = 1 -w_1({\vec t})\partial_x^{-1} -\cdots - w_N({\vec t})\partial_x^{-N},
\]
where $w_1({\vec t}),\dots,w_N({\vec t})$ are uniquely defined as
solutions to the following linear system of
equations
\begin{equation}\label{eq:SL}
\partial_x^N f^{(i)} = w_1 \partial_x^{N-1} f^{(i)}+\cdots + w_N f^{(i)}, \quad i\in [N],
\end{equation}
and, in such case, $w_1 ({\vec t})= \partial_x \tau/\tau$ and $u({\vec t})=2\partial_x w_1({\vec t})=2\partial_x^2
\log(\tau)$.
Moreover
\begin{equation}\label{eq:D}
D^{(N)}\equiv W \partial_x^N =\partial_x^N -\partial_x^{N-1} w_1 (\vec t)-\cdots - w_N(\vec t).
\end{equation}

We observe that $w_1,\dots,w_N$ is the solution to the linear system (\ref{eq:SL}) if and only if
\begin{equation}\label{SLagain}
D^{(N)} f^{(i)}\equiv W \partial_x^N f^{(i)} =0, \quad \quad i\in [N].
\end{equation}
Moreover, if the above identity holds, then
\[
\partial_{t_l} (D^{(N)} f^{(i)}) =0, \quad \forall l\in \mathbb N,
\]
that is, by construction, the $N$-th order Darboux  transformation is associated with the $N$  eigenfunctions $f^{(1)}(\vec t),\dots,f^{(N)}(\vec t)$, of the KP Lax Pair with zero potential for the infinite multiplicity eigenvalue.

The KP-eigenfunction associated to this class of solutions is
\[
\hat\Psi^{(0)} (\zeta; \vec t) = W \Psi^{(0)} (\zeta; \vec t) = \left( 1 - w_1({\vec t})\zeta^{-1} - \cdots
- w_N({\vec t})\zeta^{-N} \right) e^{\theta (\zeta, {\vec t})},
\]
or, equivalently,
\begin{equation}\label{eq:Satowf}
D^{(N)}\Psi^{(0)} (\zeta; \vec t) \equiv  W\partial_x^N \Psi^{(0)} (\zeta; \vec t)
 = \left(\zeta^N -\zeta^{N-1} w_1 (\vec t)-\cdots - w_N(\vec t)\right)
\Psi^{(0)} (\zeta; \vec t) = \zeta^N \hat\Psi^{(0)} (\zeta; \vec t).
\end{equation}

\medskip

\subsection{Real finite--gap KP solutions and $\mathtt M$--curves}

The general method to construct periodic and quasi--periodic solutions to the KP equation is due to Krichever \cite{Kr1,Kr2}: let $\Gamma$ be a
smooth algebraic curve of genus $g$ with a marked point $P_0$ and let $\zeta^{-1}$ be a local parameter in $\Gamma$ in a neighborhood of $P_0$ such that $\zeta^{-1} (P_0)=0$. The triple $(\Gamma, P_0,\zeta^{-1})$ defines a family of exact solutions to (\ref{eq:KP}) parametrized by degree $g$ non-special divisors $\mathcal D$ defined on $\Gamma \backslash \{ P_0 \}$. 

The finite gap solutions of (\ref{eq:KP}) are constructed starting from the commutation representation\footnote{The representation of KP as commutation of operators is also known in literature as Zakharov--Shabat equation or zero--curvature condition.} \cite{ZS}
\begin{equation}\label{eq:ZS}
[ -\partial_y + B_2, -\partial_t +B_3] =0,
\end{equation}
where 
\[
B_2 \equiv (L^2)_+ = \partial_x^2 + u,\quad\quad B_3 = (L^3)_+ = \partial_x^3 +\frac{3}{4} (u\partial_x +\partial_x u) + \tilde u,
\]
and $\partial_x\tilde u =\frac{3}{4} \partial_y u$.
Then, the Baker-Akhiezer function $\tilde \Psi (P, \vec t)$ meromorphic on $\Gamma\backslash \{ P_0\}$, with poles at the points of the divisor $\mathcal D$ and essential singularity at $P_0$ of the form
\[
{\tilde \Psi} (\zeta, \vec t) = e^{ \zeta x +\zeta^2 y +\zeta^3 t +\cdots} \left( 1 - \chi_1({\vec t})\zeta^{-1} - \cdots
-\chi_N({\vec t})\zeta^{-N}  - \cdots\right) 
\]
is an eigenfunction of the following linear differential operators
\[
\partial_y {\tilde \Psi}  = B_2 {\tilde \Psi}, \quad\quad \partial_t {\tilde \Psi}  = B_3 {\tilde \Psi},
\]
and in such case, imposing compatibility condition (\ref{eq:ZS}), $u(\vec t) = 2 \partial_x \chi_1 (\vec t)$ satisfies the KP equation.

The divisor of poles $\mathcal D$ does not depend on the times $\vec t$. In contrast to it, the divisor of zeroes 
$\mathcal D(\vec t)$ depends on all times. The Abel transform of $\mathcal D(\vec t)$ is a linear function of times 
$\vec t$, therefore such transform linearizes the KP hierarchy. $\mathcal D$ is an effective divisor, therefore 
$\Psi(\zeta,\vec 0)\equiv 1$, and at the point $\vec t =\vec 0$ the divisor of zeroes coincides with the divisor of poles, 
$\mathcal D(\vec 0) = \mathcal D$. The last identity justifies the use of the same letter for both divisors.

After fixing a canonical basis of cycles $a_1,\dots,a_g,b_1,\dots,b_g$ and a basis of normalized holomorphic differentials $\omega_1,\dots,\omega_g$ on $\Gamma$, that is
$\oint_{a_j} \omega_k =2\pi i \delta_{jk}$, $\oint_{b_j} \omega_k = B_{kj}$, $j,k \in [g]$,
the KP solution takes the form
\begin{equation}\label{eq:fingap}
u(x,y,t) = 2\partial_x^2 \log \theta (xU^{(1)}+yU^{(2)}+t U^{(3)} + z_0)+c_1,\
\end{equation}
where $\theta$ is the Riemann theta function and $U^{(k)}$, $k\in [3]$ are vectors of the $b$--periods of the following normalized meromorphic differentials, holomorphic on $\Gamma\backslash \{ P_0\}$ and with principal parts $\hat \omega^{(k)} = d (\zeta^k ) +\dots$, $k\in [3]$, at $P_0$ (see \cite{Kr1,DN}).

The necessary and sufficient conditions for the smoothness and realness of the solution (\ref{eq:fingap}) associated with smooth curve $\Gamma$ of genus 
$g$ have been proven by Dubrovin and Natanzon (see \cite{DN} and references therein): $\Gamma$ must be an $\mathtt M$--curve, that is it possesses an 
antiholomorphic involution\footnote{Here $\bar{\cdot}$ denotes complex conjugation.} 
\[
\sigma : \Gamma \to \Gamma, \quad \sigma^2 =1 , \quad \sigma(P_0)=P_0, \quad \sigma^* (\zeta) = \bar \zeta,
\]
such that the set of fixed points of ${\sigma}$ consists of $g+1$ ovals (the maximum number of ovals \cite{GH}), $\Omega_0,\Omega_1,\dots,\Omega_g$. 
These ovals are called ``fixed'' or ``real''. The set of real ovals divides  $\Gamma$ into two connected components. 
Each of these components is homeomorphic to a sphere with $g+1$ holes. In Figure~\ref{fig:regM_g2}[left]  we show an example.

\begin{figure}[!tbp]
  \centering
  {\includegraphics[scale=0.15,angle=0,width=0.44\textwidth]{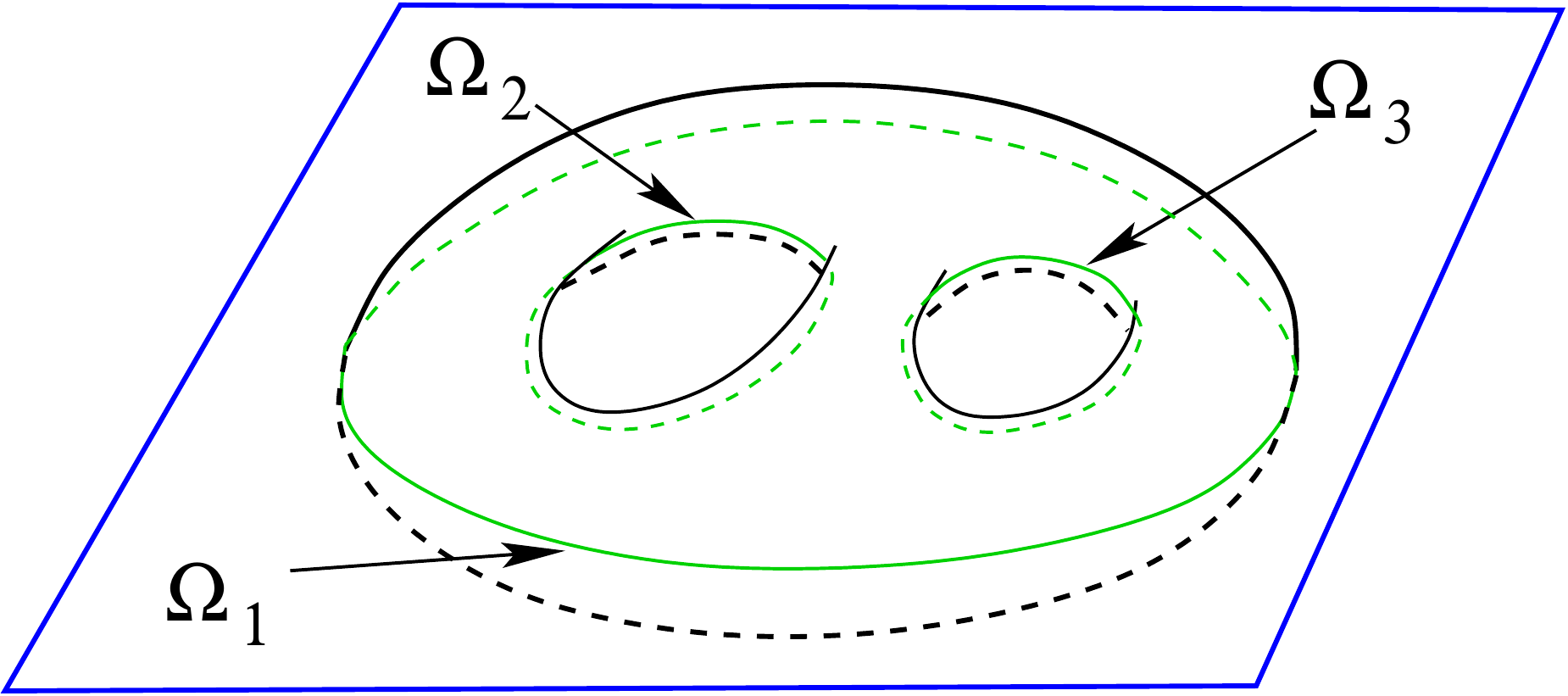}}
  \hfill
  {\includegraphics[scale=0.15,angle=0,width=0.44\textwidth]{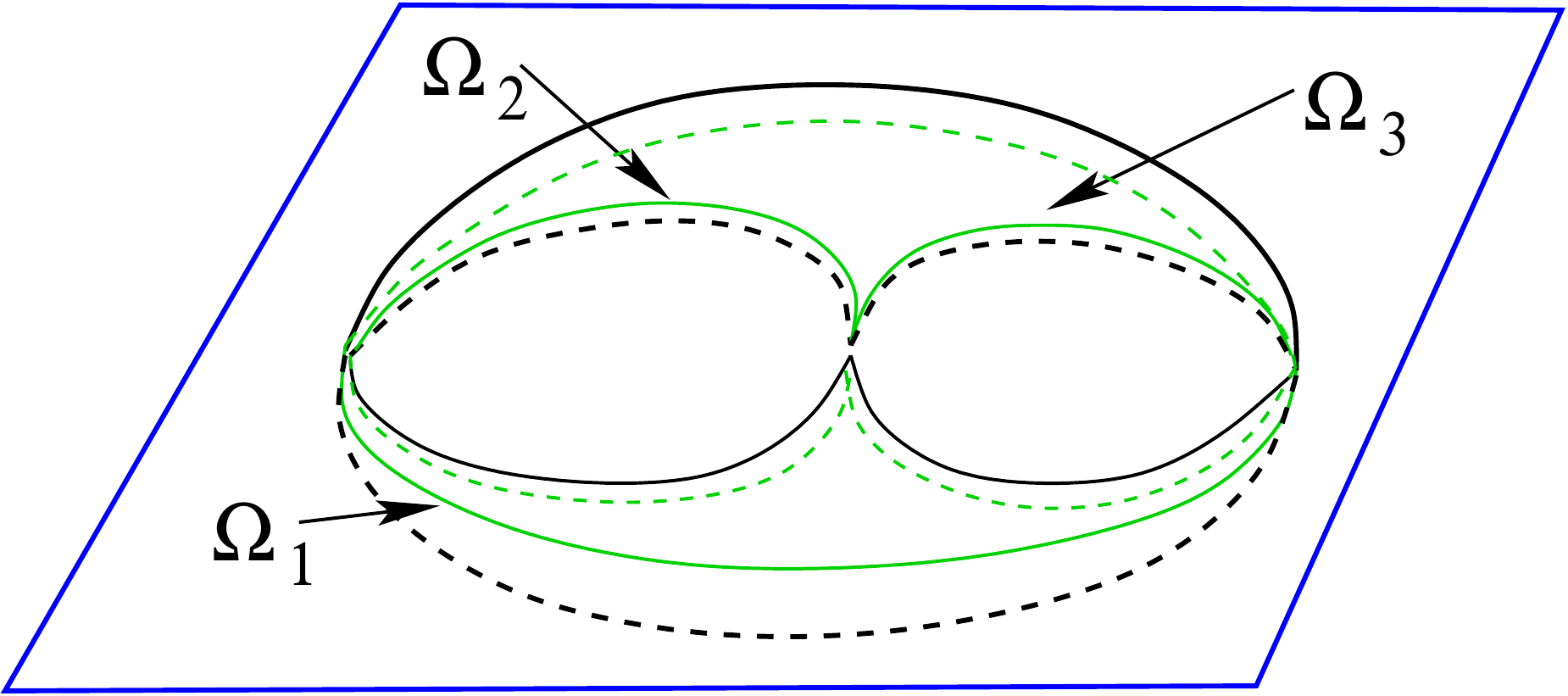}}
\caption{\small{\sl Left: a regular $\mathtt M$-curve,  $g=2$, 3 real ovals (painted
grey), 
involution $\sigma$ is reflection, orthogonal to the $\Pi$ plane. Right:
degeneration of a genus 2 $\mathtt M$-curve, 3 real ovals (painted
grey).}}\label{fig:regM_g2}        
\end{figure}

On such smooth $\mathtt M$--curve of  genus $g$ it is possible to fix a basis of cycles such that the essential singularity $P_0$ belongs to one oval $\Omega_0$
(which is called ``infinite'' oval), and the remaining $g$ fixed ovals $\Omega_j$, $j\in [g]$ coincide with the $a_j$-cycles of this basis:
\[
\sigma(a_j)=a_j, \quad \sigma(b_j)= -b_j,\quad\quad j\in [g].
\]
We call the ovals $\Omega_j$, $j\in [g]$  ``finite''. Finally, \textbf{in order to have regular real quasi--periodic solutions it 
is necessary and sufficient that each finite oval contains exactly one pole divisor point}\cite{DN}.

The proof that this condition is sufficient is rather simple. At $\vec t=\vec 0$ each finite oval contains 
exactly one zero divisor point. During real times evolution, each point of the zero divisor moves in the corresponding 
oval, and no point of the zero divisor can touch the oval $\Omega_0$. Taking into account that the singularities of the 
solution arise when a divisor point touches the point $P_0$, we see that for real times solutions are regular. 
It is easy to check that this argument is valid for degenerate $\mathtt{M}$-curves if the distance of the pole divisor from the intersection of the finite ovals and the infinite ovals is positively bounded from below.

\section{Algebraic-geometric approach to real bounded regular multiline KP solitons}\label{sec:3}

Soliton solutions of KP correspond to algebraic-geometric data associated to rational curves obtained by shrinking some cycles to double points (\cite{Kr3}, see also the book \cite{DKN} and references therein). 
After applying the Darboux transformation (\ref{eq:D}) we obtain the wave function (\ref{eq:Satowf}) defined on a Riemann sphere denoted by $\Gamma_0$ 
throughout the paper, whose effective divisor ${\mathcal D}^{(0)} = \{ \gamma^{(0)}_l \, ; l\in [M-1] \}\subset \Gamma_0 \backslash \{ P_0 \}$, consists 
of the $N$ real simple poles of the dressed wavefunction (\ref{eq:Satowf}) at time $\vec t= \vec 0 = (0,0,\dots)$,
\begin{equation}
\label{eq:Satodiv}
(\gamma^{(0)}_l)^N - w_1 (\vec 0) (\gamma^{(0)}_l)^{N-1} - \cdots - w_{N-1} (\vec 0) \gamma^{(0)}_l-w_N (\vec 0) = 0 ,\quad\quad l\in [N],
\end{equation}
which satisfy $\kappa_1 \le \gamma^{(0)}_1 < \gamma^{(0)}_2 < \cdots \gamma^{(0)}_N \le \kappa_M$ \cite{Mal}, after convenient labeling.

The sufficient part of Dubrovin and Natanzon's proof holds also when the algebraic  $\mathtt M$-curve is singular. Since 
the multi-soliton solutions associated to points in $Gr^{\mbox{\tiny TNN}} (N,M)$ are real bounded  regular for all $\vec t$, it is natural to expect that they may be associated 
to algebraic-geometric data on {\bf reducible} curves which are rational degenerations of regular $\mathtt M$--curves. In the following we provide such construction. For example, in 
 Figure~\ref{fig:regM_g2}[right] we show the rational degeneration of the smooth $\mathtt M$--curve of 
Figure~\ref{fig:regM_g2}[left]. For what concerns the degenerate $\mathtt M$-curve, we describe them in terms of their real parts represented 
as a collection of circles in the plane (with non-intersecting interiors) with marked points, where each marked point at one circle is connected to the corresponding marked point at another circle. For an example, see 
Figure~\ref{fig:fig2}[left]. Let us point out that this diagram is a pure topological representation, it does not respect the complex structure. 
To obtain a degenerate $\mathtt M$-curve it is 
necessary and sufficient that this diagram can be drawn in the plane without intersection (see Figure~\ref{fig:fig2}).

\smallskip

In our text we present a solution of the following 

\textbf{Problem.} \textbf{Associate a rational degeneration of a smooth $\mathtt M$--curve and a divisor satisfying the reality and regularity conditions}, to a regular real bounded multiline  soliton solution of KP represented by the following data:
\begin{enumerate}
\item The fixed numbers  $M$, $N$, $M>N$.
\item A set of $M$ real ordered phases $\kappa_1<\kappa_2<\dots<\kappa_M$.
\item A point $[A]$ in the totally positive Grassmannian $Gr^{\mbox{\tiny TP}} (N,M)$. 
\end{enumerate}

\smallskip

\begin{remark}
Of course, the solution to this problem is not unique, and our construction provides an infinite number of curves depending on a parameter. 
\end{remark}

To be more precise, in our text we associate the following algebraic-geometrical objects to a given soliton datum:
\begin{enumerate}
\item A reducible curve $\Gamma$ which is the rational degeneration of a smooth $\mathtt M$--curve of genus $N(M-N)$ equal to the dimension of 
the totally positive Grassmannian. $\Gamma$ has exactly $N(M-N)+1$ ovals. The curve $\Gamma_0$ in our approach is one of the irreducible components 
of $\Gamma$.
\item On $\Gamma$  we construct a unique wave--function $\tilde \Psi$ such that its restriction on  $\Gamma_0$ is the normalized Sato wave function.
\item On $\Gamma\backslash \{P_0\}$ the wave function $\tilde \Psi$ has effective divisor of degree $N(M-N)$  which coincides with (\ref{eq:Satodiv}) 
on the restriction to $\Gamma_0$. The essential singularity $P_0\in\Gamma_0$ lies in the infinite oval of $\Gamma$. Each finite oval of $\Gamma$ contains 
exactly one divisor point. 
\end{enumerate}

In our paper we construct the wave function  $\tilde \Psi$ in \textbf{three steps:}
\begin{enumerate}
\item We first extend the vacuum wave function $\Psi^{(0)}=e^{\theta (\zeta;\vec t)}$ from the curve $\Gamma_0$ to the 
curve $\Gamma$ as a meromorphic function $\Psi$ in $\Gamma\backslash \{P_0\}$. The poles of this function are called the 
\textbf{vacuum divisor}. The function $\Psi$ satisfies the heat hierarchy for all values of spectral parameter 
$\zeta\in\Gamma$.
\item We then apply the Darboux transformation (dressing) $D^{(N)}$ to $\Psi$. In the algebraic--geometric setting, the dressing corresponds to the following shift 
of the divisor: we add a $N$-th order pole to the point $P_0$ and $N$ simple zeroes to some points $Q_1$, \ldots, $Q_N$. 
The divisor of $ D^{(N)} \Psi$ is non-effective. 
\item We renormalize the dressed wave function:
$$
\tilde\Psi(\zeta,\vec t)= \frac{D^{(N)} \Psi(\zeta,\vec t)}{D^{(N)} \Psi(\zeta,\vec 0)}.
$$
The divisor of $\tilde\Psi$ is the effective KP divisor.
\end{enumerate}
\begin{remark}
\label{rem:rem2}
If we work with regular spectral curves, the shift of the divisor corresponds to a phase shift in the KP solution $u(\vec t)$. In the case of degenerate curves, the Jacobian consists of several components, and some of them correspond to trivial solutions of KP, and
other ones correspond to nontrivial solutions. By applying the Darboux transformation we jump from one component to 
another, therefore we generate nontrivial solutions from the trivial ones.
\end{remark}
\begin{remark}
The advantage of our multi-step procedure is that it is easier to control indirectly the position of the dressed KP divisor via the position of vacuum divisor than directly extend the Sato divisor to $\Gamma$. Indeed, we 
are able to control the sign of the vacuum wave function at all double points and its value at the Darboux points through an algebraic lemma.
\end{remark}

\begin{figure}[!tbp]
  \centering
  {\includegraphics[scale=0.15,angle=0,width=0.44\textwidth]{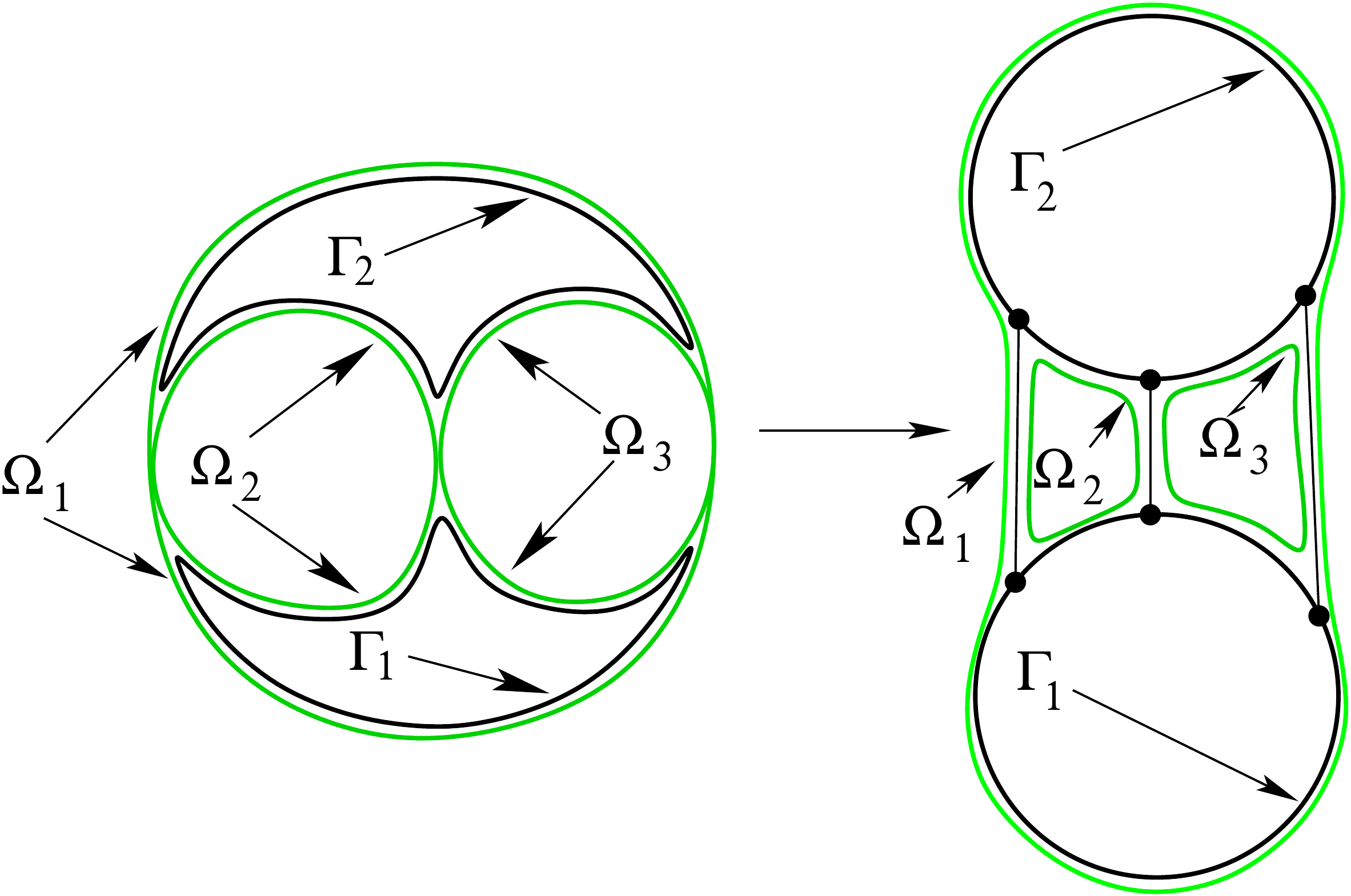}}
  \hspace{2cm}
  {\includegraphics[scale=0.15,angle=0,width=0.18\textwidth]{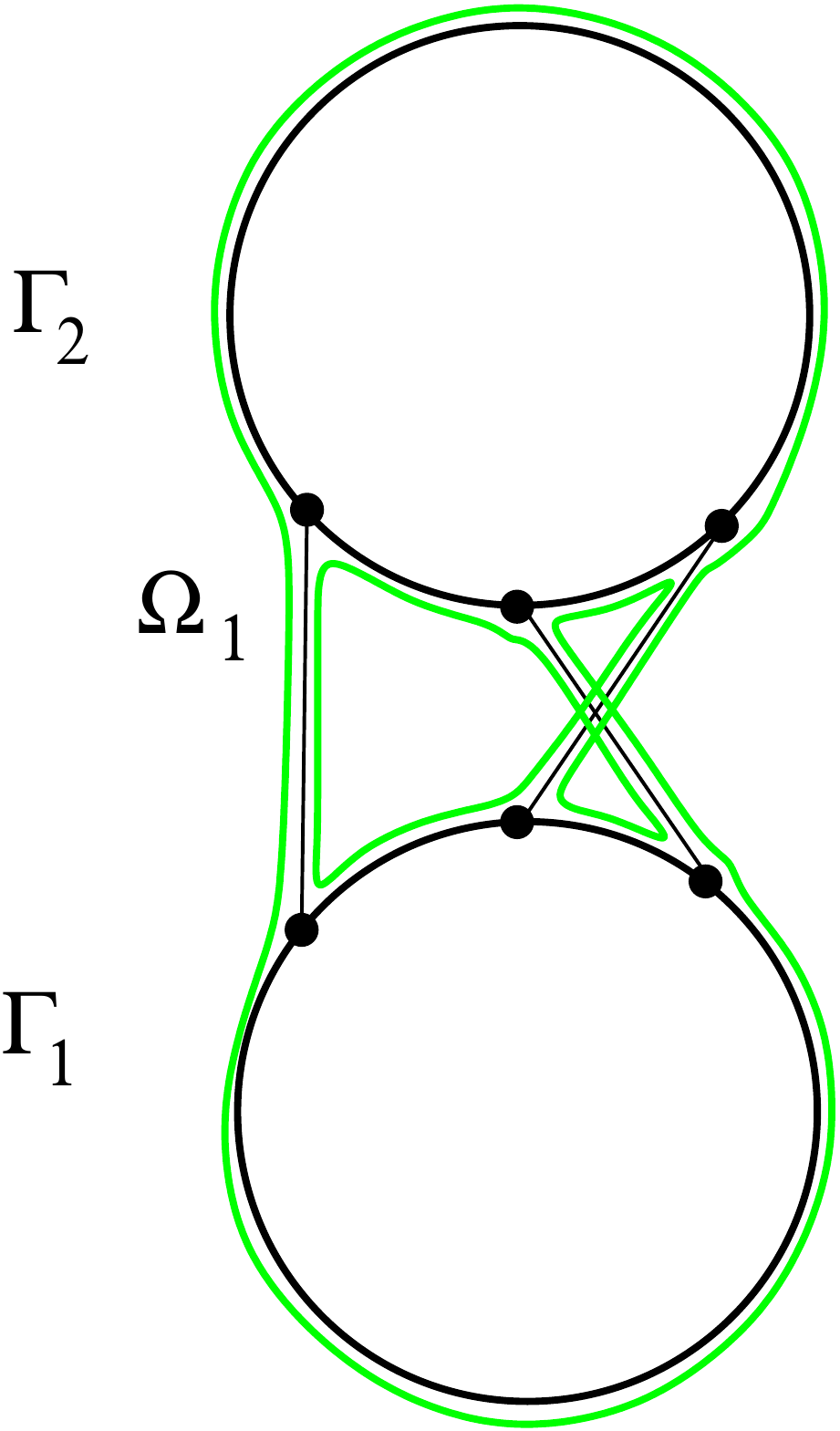}}
\caption{\small{\sl Left: The real part of degenerated $\mathtt M$- curve from the previous example is represented as a 
pair of circles with 3 connecting lines. 
 Right: Not a $\mathtt M$-curve, genus =2, the diagram is non-planar, 1 real oval}}
\label{fig:fig2}        
\end{figure}

\subsection{Algebraic-geometric construction for soliton data in $Gr^{\mbox{\tiny TP}} (1,M)$}\label{subseec:hyper} 

In the case of soliton data $[A]= [A_1,\dots,A_M] \in Gr^{\mbox{\tiny TP}} (1,M)$, the expected minimal degree of the divisor is $M-1 = \mbox{ dim } \left( Gr^{\mbox{\tiny TP}} (1,M) \right)$. Since Sato constraints produce a one point divisor on $\Gamma_0$, we need to glue at least another copy of $\mathbb{CP}^1$ to $\Gamma_0$ creating double points and $M$ ovals. Then, on such reducible curve, we must extend the wavefunction to a Baker--Akhiezer function with $M-1$ poles with the following constraints: the poles are simple, the essential singularity $P_0\in \Gamma_0$ belongs to one such oval and each other oval contains exactly one such pole.  

To achieve such a goal, we make the ansatz that $\Gamma$ is the rational degeneration of a hyperelliptic curve of genus $M-1$ 
such that all branch points are real. Let $\upsilon$ be the hyperelliptic involution and $\Gamma_1 = \upsilon (\Gamma_0)$. 
Here $\sigma$ is the standard complex conjugation of the Riemann spheres $\Gamma_0$, $\Gamma_1$, $\sigma\Gamma_0=\Gamma_0$, 
$\sigma\Gamma_1=\Gamma_1$. The double points $\kappa_j$, $j\in[M]$, are fixed points of both involutions $\sigma$, $\upsilon$.
Finally, let $Q_1 =\upsilon (P_0)$. By definition, $\Gamma=\Gamma_0 \sqcup\Gamma_1$ possesses $M$ ovals, $\Omega_0,\Omega_{1,1},\dots,\Omega_{1,M-1}$ such that $P_0, Q_1\in \Omega_0$ and, for each $l\in [M-1]$, $\Omega_{1,l}$ is the topological circle formed by the connected union of the intervals $[\kappa_l, \kappa_{l+1}]\subset \Gamma_0$ and $[\upsilon(\kappa_l), \upsilon(\kappa_{l+1})]\subset \Gamma_1$ (see Figure~\ref{fig:fig3}).

We then extend the vacuum wavefunction $\Psi^{(0)} (P, \vec t)$ to a meromorphic function $\Psi (P, \vec t)$ on $P\in\Gamma\backslash \{P_0\}$ with an $(M-1)$--point divisor in the intersection of $\Gamma_1$ with the finite ovals $\Omega_{1,l}$, $l\in [M-1]$. We also set the value of $\Psi (Q_1, \vec t)$ to the heat hierarchy solution $f^{(1)} (\vec t)$ so that the Darboux transformation
\[
D^{(1)} = \partial_x - \frac{\partial_x f^{(1)} (\vec t)}{f^{(1)} (\vec t)}
\]
creates a non--effective divisor with one zero at $Q_1$ for all $\vec t$. Then, after normalization, the effective KP divisor satisfies both the reality and regularity conditions in \cite{DN} and Sato's constraints.

Without loss of generality, let us choose the representative matrix $[A]\in Gr^{\mbox{\tiny TP}}(1,M)$ which satisfies $\sum_{l=1}^M A_l =1$. Then, the vacuum wavefunction on $\Gamma= \Gamma_0 \sqcup \Gamma_1$ is defined as follows
\begin{equation}\label{eq:vacHyp}
\Psi(\zeta,\vec t) = \left\{ \begin{array}{ll}
\Psi^{(0)}(\zeta;\vec t)=e^{\theta (\zeta;\vec t)}, &\quad \zeta \in \Gamma_0,\\
\displaystyle \Psi^{(1)} (\zeta; \vec t) = \frac{ \sum_{l=1}^{M} A_l E_l (\vec t) \prod_{j\not = l}^M
(\zeta -\kappa_j) }{\prod_{r=1}^{M-1} (\zeta-b_r)}, &\quad \zeta \in \Gamma_1
\end{array}
\right.
\end{equation}
with 
\[
\theta \equiv \theta (\zeta; \vec t) = \sum\limits_{i\ge 1} \zeta^i t_i,
\quad\quad\quad\quad
E_l (\vec t)= e^{\theta (\kappa_l;\vec t)}, \quad\quad l\in [M],
\]
where the vacuum divisor $b_j$  consists of $M-1$ real simple poles on $\Gamma_1$, uniquely defined by the gluing conditions:
\[
\Psi^{(0)}(\kappa_l ;\vec t)=\Psi^{(1)}(\kappa_l ;\vec t), \quad\quad l\in [M].
\]
It is easy to check that
\[
\kappa_1 < b_1 < \kappa_2 < b_2 <\kappa_3 < \cdots < \kappa_{M-1} < b_{M-1} < \kappa_M.
\]
Moreover, by definition,
\begin{equation}\label{eq:DarHyp}
\Psi^{(1)}(Q_1 ;\vec t) = \lim\limits_{\zeta\to \infty} \Psi^{(1)}(\zeta ;\vec t) = \sum\limits_{l=1}^{M}a_l E_l (\vec t) = f^{(1)} (\vec t), \quad\quad \forall \vec t,
\end{equation}
is the heat hierarchy solution generating the Darboux transformation $D^{(1)}$
for the given soliton data. Then, it is immediate to verify that the KP wavefunction 
\begin{equation}\label{eq:KPHyp}
{\tilde \Psi} (\zeta,\vec t) = \frac{D^{(1)} \Psi (\zeta, \vec t)}{D^{(1)} \Psi (\zeta, \vec 0)}=\left\{ \begin{array}{ll}
{\tilde \Psi}^{(0)}(\zeta;\vec t)=\frac{\zeta -\gamma^{(0)}_1(\vec t)}{\zeta - \gamma^{(0)}_1}e^{\theta (\zeta;\vec t)}, &\zeta \in \Gamma_0,\\
 {\tilde \Psi}^{(1)} (\zeta; \vec t) = \frac{ \sum\limits_{1\le j<l\le M}A_l A_j (\kappa_l-\kappa_j)^2 E_l (\vec t) E_j (\vec t)\prod\limits_{s\not = j,l}^M
(\zeta -\kappa_s) }{f^{(1)}(\vec t)\sum\limits_{1\le j<l\le M}A_l A_j (\kappa_l-\kappa_j)^2
\prod\limits_{s\not = j,l}^M(\zeta -\kappa_s)}, &\zeta \in \Gamma_1,
\end{array}
\right.
\end{equation}
has effective divisor ${\mathcal D}^{(1)} = \{ \gamma^{(0)}_1 ;  \gamma^{(1)}_1, \cdots, \gamma^{(1)}_{M-2}\}$, which satisfies the required conditions. Indeed the following Lemma holds.

\begin{figure}[!tbp]
  \centering
  {\includegraphics[scale=0.15,angle=0,width=0.36\textwidth]{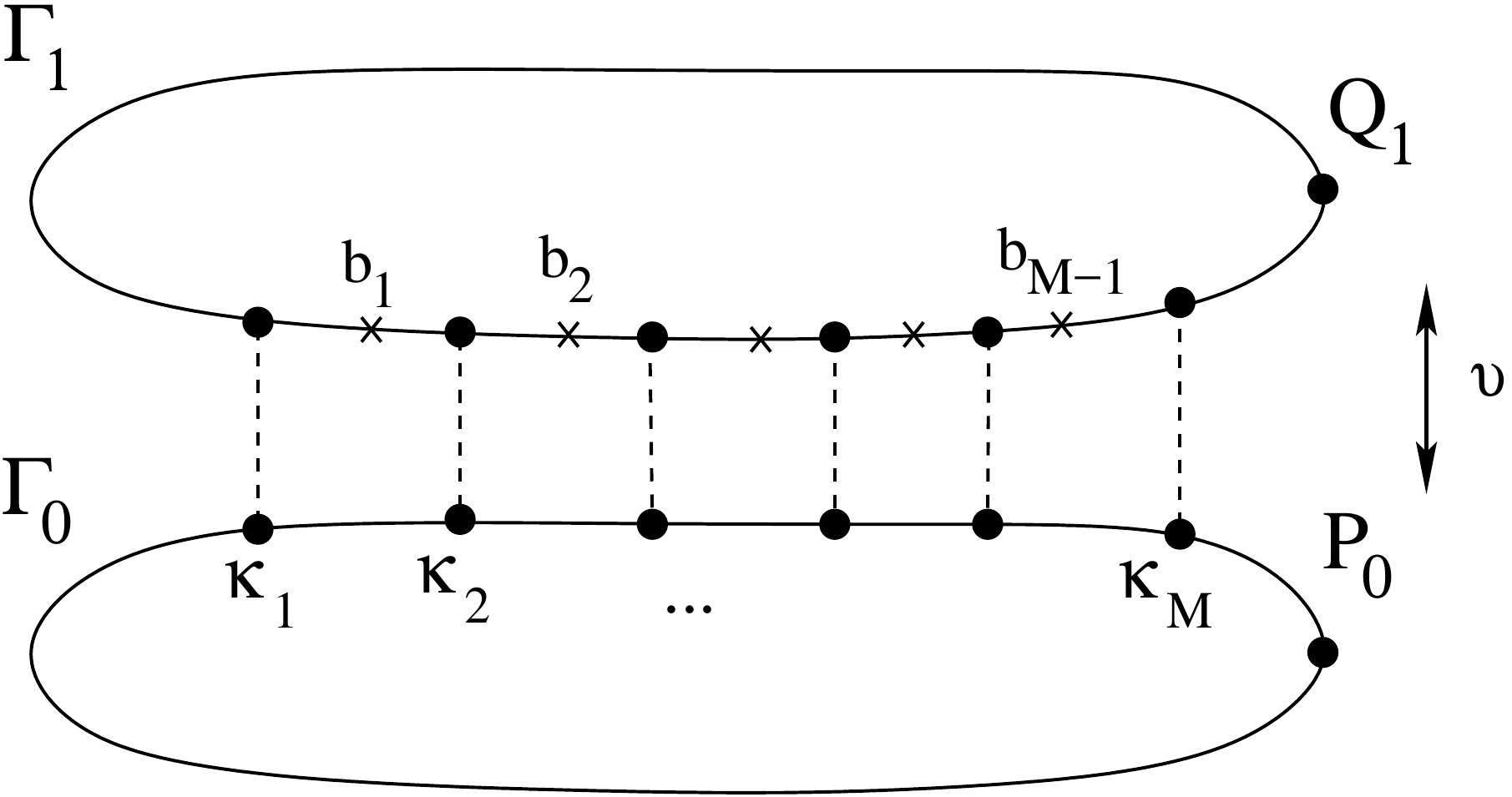}
     \hspace{1cm}
     \includegraphics[scale=0.15,angle=0,width=0.36\textwidth]{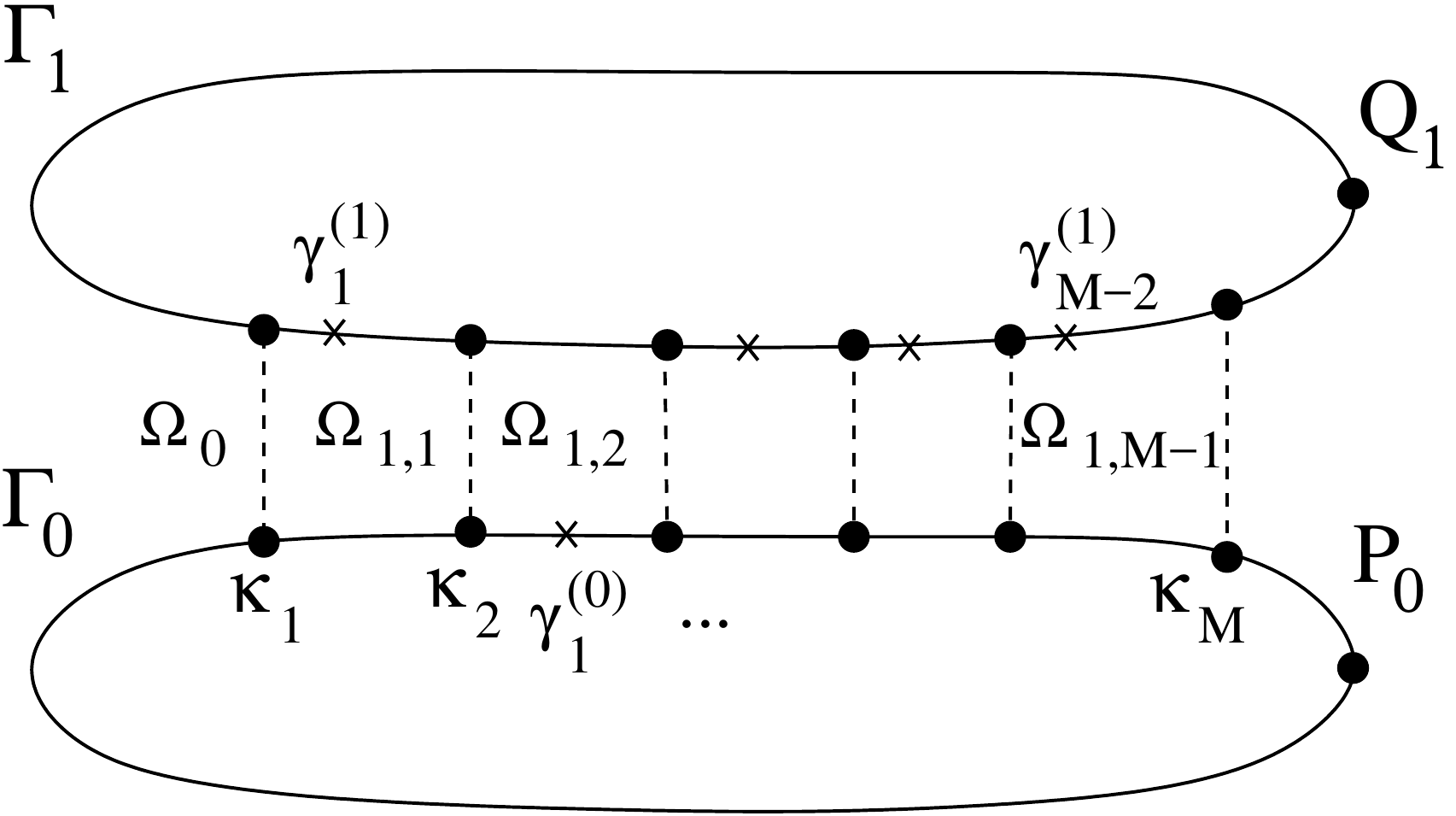}
  }
\caption{\small{\sl The rational degeneration of an hyperelliptic curve is associated to soliton data in  
$Gr^{\mbox{\tiny TP}}(1,M)$. Left: the vacuum divisor. Right: a possible effective divisor after the Darboux transformation. }}\label{fig:fig3}        
\end{figure}

\begin{lemma}\cite{A}\label{lemma:hyper}
Let $\kappa_1 < \cdots \kappa_M$, $[A] \in Gr^{\mbox{\tiny TP}} (1,M)$ as above, and let ${\tilde \Psi} (\zeta,\vec t)$ be as in (\ref{eq:KPHyp}). Then:
\begin{enumerate}
\item the restriction of the divisor ${\mathcal D}^{(1)}$ on $\Gamma_0$ consists of one pole: ${\mathcal D}^{(1)}\cap \Gamma_0 = \{ \gamma^{(0)}_1\}$, where $\gamma^{(0)}_1 = \frac{\sum_{j=1}^M A_j \kappa_j}{\sum_{j=1}^M A_j}$;
\item ${\mathcal D}^{(1)}\cap \Gamma_1 = \{ \gamma^{(1)}_1, \dots \gamma^{(1)}_{M-2}\}$, where $\gamma^{(1)}_r$, $r\in [M-2]$, are the real simple roots of $\sum_{1\le j<l\le M}A_l A_j (\kappa_l-\kappa_j)^2
\prod_{s\not = j,l}^M(\zeta -\kappa_s)=0$;
\item each finite oval contains exactly one pole and no pole is in the infinite oval, $\# \left(  {\mathcal D}^{(1)} \cap \Omega_{1,l}  \right) = 1$, for any $l\in [M-1]$, and
${\mathcal D}^{(1)} \cap \Omega_{0} = \emptyset$. 
\end{enumerate}
Moreover, if $\gamma_0$ coincides with a double point, that is for some ${\bar l} \in [2,M-1]$, $\gamma^{(0)}_1 = \kappa_{{\bar l}}$, then also $\gamma^{(1)}_{{\bar l}-1} =\kappa_{{\bar l}}$ and ${\mathcal D}^{(1)} \cap\left([\kappa_{{\bar l}-1}, \kappa_{{\bar l}+1}]\backslash \{ \kappa_{{\bar l}} \}\right) = \emptyset$.
\end{lemma}

We remark that the case in which $\gamma^{(1)}_0$ coincides with a double point is not generic and, for any $\epsilon >0$ there exists $\vec t_0 =(x_0,y_0,t_0,0,\dots)$ with $||\vec t_0||<\epsilon$ such that $\gamma^{(0)}_1 (\vec t_0)\not = \kappa_{{\bar l}} $. Finally, ${\mathcal D}^{(1)} \cap \Omega_0=\emptyset$ implies that,  under our hypotheses, the pole divisor satisfies a stronger conditions than in \cite{Mal} , that is
\[
\gamma^{(0)}_1 ,\gamma^{(1)}_l \in ]\kappa_1, \kappa_M[, \quad\quad \forall l\in [M-2].
\]

\smallskip

If we start from the vacuum wavefunction
as in (\ref{eq:vacHyp}) and apply the Darboux transformation $D^{(N)}$ associated to a generic point in $Gr^{\mbox{\tiny TP}} (N,M)$, we obtain an effective divisor which coincides with Sato's on $\Gamma_0$ and consists of further $M-1$ distinct poles on $\Gamma_1$.
In \cite{A} one of us (S.A.) has characterized the $(N-1)$--dimensional variety of soliton data 
in $Gr^{\mbox{\tiny TP}} (N,M)$ such that the Darboux transformation $D^{(N)}$  generates a non--effective divisor with a zero 
of order $N$ in $Q_1$, so that the effective divisor has $M-N-1$ real simple poles on $\Gamma_1$ and is compatible with 
Dubrovin and Natanzon conditions (see Remark~\ref{rem:rem2}). Such multi--line solitons are naturally related to the finite Toda lattice (see \cite{A} and references therein) and correspond to a well--defined immersion of $Gr^{\mbox{\tiny TP}} (1,M)\hookrightarrow Gr^{\mbox{\tiny TP}}(N,M)$, for any fixed $N\in [M-1]$. In particular, in \cite{A} it is remarked that the KP vacuum divisor $b_1,\dots, b_{M-1}$ of the KP vacuum wavefunction (\ref{eq:vacHyp}) on $\Gamma_1$ coincides with the Toda divisor in \cite{KV} upon identification of the KP phases with the Toda spectrum and of the soliton datum $[A]\in Gr^{\mbox{\tiny TP}} (1,M)$ in the Toda IVP.

\subsection{The algebraic-geometrical setting for soliton data in $Gr^{\mbox{\tiny TP}} (N,M)$} 

Now we present the main geometric construction of our paper.

\begin{definition}
\label{def:gamma}
\textbf{The rational spectral curve.}
Assume that we have the following data
\begin{enumerate}
\item The fixed numbers  $M$, $N$, $M>N$.
\item A set of $M$ real ordered phases $\kappa_1<\kappa_2<\dots<\kappa_M$.
\end{enumerate}
To such data we associate the curve $\Gamma$, obtained by gluing $N+1$ copies of $\mathbf{CP}^1$, $\Gamma_0$, 
$\Gamma_1, \cdots ,\Gamma_{N}$, in the following way:

Denote by $\zeta$ the local parameter on each copy of $\mathbb{CP}^1$.
We have the following marked points:
\begin{enumerate}
\item On $\Gamma_0$ we have $M$ real ordered marked points $\kappa_1<\kappa_2<\ldots<\kappa_M$ and the infinite point $P_0$, 
$\zeta(P_0) =\infty$.
\item On each $\Gamma_r$, $r\in[N]$ we have $M-N+1$ real negative ordered marked points in the local coordinate $\zeta$
\begin{equation}\label{eq:lambdas0}
\lambda_1^{(r)} =0> \lambda_2^{(r)}> \cdots \lambda_{M-N+1}^{(r)},
\end{equation}
$M-N$ real positive ordered marked points in the local coordinate $\zeta$
\begin{equation}\label{eq:alphas0}
0<\alpha_2^{(r)} <\alpha_3^{(r)}< \cdots <\alpha_{M-N+1}^{(r)},
\end{equation}
and the infinite point $Q_r$, $\zeta^{-1} (Q_r)=0$.
\end{enumerate}
To such data we associate the following set of gluing rules:
\begin{enumerate}
\item We glue $\lambda^{(1)}_l\in\Gamma_1$ to $\kappa_{N+l-1}\in \Gamma_0$, for $l\in [M-N+1]$;
\item For any $r\in [2, N]$, we glue $\lambda^{(r)}_1 \in \Gamma_r$ to $\kappa_{N-r+1}\in \Gamma_0$;
\item For any $r\in [2, N]$, and for all $l	\in [2,M-N+1]$, we glue $\lambda^{(r)}_l \in \Gamma_r$ to $\alpha^{(r-1)}_l \in \Gamma_{r-1}$.
\end{enumerate}
\end{definition}

\begin{figure}[!tbp]
  \centering
  {\includegraphics[scale=0.15,angle=0,width=0.38\textwidth]{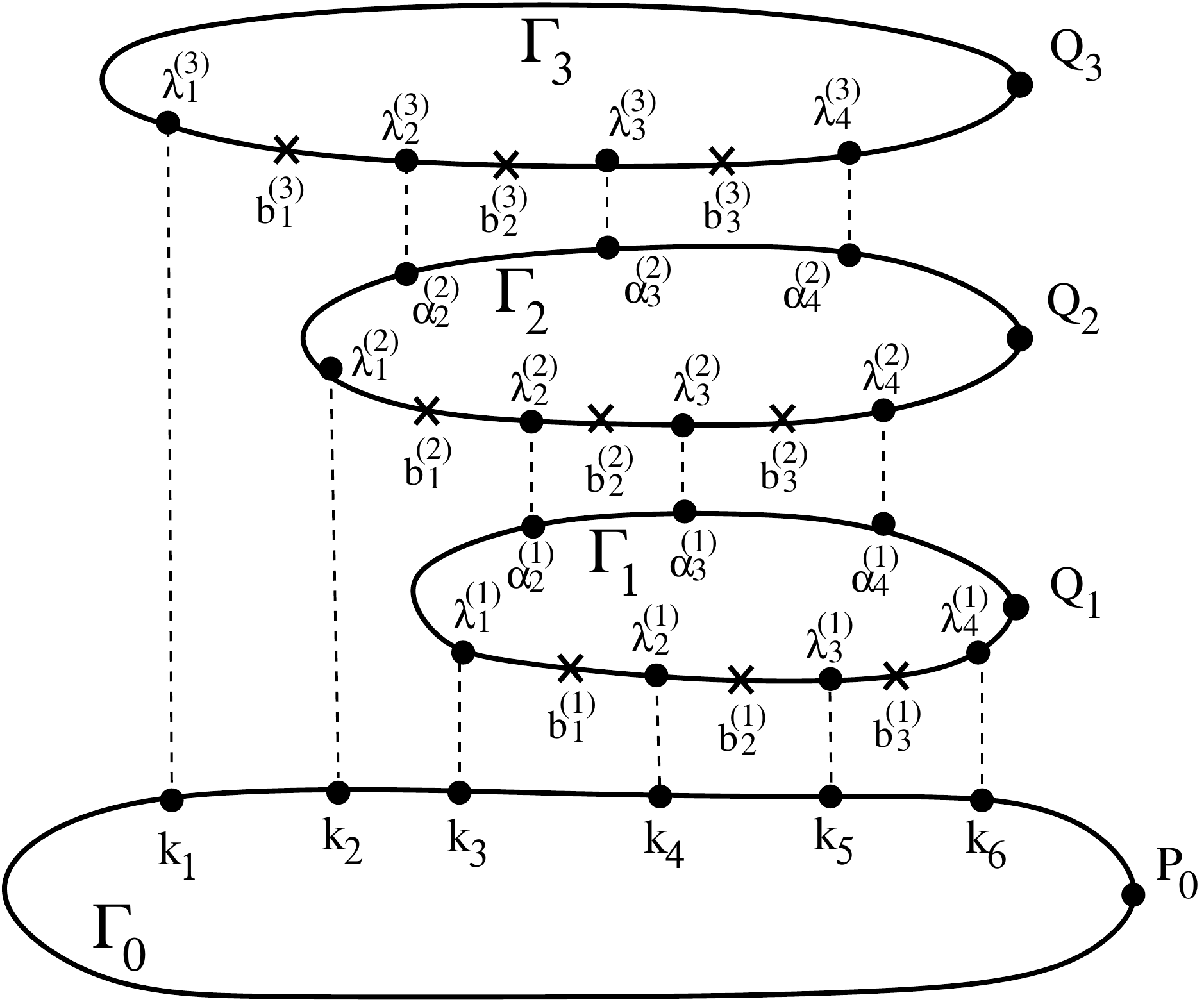}
     \hspace{1cm}
     \includegraphics[scale=0.15,angle=0,width=0.38\textwidth]{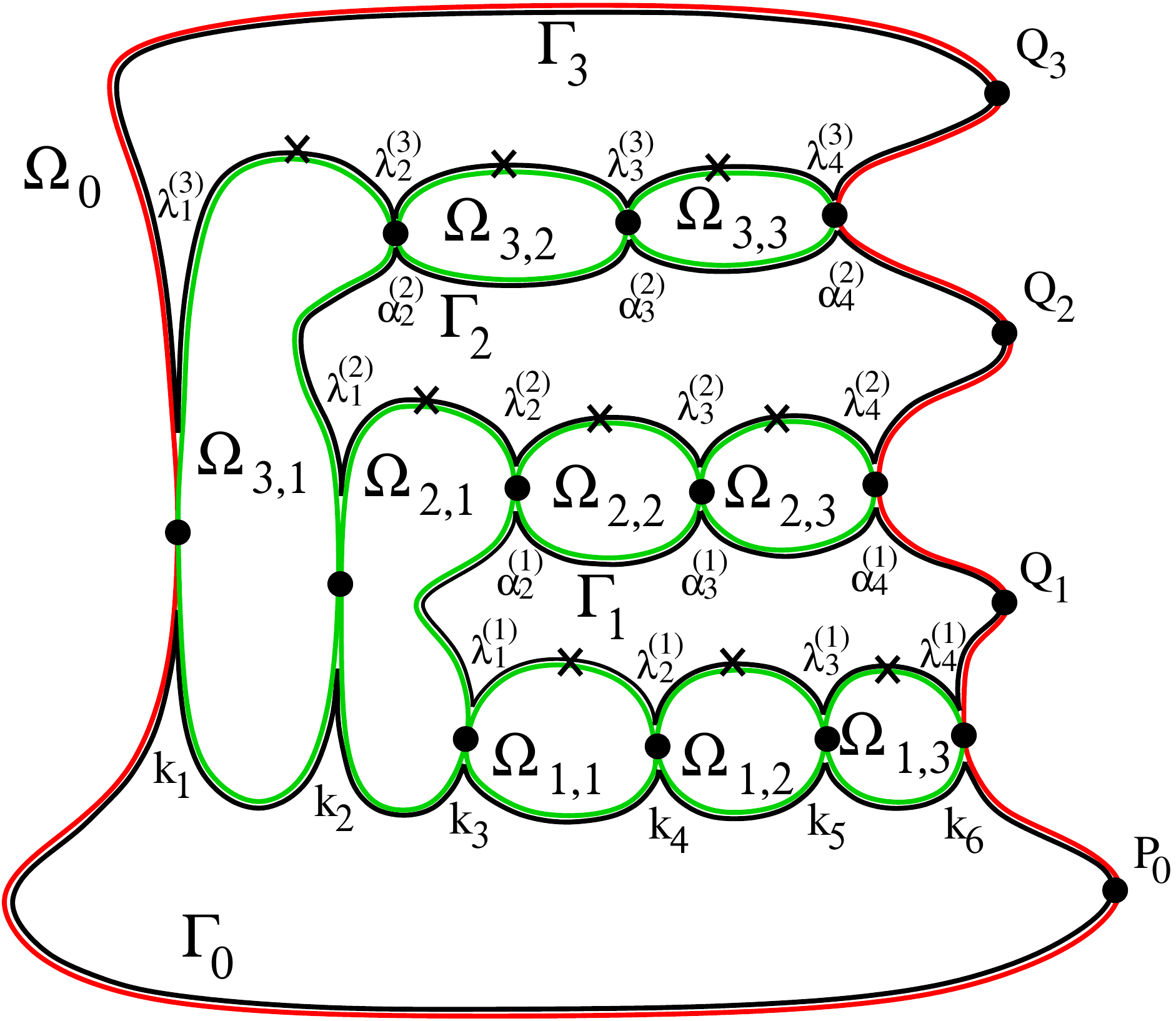}
  }
\caption{\small{\sl Left: The real part of $\Gamma$ for $M=6$, $N=3$. Right: The 10 real ovals of $\Gamma$ for $M=6$, $N=3$.}}
\label{fig:fig4}        
\end{figure}

Let us recall that we have the standard complex conjugation on each copy of $\mathbf{CP}^1$, and, taking into 
account that we glue only real points, we obtain a complex conjugation on $\Gamma$.
It is easy to check that we construct a connected reducible curve whose real part possesses the needed number of ovals. In Proposition \ref{prop:Mcurve} we characterize the topological properties of $\Gamma$ and we show the gluing rules and the ovals structure in Figure~\ref{fig:fig4} for the case $N=3$, $M=6$.

\begin{figure}[!tbp]
  \centering
  {\includegraphics[scale=0.15,angle=0,width=0.38\textwidth]{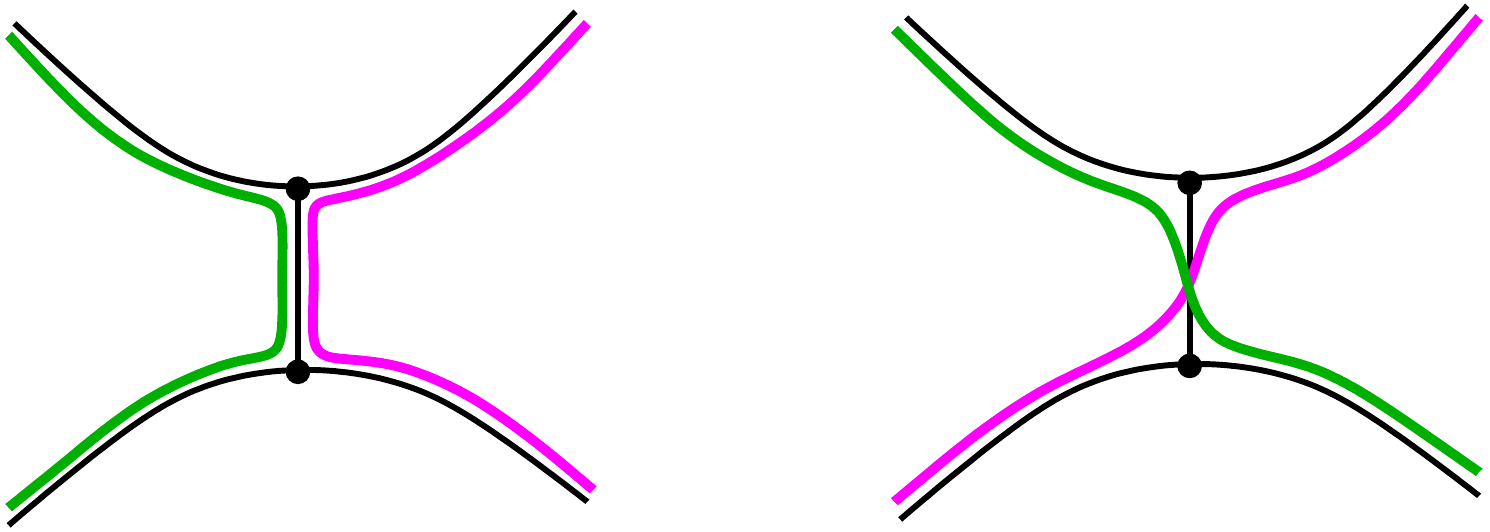}
  }
\caption{\small{\sl Left: Ovals passing through double points in agreement with the planar diagram. 
Right: Oval passing the double points not in agreement with planar diagram.}}
\label{fig:fig4.1}        
\end{figure}

\begin{prop}\label{prop:Mcurve}
Let  $\Gamma=\Gamma_0\sqcup\Gamma_1\sqcup \cdots \sqcup\Gamma_{N}$ be the connected reducible curve of Definition \ref{def:gamma}.
Then the real part of $\Gamma$ which we denote $\Gamma_{\mathbb R}$ possesses $1+(M-N)N$ ovals and each oval is topologically equivalent to a circle. Each double point of $\Gamma$ is a common point to exactly a pair of ovals. Let us denote $\Omega_0$ the oval containing the infinity point $P_0\in \Gamma_0$ (we call this oval infinite), and $\Omega_{r,n}$, $r\in [N]$, $n\in [M-N]$  be the remaining $ (M-N)\times N$ (finite) ovals. 

Then $Q_r \in \Omega_0$, $r\in [N]$,
$\Omega_{r,k}$ are defined by the following properties:
\begin{enumerate}
\item For $r=1$, $n\in [M-N]$,
\[
\begin{array}{ll}
\Omega_{1,n} \cap \Gamma_0 = [\kappa_{N+n-1}, \kappa_{N+n}], &
\Omega_{1,n}\cap \Gamma_1 = [\lambda^{(1)}_{n+1}, \lambda^{(1)}_{n}];\\
\Omega_{1,n}\cap\Gamma_r = \emptyset , \quad r\in [2,N];&
\end{array}
\]
\item For  $r\in [2,N]$
\[
\begin{array}{ll}
\Omega_{r,1} \cap \Gamma_0 = [\kappa_{N-r+1}, \kappa_{N-r+2}], 
&
\Omega_{r,1} \cap \Gamma_{r-1} = [\lambda^{(r-1)}_{1}, \alpha^{(r-1)}_{2}], 
\\
\Omega_{r,1} \cap \Gamma_{r} = [\lambda^{(r)}_{2}, \lambda^{(r)}_{1}],
&
\Omega_{r,1}\cap \Gamma_j =\emptyset, \quad \forall j\in [N]\backslash \{0,r-1,r\};
\end{array}
\]
\item For  $r\in [2,N]$ and $n\in [2,M-N]$,
\[
\begin{array}{ll}
\Omega_{r,n} \cap \Gamma_{r-1} = [\alpha_{n}^{(r-1)}, \alpha_{n+1}^{(r-1)}], &
\Omega_{r,n} \cap \Gamma_r = [\lambda_{n+1}^{(r)}, \lambda_{n}^{(r)}],\\
\Omega_{r,n}\cap \Gamma_j =\emptyset, \quad \forall j\in [N]\backslash \{r-1,r\}; &
\end{array}
\]
\end{enumerate}
\end{prop}

\begin{proof}
The real ovals of $\Gamma$ are defined as the union of the corresponding intervals. The structure of the ovals can be easily determined from the gluing law settled in Definition \ref{def:gamma}.
\end{proof}

\begin{remark}
Our curve can be obtained as a degeneration of a smooth compact Riemann surface of genus $N\times(M-N)$ with an antiholomorphic involution $\sigma$ having  $N\times(M-N)+1$ fixed ovals. Therefore this smooth Riemann surface is an  $\mathtt M$--curve (see \cite{Nat}), and it is natural to treat $\Gamma$ as as reducible rational  $\mathtt M$--curve.
\end{remark}

\begin{remark}
In our representation of the real part of our curve as a planar diagram, we use the following rule: the ovals pass the 
double points in such a way that they do not intersect each other on the diagram (see Fig.~\ref{fig:fig4.1}). Moreover,a perturbation is admissible if respects the antiholomorphic involution  $\sigma$, and if it preserves the number of ovals (see Fig.~\ref{fig:fig4.2}). For example, in the $Gr^{\mbox{\tiny TP}} (1,M)$ case, we assume, that after perturbation the double points split into pairs of \textbf{real} points. The perturbations splitting double points into pairs of complex conjugate point also respect the antiholomorphic involution, but they reduce the number of ovals and are forbidden.
\end{remark}

In Section~\ref{sec:example} we represent rational degenerations of smooth $\mathtt M$--curves corresponding to the soliton data in $Gr^{\mbox{\tiny TP}} (2,4)$ as partial normalizations of reducible algebraic plane nodal curve whose irreducible components are rational (for necessary definitions see \cite{ACG}). We also show that such reducibl plane curv is the rational degeneration of a smooth $\mathtt M$--curve of genus 4. 

\medskip

Let us now describe the analytic properties of the vacuum wave function. Let us introduce the following notation: the restriction of $\Psi$ to the component $\Gamma_r$ is denoted by  $\Psi^{(r)}$; if
$P\in \Gamma_r$ and $\zeta=\zeta(P)$, then
\begin{equation}
\label{eq:notat}
\Psi(\zeta,\vec t)\left.\vphantom{\sum}\right|_{\Gamma_r} = \Psi^{(r)}(\zeta,\vec t).
\end{equation}

\begin{figure}[!tbp]
  \centering
  {\includegraphics[scale=0.5,angle=0,width=0.5\textwidth]{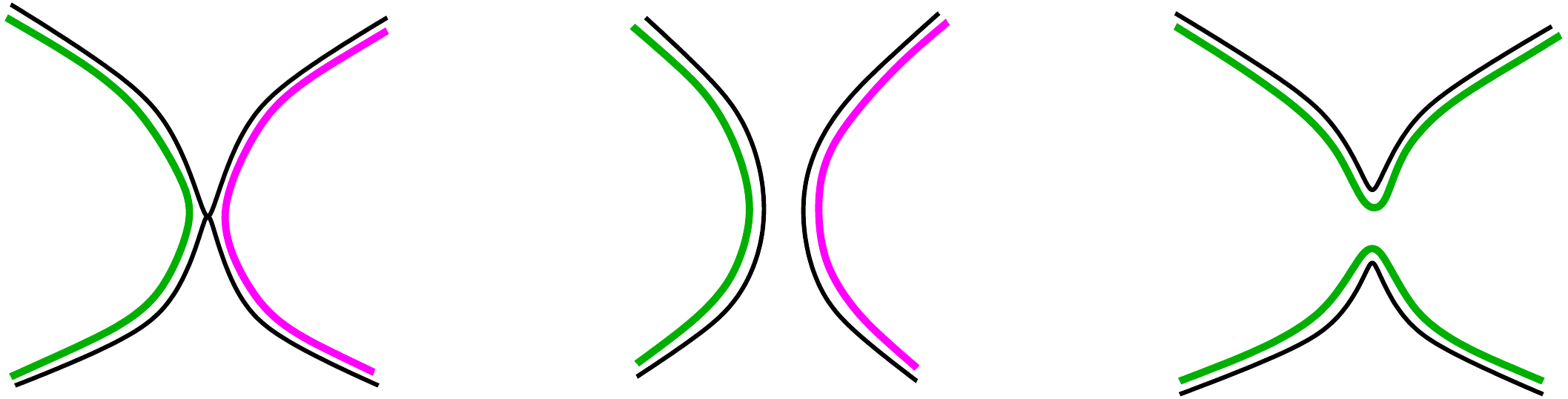}
  }
\caption{\small{\sl Left: A curve with double points and a pair of ovals passing through the double point.  
Middle: An admissible perturbation of the curve: the number of ovals is preserved. Right: This perturbation is not admissible since two ovals merge.}}
\label{fig:fig4.2}        
\end{figure}

\begin{definition}\label{def:psi1}
\textbf{The vacuum wave function.}
Assume that we have the following data
\begin{enumerate}
\item The reducible $M$--curve $\Gamma$ as in Definition \ref{def:gamma}.
\item A collection of $N\times(M-N)$ real divisor points $b^{(r)}_k\in\Gamma_r$, $r\in[N]$, $k\in[M-N]$ such that 
$\lambda_{k+1}^{(r)}<b^{(r)}_k < \lambda_k^{(r)}$.
\end{enumerate}
To such data we associate the vacuum wave function $\Psi(P,\vec t)$, where $P$ is a point of the curve $\Gamma$, and $\vec t$ is the
vector of KP times (we may use the standard assumption that only a finite number of $t_j\ne0$)
with the following analytic properties:
\begin{enumerate}
\item \label{def:psi1:prop1} $\Psi^{(0)}(\zeta,\vec t)$ is the Sato vacuum KP wave function:
\begin{equation}
\Psi^{(0)}(\zeta,\vec t)=e^{\theta}, 
\end{equation}
where
\begin{equation}\label{eq:theta}
\theta \equiv \theta (\zeta, \vec t) = \sum\limits_{i\in[\infty]} \zeta^i t_i, \ \ t_1=x, \ t_2=y, \ t_3=t, \ldots 
\end{equation}
\item\label{def:psi1:prop2} For any fixed collection of KP times $\vec t$ the function $\Psi^{(r)} (\zeta, \vec t)$, $r>0$ is meromorphic 
in $\zeta$ on $\Gamma_r$, and its divisor is exactly $b^{(r)}_1$, \ldots,  $b^{(r)}_{M-N}$.
\item\label{def:psi1:prop3} \textbf{Matching rules at the double points:} at all double points of $\Gamma$ obtained by gluing a pair
of points $P_1\in\Gamma_i$ and $P_2\in\Gamma_j$ the wave function has the same value for all $\vec t$: 
$\Psi(P_1,\vec t)\equiv \Psi(P_2,\vec t)$. More precisely,
\begin{enumerate}
\item For $r=1$ and $l\in [M-N+1]$ the values of $\Psi^{(1)}$ at the point $\lambda_l^{(1)}$ is equal to the value of $\Psi^{(0)}$ at the point $\kappa_{N+l-1}$:
\begin{equation}\label{eq:PP1}
\Psi^{(1)} (\lambda_l^{(1)} , \vec t) = \Psi^{(0)} (\kappa_{N+l-1}, \vec t), {\quad \forall \vec t};
\end{equation}
\item For $r\in[2,N]$ the value of $\Psi^{(r)}$ at the point $\lambda_1^{(r)}$ is equal to the value of $\Psi^{(0)}$ at the point 
$\kappa_{N-r+1}$:
\begin{equation}\label{eq:PPr1}
\Psi^{(r)} (\lambda_1^{(r)} , \vec t) =\Psi^{(0)} (\kappa_{N-r+1}, \vec t), {\quad \forall \vec t};
\end{equation}
\item For $r\in[2,N]$ and $l\in[2,M-N+1]$ the  value of $\Psi^{(r)}$ at the point $\lambda_l^{(r)}$, is equal to the value of $\Psi^{(r-1)}$ at the point $\alpha_l^{(r-1)}$:
\begin{equation}\label{eq:PPrn}
\Psi^{(r)} (\lambda_l^{(r)} , \vec t) = \Psi^{(r-1)} (\alpha_l^{(r-1)}, \vec t), {\quad \forall \vec t}.
\end{equation}
\end{enumerate}
\end{enumerate}
\end{definition}

\begin{prop}
The properties (\ref{def:psi1:prop1})-(\ref{def:psi1:prop3}) uniquely define the function $\Psi(P,\vec t)$. Moreover, for all 
$P\in\Gamma\backslash\{P_0,b^{(r)}_j,r\in[N],j\in{M-n}\}$ the function $\Psi(P,\vec t)$ is a smooth function in the
variables $\vec t$. If, in addition, $P$ is a real point, i.e. $\sigma P=P$, then  $\Psi(P,\vec t)$ is real for real $\vec t$.   
\end{prop}
The proof of this statement is standard and we omit it.

\begin{remark}\label{rem:Psi}
Properties (\ref{def:psi1:prop1}),  (\ref{def:psi1:prop2}) and (\ref{def:psi1:prop3}) in Definition~\ref{def:psi1} 
immediately impose that the vacuum wave function takes the following form:
\begin{equation}\label{eq:psi0}
\Psi^{(r)} (\zeta, \vec t) = \sum_{l=1}^{M-N+1} \mathring B^{(r)}_l \frac{\prod_{j\not = l} (\zeta -\lambda_j^{(r)}) }{ \prod_{k=1}^{M-N} (\zeta-b_k^{(r)})} V^{(r)}_l (\vec t), \ \ r\in[1,N],
\end{equation}
for some real  coefficients $\mathring B^{(r)}_l$, $l\in[M-N+1]$, and $b_k^{(r)}$, $k\in[M-N]$ depending only on $\xi$, where
\begin{equation}\label{eq:V0}
V^{(r)}_l (\vec t) =\left\{ \begin{array}{lll} 
\displaystyle e^{\theta_{N+l-1}(\vec t)},& \quad l\in [M-N+1],& \quad r=1\\[0.5ex]
\displaystyle e^{\theta_{N-r+1}(\vec t)} & \quad l=1,& \quad r\in [2,N]\\[0.5ex]
\Psi^{(r-1)} (\alpha_l^{(r-1)}, \vec t) &\quad l\in [2,M-N+1],&\quad r\in [2,N].
\end{array}
\right. 
\end{equation}
\end{remark}

Let us define $f_r(\vec t)=\Psi(Q_r,\vec t)$, $r\in[N]$.
Then by applying the Darboux transformation we obtain the dressed (unnormalized) KP wave function, associated to the set
$f_r$ of the heat hierarchy solutions. Then its normalization 
\begin{equation}
\tilde\Psi(P,\vec t) =\frac{D^{(N)} \Psi(P,\vec t) }{D^{(N)} \Psi(P,\vec 0)},
\end{equation}
is the Baker-Akhiezer function for the KP multisoliton solution used in the Krichever construction in the 
case of reducible spectral curve. It also satisfy the Dubrovin-Natanzon reality and regularity constraints. 

Our goal is to solve the inverse problem: \textbf{construct the spectral data from a given soliton datum, consisting 
of $M$ ordered phases and a point $[A]$ of the totally positive Grassmannian  $Gr^{\mbox{\tiny TP}} (N,M)$.} Such inverse problem is equivalent to 
construct a curve and a vacuum divisor such that the functions $\Psi(Q_r,\vec t)$ form a basis of heat hierarchy 
solutions generating the Darboux transformation associated with the prescribed point $[A]$.

A solution of this problem is given in the next Section. Assume that we have succeeded in constructing such a curve and 
a vacuum divisor on it. Let us describe the action of the Darboux transformation in the algebraic-geometrical terms.

\begin{theorem}\label{theo:noneff_div}
Assume that we fix the soliton data: a set of $M$ real ordered phases ${\mathcal K} = \{ \kappa_1<\kappa_2<\dots<\kappa_M\} $ and a point of the totally positive Grassmannian $[\hat A]\in Gr^{\mbox{\tiny TP}} (N,M)$.
If, to such soliton data, we may associate a curve $\Gamma$ as in Definition \ref{def:gamma}, a $
N(M-N)$ vacuum divisor $b^{(r)}_j$ as in Definition \ref{def:psi1} and a representative matrix $\hat A$, such that 
\begin{equation}\label{eq:f}
\Psi(Q_r,\vec t)=\sum\limits_{j\in[M]} {\hat A}^{N-r+1}_{j} e^{\theta_j(\vec t)}, \ \ r\in[N],
\end{equation}
then the Darboux transformation $D^{(N)}$ associated to the soliton data $({\mathcal K}, [\hat A])$
generates the following \textbf{shift of divisor}: we add $N$ simple zeroes at the points $Q_r$ and $N$-th order pole at the point
$P_0$. 
\end{theorem}
\begin{proof}
By construction, the Darboux transformed wave function is obtained from the vacuum wave function by applying the $N$-th order 
ordinary differential operator $D^{(N)}$ such that $\Psi(Q_r,\vec t)$, $r\in [N]$, generate its kernel. Therefore
\begin{equation}
\label{eq:div_zeroes}
D^{(N)}\Psi(Q_r,\vec t)\equiv0 \ \ \mbox{for all} \ \ \vec t.
\end{equation}
Property (\ref{eq:div_zeroes}) exactly means that we add a simple zero at each $Q_r$ to the divisor. 
Near the point $P_0$ we have
\begin{equation}
\label{eq:div_pole}
D^{(N)}\Psi_0(\zeta,\vec t)=(\zeta^N+O(\zeta^{N-1}))e^{\theta(\zeta,\vec t)},
\end{equation}
therefore an $N$-order pole is added at $P_0$. Since the operator $D^{(N)}$ does not affect the poles $b^{(r)}_j$, $r\in [N]$, $j\in [M-N]$, of $\Psi$,
the proof is complete.
\end{proof}

\begin{remark}
We remark that the gluing conditions at all double points are preserved by the 
Darboux transformation for all $\vec t$.
\end{remark}

\begin{figure}
\centering
{\includegraphics[scale=0.15,angle=0,width=0.56\textwidth]{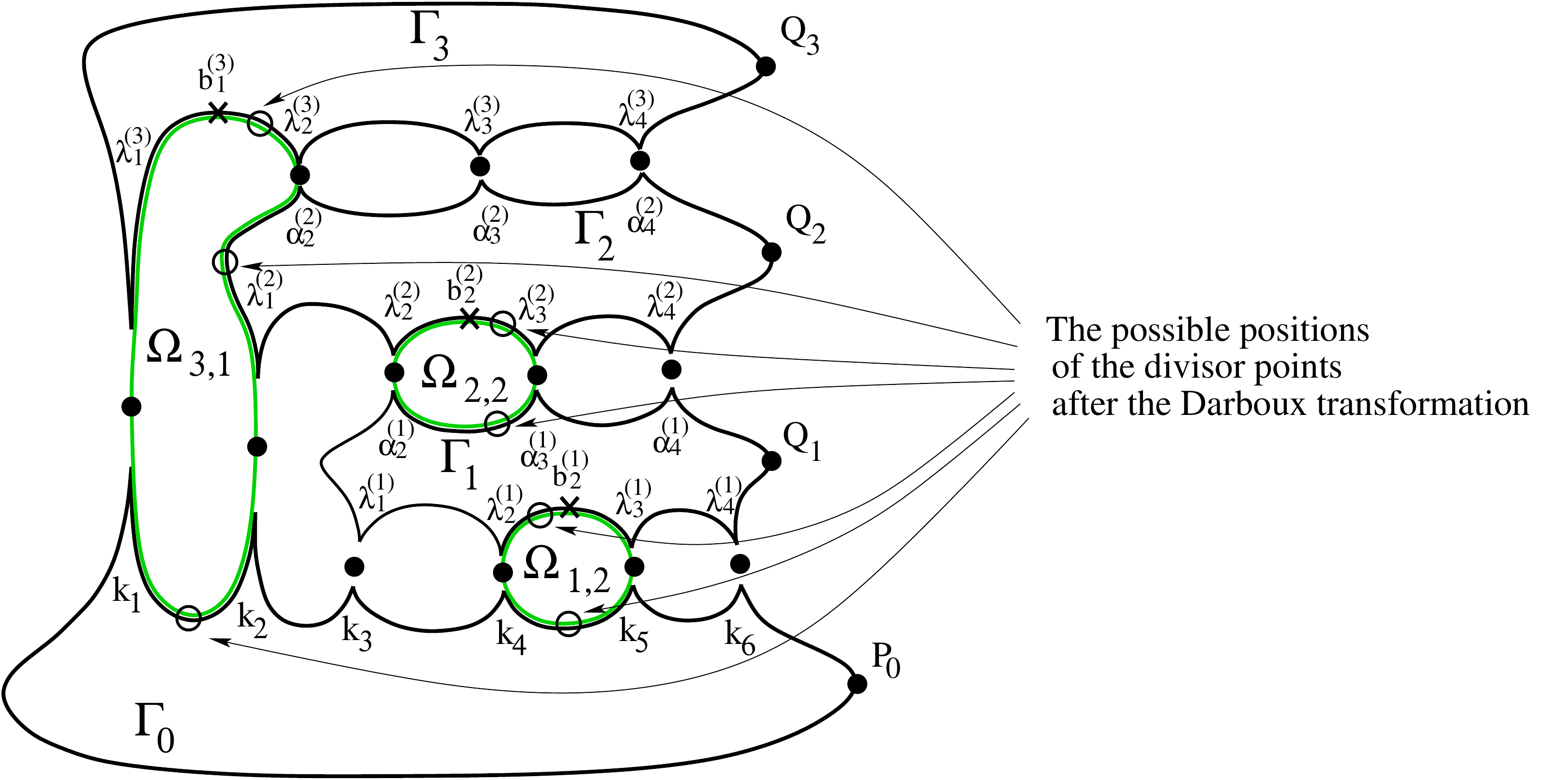}}
\caption{\sl The KP effective divisor may be described as the following action of the Darboux (dressing) $D^{(N)}$ transformation on the effective vacuum divisor: poles are allowed to move inside the finite ovals in such a way that $N$ poles lie in $\Gamma_0$ and $M-N-1$ poles lie in each copy $\Gamma_r$, $r\in [N]$.} 
\label{fig:effdiv}
\end{figure}

Let us now characterize the effective divisor (see also Figure~\ref{fig:effdiv}).

\begin{theorem}\label{theo:eff_div}
Under the hypotheses of Theorem \ref{theo:noneff_div}, the normalized wavefunction
\begin{equation}
\tilde \Psi (P, \vec t) = \frac{D^{(N)} \Psi(P, \vec t)}{D^{(N)} \Psi(P, \vec 0)}, \quad P\in \Gamma \backslash \{ P_0 \}
\end{equation}
has the following analytic properties:
\begin{enumerate}
\item\label{def:psin:divprop} \textbf{Characterization of the effective divisor:} ${\tilde \Psi}(\zeta, \vec t)$ is meromorphic for $P\in \Gamma\backslash \{ P_0\}$, $\zeta =\zeta(P)$, and it possesses an $N(M-N)$ divisor of real simple poles ${\mathcal D} = \{ \gamma^{(0)}_s, \gamma^{(r)}_l, \, s \in [N], r\in [N], l \in [M-N-1] \}$ independent on $\vec t$ and with the following properties:
\begin{enumerate}
\item \label{def:psin:prop1} The restriction of $\tilde\Psi(\zeta,\vec t)$ to $\Gamma_0$ (we use again notations (\ref{eq:notat}) and denote it ${\tilde \Psi}^{(0)}(\zeta,\vec t)$) is the normalized Sato KP wave function:
\begin{equation}
{\tilde \Psi}^{(0)}(\zeta,\vec t)= \frac{D^{(N)} e^{\theta (\zeta, \vec t)}}{D^{(N)} e^{\theta (\zeta, \vec 0)}}, 
\end{equation}
where  $\theta (\zeta, \vec t)$ is as in (\ref{eq:theta}) and $D^{(N)}=\partial_x^N -w_1 (\vec t) \partial_x^{N-1}-\cdots w_N(\vec t) $ is the Darboux transformation associated to the soliton data $({\mathcal K}, [A])$, with ${\mathcal K} = \{\kappa_1< \cdots <\kappa_M\}$. In particular the divisor of poles restricted to $\Gamma_0$ is $\{ \gamma^{(0)}_s,  s \in [N]\}$ and is the zero set of the characteristic polynomial of $D^{(N)}$ at time $\vec t=\vec 0$, {\sl i.e.} 
\[
\zeta^N -w_1 (\vec 0) \zeta^{N-1}-\cdots w_N(\vec 0) =0
\]
\item\label{def:psin:prop2} For any fixed collection of KP times $\vec t$ and $r\in[N]$ the restriction of $\tilde\Psi(\zeta,\vec t)$ to $\Gamma_r$ (we use again notations (\ref{eq:notat}) and denote it ${\tilde \Psi}^{(r)} (\zeta, \vec t)$) is meromorphic in $\zeta$ on $\Gamma_r$, and its divisor consists of $M-N-1$ poles $\gamma^{(r)}_1$, \ldots,  $\gamma^{(r)}_{M-N-1}$;
\item\label{def:psin:finov} There is exactly one divisor pole in each finite oval: $\# \left( {\mathcal D} \cap \Omega_{r,l} \right) = 1$, for all $r\in [N]$, $r\in [M-N]$. Here we use the counting rule, see Definition~\ref{def:counting}; 
\item\label{def:psin:infov} There is no pole in the infinite oval: ${\mathcal D} \cap \Omega_{r,l} = \emptyset$.
\end{enumerate}

\item\label{def:psin:prop3} \textbf{Matching rules at the double points:} at all double points of $\Gamma$ obtained by gluing a pair
of points $P_1\in\Gamma_{r_i}$ and $P_2\in\Gamma_{r_j}$ the wave function has the same value for all $\vec t$: 
${\tilde \Psi}(P_1,\vec t)\equiv {\tilde \Psi}(P_2,\vec t)$. More precisely,
\begin{enumerate}
\item For $r=1$ and $l\in [M-N+1]$ the value of ${\tilde \Psi}^{(1)}$ at the point $\lambda_l^{(1)}$ is equal to the value of ${\tilde \Psi}^{(0)}$ at the point $\kappa_{N+l-1}$:
\begin{equation}\label{eq:PPn1}
{\tilde \Psi}^{(1)} (\lambda_l^{(1)} , \vec t) = {\tilde \Psi}^{(0)} (\kappa_{N+l-1}, \vec t), {\quad \forall \vec t};
\end{equation}
\item For $r\in[2,N]$ the value of ${\tilde \Psi}^{(r)}$ at the point $\lambda_1^{(r)}$ is equal to the value of ${\tilde \Psi}^{(0)}$ at the point 
$\kappa_{N-r+1}$:
\begin{equation}\label{eq:PPrn1}
{\tilde \Psi}^{(r)} (\lambda_1^{(r)} , \vec t) ={\tilde \Psi}^{(0)} (\kappa_{N-r+1}, \vec t), {\quad \forall \vec t};
\end{equation}
\item For $r\in[2,N]$ and $l\in[2,M-N+1]$ the  value of ${\tilde \Psi}^{(r)}$ at the point $\lambda_l^{(r)}$ is equal to the value of ${\tilde \Psi}^{(r-1)}$ at the point $\alpha_l^{(r-1)}$:
\begin{equation}\label{eq:PPrn2}
{\tilde \Psi}^{(r)} (\lambda_l^{(r)} , \vec t) = {\tilde \Psi}^{(r-1)} (\alpha_l^{(r-1)}, \vec t), {\quad \forall \vec t}.
\end{equation}
\end{enumerate}
\end{enumerate}
\end{theorem}

{\begin{definition}\label{def:counting} \textbf{(The counting rule)} For $\vec t$ fixed, we call the divisor $\mathcal{ D} (\vec t)$ of ${\tilde \Psi}(\zeta; \vec t)$ generic, if no points of $\mathcal{ D} (\vec t)$ lie at the double points of $\Gamma$, otherwise we call it non generic.
In the non generic case, we have at least a zero (resp. a pole) of $\tilde \Psi(P,\vec t)$ at a double point $P=X$ belonging to a pair of finite ovals, that is $X\in \Gamma_{r_1} \cap \Gamma_{r_2}$, ($r_1 \not = r_2$). In such case, the function $\tilde \Psi(P,\vec t)$ has simple zeroes (resp. simple poles) at $X$ at both the components $\Gamma_{r_1}$ and $\Gamma_{r_2}$, i.e. we have a collision of 2 divisor points 
$\gamma^{(r_1)}_{k_1}\in\Gamma_{r_1}$ and $\gamma^{(r_2)}_{k_2}\in\Gamma_{r_2}$. 
Then we use the following counting rule: if we have a pair of divisor points at a double point, then one of them is assigned to the first oval and the other is assigned 
to the second oval.

The counting rule has the following interpretation. If we have a divisor point at a double point, we may apply a generic small shift 
of $\vec t$, and we obtain a generic divisor with the property formulated in Theorem~2 (see Figure \ref{fig:countrule2}).
\end{definition}}

\begin{figure}
\centering
{\includegraphics[scale=0.15,angle=0,width=0.44\textwidth]{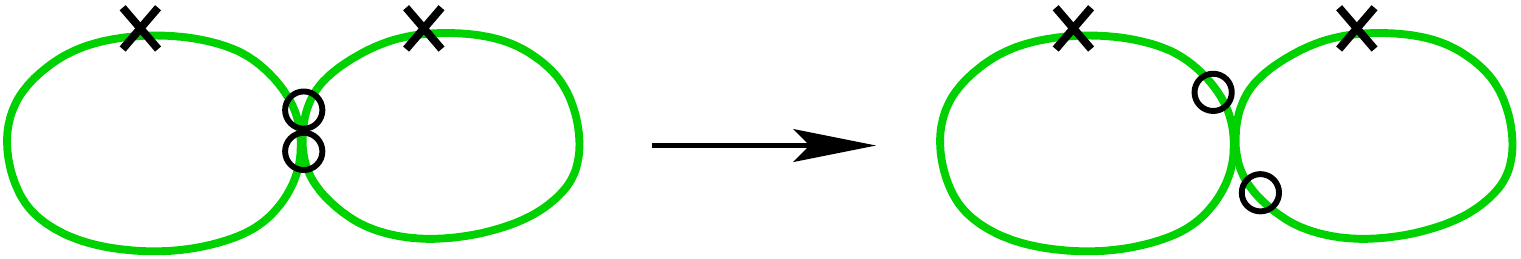}}
\caption{Fig. 8: A small perturbation of a pair of zeroes at a double point.}
\label{fig:countrule2}
\end{figure}

\section{The reducible $\mathtt M$--curve and the effective divisor for soliton data in $Gr^{\mbox{\tiny TP}} (N,M)$}\label{sec:mainres}

In this section, in Theorem \ref{theo:divisor} we associate a spectral curve and an effective divisor to any given soliton datum $[\hat A]\in Gr^{\mbox{\tiny TP}} (N,M)$ represented by the normalized band matrix $\hat A$ as in Definition \ref{def:bandmat}. More precisely, for any fixed sufficiently large positive parameter $\xi$, we construct such algebraic-geometrical data on the spectral curve $\Gamma(\xi)$ as in Definition \ref{def:gamma_xi}.
The proof of Theorem \ref{theo:divisor} follows from Theorem \ref{theo:main}, where we construct a vacuum wavefunction on $\Gamma(\xi)$ satisfying all the properties settled in Definition \ref{def:psi1}. We prove Theorem \ref{theo:main} in the next section: such proof is constructive and is divided into an algebraic and an analytic part.
In particular, we explicitly relate the elementary Darboux transformations to a sequence of maps 
$Gr^{\mbox{\tiny TP}} (1,M-N+1)\leftarrow Gr^{\mbox{\tiny TP}} (2,M-N+2)\leftarrow\ldots\leftarrow 
Gr^{\mbox{\tiny TP}} (r-1,M-N+r-1) \leftarrow Gr^{\mbox{\tiny TP}} (r,M-N+r)\leftarrow\ldots \leftarrow  Gr^{\mbox{\tiny TP}} (N,M)$, by constructing a specific upper triangular representative matrix $\hat A(\xi)$ for the soliton datum in $Gr^{\mbox{\tiny TP}} (N,M)$.
Such matrix $\hat A(\xi)$ rules the exact behavior of the vacuum wavefunction at all marked points and it is obtained as a perturbation of the band matrix $\hat A$  which rules the dominant behavior of the vacuum wavefunction at the same marked points. Finally, by construction, the two matrices represent the same point in $Gr^{\mbox{\tiny TP}} (N,M)$: $[\hat A(\xi)] = [\hat A]$.

\begin{definition}
\label{def:gamma_xi}
\textbf{Specification of positions of the double points for the rational spectral curve.}
Denote by $\Gamma(\xi)$, $\xi>1$, the spectral curve from Definition~\ref{def:gamma} with the following collection of double points:
on each $\Gamma_r$, $r\in[N]$
\begin{equation}\label{eq:lambdas}
\lambda_1^{(r)} =0, \;\; \lambda_l^{(r)} = - \xi^{2(l-2)},\;\; l\in [2,M-N+1],
\end{equation}
\begin{equation}\label{eq:alphas}
\alpha_s^{(r)} = \xi^{2s-5},\quad s\in [2,M-N+1].
\end{equation}
\end{definition}

\begin{remark}
For soliton data in $Gr^{\mbox{\tiny TP}} (1,M)$, $M\ge 4$, we remark that $\Gamma(\xi)$ in Definition \ref{def:gamma_xi} is not equivalent as a complex variety to the rational degeneration of the hyperelliptic curve $\Gamma$ used in section \ref{subseec:hyper}.
\end{remark}

\begin{definition}\label{def:bandmat}
To any given element of the totally positive Grassmannian $Gr^{\mbox{\tiny TP}} (N,M)$ we associate the unique normalized band matrix:
\begin{equation}
\label{eq:our_form2}
\hat A=\left[ \begin{array}{cccccccccccc}  
\hat A_1^1 & \hat A_2^1 &  \hat A_3^1 & \hat A_4^1 &\ldots & \hat A_{M-N+1}^1& 0 & 0 & \ldots  & 0 & 0 & 0\\
0 & \hat A_2^2 &  \hat A_3^2 & \hat A_4^2 & \ldots &   \hat A_{M-N+1}^2 & \hat A_{M-N+2}^2  & 0 & \ldots & 0 & 0 & 0 \\
0 & 0 &  \hat A_3^3 & \hat A_4^3 &  \ldots  & \hat A_{M-N+1}^3 & \hat A_{M-N+2}^3 &  \hat A_{M-N+3}^3  &  \ldots & 0 & 0 & 0  \\
\vdots & & & & & \ddots & & & & & &\vdots \\
0 & 0 &  0 & \ldots & 0 & \hat A_{N-1}^{N-1}   & \ldots & \ldots & \ldots &  \hat A_{M-2}^{N-1} & \hat A_{M-1}^{N-1} & 0  \\
0 & 0 & 0 &  0 & \ldots & 0 &  \hat A_N^N & \ldots & \ldots &  \hat A_{M-2}^{N}  &  \hat A_{M-1}^{N} & \hat A_{M}^{N}  \\
\end{array}\right],
\end{equation}
\begin{equation}
\label{eq:our_form2_1}
\sum\limits_{j=r}^{M-N+r} A_j^r =1, \ \ r\in[N],
\end{equation}
and the following set of heat hierarchy solutions:
\begin{equation}\label{eq:heatsol2}
f^{(N-r+1)}(\vec t) = \sum_{j=r}^{M-N+r} \hat A^r_j e^{\theta_j}, \quad r\in [N].
\end{equation}
It is easy to check that all elements $\hat A^i_j$ in this representation are positive.
\end{definition}
\begin{remark}\label{rem:norm}
In Appendix \ref{sec:totpos} we review a certain number of properties of the band matrix $\hat A$. In particular, there 
are two different natural normalizations for banded matrices: 
\begin{enumerate}
\item Normalization (\ref{eq:our_form2_1}) is convenient in the finite-gap approach, because it corresponds to the normalization of the wavefunction for effective divisors.
\item The normalization $A^i_i=1$, $i\in[N]$, provides local affine coordinates $x_{r,s}$ on $Gr^{\mbox{\tiny TP}} (N,M)$
defined in Proposition~\ref{prop:xcoor} in Appendix~\ref{sec:totpos}.
\end{enumerate}
The connection between these two normalizations is very simple:
\begin{equation}
\label{eq:norms}
\hat A^i_j=\frac{A^i_j}{A^i_i+A^i_{i+1}+\ldots+A^i_{M-N+i}}.
\end{equation}
\end{remark}

\begin{remark}
The positive banded matrices arise in many applications, see, for example, \cite{CDMS}, \cite{BP}. 
\end{remark}

\begin{theorem}\label{theo:main}{\bf (The vacuum wavefunction on $\Gamma(\xi)$)}
 Let the soliton data ${\mathcal K} =\{\kappa_1 < \cdots <\kappa_M\}$ and $[\hat A]\in Gr^{\mbox{\tiny TP}} (N,M)$ be fixed, with $\hat A$ as in Definition \ref{def:bandmat}. Let $f^{(r)}(\vec t)$, $r\in [N]$, be the heat hierarchy solutions associated to $\hat A$ as in (\ref{eq:heatsol2}). Let $\Gamma(\xi)$ denote the curve from Definition~\ref{def:gamma_xi} for $\xi>1$.

Then for any sufficiently big $\xi$ there exists an unique collection of divisor points 
$b^{(r)}_k =b^{(r)}_k (\xi)\in ]\lambda^{(r)}_{k+1},\lambda^{(r)}_k[$ and positive coefficients $\epsilon_{{j}}^{(r)}=\epsilon_{{j}}^{(r)}(\xi)$, 
$r\in[2,N]$, $j\in[r-1]$ such that
\begin{enumerate}
\item \begin{equation}\label{eq:infinity000}
\lim\limits_{\zeta\to\infty} \Psi^{(1)} (\zeta, \vec t) = f^{(1)}(\vec t),
\end{equation}
\item  For $r\in [2,N]$
\begin{equation}\label{eq:infinity00}
\lim\limits_{\zeta\to\infty} \Psi^{(r)} (\zeta, \vec t) = \frac{f^{(r)}(\vec t) + \sum_{j=1}^{r-1} \epsilon_{{j}}^{(r)} f^{(j)}(\vec t)}{ 1+\sum_{{{j=1}}}^{r-1} \epsilon_{{j}}^{(r)}}.
\end{equation}
\item At all double points the matching conditions of Definition~\ref{def:psi1} are fulfilled.
\end{enumerate}
Moreover, if $\xi\rightarrow+\infty$, then all the parameters $\epsilon_j^{(r)}$ tend to 0.
\end{theorem}

\begin{remark}
The proof of Theorem~\ref{theo:main} provides us with a small perturbation of the band matrix $\hat A$, $\hat A (\xi)$
which represents the same point in $Gr^{\mbox{\tiny TP}} (N,M)$ as $\hat A$: $[\hat A(\xi)]=[\hat A]$. Indeed the 
following set of solutions of the heat hierarchy is uniquely associated to (\ref{eq:infinity00})
\begin{equation}\label{eq:heatxi}
 f^{(r)}_{ \xi} (\vec t) \equiv \frac{f^{(r)} (\vec t) + \sum_{j=1}^{r-1} \epsilon^{(r)}_{j} f^{(j)} (\vec t)}{1 +\sum_{j=1}^{r-1} \epsilon^{(r)}_{j}}  = \sum_{j=N-r+1}^M \hat A (\xi)^{N-r+1}_j e^{\theta_j}, \quad\quad r\in [N].
\end{equation}
Then matrix $\hat A (\xi)$ is upper triangular by construction and it is defined starting from $\hat A$ in the following way:
for each $r\in [N]$,
the $(N-r+1)$-th row of $\hat A (\xi)$ is the linear combinations of the last $r$ rows of $\hat A $ with the positive coefficients $\epsilon^{(r)}_k (\xi)$, that is
\[
\hat A(\xi)^{N-r+1}_j = \hat A^{N-r+1}_j+\sum\limits_{k=1}^{r-1} \epsilon^{(r)}_k (\xi) \hat A^{N-k+1}_j , \quad\quad j\in [M].
\]
\end{remark}

The proof of this Theorem is provided in the section \ref{sec:proof}. We end this section with the Theorem which summarizes the properties of the effective divisor of the normalized KP wavefunction $\tilde \Psi$ which is obtained dressing the vacuum wavefunction in Theorem \ref{theo:main}.

\begin{theorem}\label{theo:divisor}{\bf (The effective divisor on $\Gamma_{\xi}$)}
Let $\xi\gg1$, $1\le N<M$ be fixed and let $\Gamma=\Gamma(\xi)$ be as in Definition \ref{def:gamma}. Let the soliton data ${\mathcal K} =\{\kappa_1 < \cdots <\kappa_M\}$ and $[\hat A]\in Gr^{\mbox{\tiny TP}} (N,M)$ be fixed and let $\Psi(\zeta, \vec t)$ be the corresponding vacuum wavefunction on $\Gamma$ satisfying Theorem \ref{theo:main}. Let $D^{(N)}= \partial_x^N - w_1 (\vec t) \partial_x^{N-1}- \cdots- w_N (\vec t)$ be the Darboux (dressing) transformation associated to the soliton data. Then the effective divisor ${\mathcal D}_{\xi}$ of the normalized dressed wavefunction
\[
{\tilde \Psi} (\zeta, \vec t) = \frac{D^{(N)} \Psi(\zeta, \vec t)}{D^{(N)} \Psi(\zeta, \vec 0)}
\]
has degree $N(M-N)$ and satisfies the reality and regularity conditions. 

More precisely let us denote ${\mathcal D}^{(0)} =\{ \gamma^{(0)}_k, \, k \in [N]\}$ the set of solutions to (\ref{eq:Satodiv})
\[
(\gamma^{(0)}_l)^N - w_1 (\vec 0) (\gamma^{(0)}_l)^{N-1} - \cdots - w_{N-1} (\vec 0) \gamma^{(0)}_l-w_N (\vec 0) = 0 ,\quad\quad l\in [N]
\]
and ${\mathcal D}^{(r)}_{\xi} = \{ \gamma^{(r)}_l(\xi), \, l \in [M-N-1]\}$ the set of solutions to
\[
D^{(N)} \Psi^{(r)} (\zeta, \vec 0) =0, \quad r\in [N].
\]
Then
${\mathcal D}_{\xi} = {\mathcal D}^{(0)}\cup {\mathcal D}^{(1)}_{\xi}\cup \cdots \cup {\mathcal D}^{(N)}_{\xi}$ and it has the following properties:
\begin{enumerate}
\item ${\mathcal D}_{\xi}\cap\Gamma_0 = {\mathcal D}^{(0)}$ and all points $ \gamma^{(0)}_k $ lying in $\Gamma_0$ are pairwise different;
\item For any fixed $r\in [N]$, ${\mathcal D}_{\xi}\cap\Gamma_r = {\mathcal D}^{(r)}_{\xi}$ and all points $ \gamma^{(r)}_k (\xi)$ lying in $\Gamma_r$ are pairwise different;
\item $\mathcal{D}_{\xi} \cap \Omega_0 =\emptyset$;
\item $\mathcal{D}_{\xi} \subset \bigcup\limits_{r,j} \Omega_{r,j}$, that is each $\gamma^{(r)}_k(\xi) $ is real and lies in some finite oval;
\item Each finite oval  $\Omega_{r,j}$ contains exactly one point of $\mathcal{ D} $ according to the counting rule 
from Definition~\ref{def:counting}.
\end{enumerate}
\end{theorem}

The proof of Theorem \ref{theo:divisor} follows immediately from Theorems \ref{theo:noneff_div}, \ref{theo:eff_div} and \ref{theo:main}.

\begin{remark}
Replacing the double points  of the reduced curve $\Gamma$ by thin handles, we pass from real bounded regular solitons 
to real bounded regular quasiperiodic finite-gap solutions to the KP equations.
Moreover, in a neighborhood of the point $[A]$ we can work with a fixed curve $\Gamma$ and we have a local bijection between the admissible divisors and an open subset in $Gr^{\mbox{\tiny TP}} (N,M)$, containing $[A]$.
\end{remark}

\section{Proof of Theorem~\ref{theo:main}}\label{sec:proof}

\textbf{The main ideas of the proof.} The proof of the Theorem~\ref{theo:main} may be divided in and algebraic and an analytic part. 

In the first part of the proof we present a recursive algebraic model of the principal cell of $Gr^{\mbox{\tiny TP}} (N,M)$,
providing a birational map of this cell to the positive octant of $\RR^{N(M-N)}$. In fact this construction can be treated 
as a special case of the Whitney theorem \cite{W}. 
We start defining natural projections: $\pi_N:Gr^{\mbox{\tiny TP}} (N,M)\rightarrow Gr^{\mbox{\tiny TP}} (N-1,M-1)$. 
If $\hat A$ is the \textbf{band} matrix of Definition~\ref{def:bandmat} representing the class $[\hat A]\in Gr^{\mbox{\tiny TP}} (N,M)$, then, removing 
the first row and the first column of $\hat A$ we obtain a point in $Gr^{\mbox{\tiny TP}} (N-1,M-1)$, represented by a $(N-1)\times (M-1)$ 
band matrix with the same normalization of the initial $\hat A$. Therefore we have a chain of maps:
\begin{equation}
\label{eq:chain1}
Gr^{\mbox{\tiny TP}} (N,M)\rightarrow \ldots\rightarrow  Gr^{\mbox{\tiny TP}} (r,M-N+r)
\rightarrow  Gr^{\mbox{\tiny TP}} (r-1,M-N+r+1) \rightarrow \ldots\rightarrow Gr^{\mbox{\tiny TP}} (1,M-N+1),
\end{equation}
which we use to settle a recursive construction of the totally positive band matrix $\hat A$ starting from the last row and moving up (Theorem~\ref{lemma:vectors}). The key point in such construction is Lemma~\ref{lemma:PAL}, where, to the $r$--th row of $\hat A$, we associate a minimal set of $M-N$ vectors with non--negative entries  
with the following property: \textbf{the band matrix obtained by adding the $(r-1)$--th row is totally positive if and only if 
the row $r-1$ is a linear combination with  
positive coefficients of these basic vectors and of the corresponding pivot vector.} The recursive procedure of Theorem~\ref{lemma:vectors} for the given band matrix $\hat A$ then follows from the following two properties: 1) at any fixed step, the positive coefficients in Lemma~\ref{lemma:PAL} are uniquely defined by the entries of the $r$--th row of $\hat A$, and 2) the minimal set of vectors associated with the $(r-1)$--th row 
is computed starting from the set of vectors associated with the $r$--th row via an explicit transition matrix. 

The second part of the proof uses the following idea: assume that we have a Riemann sphere with $M-N$ divisor points, $M-N+1$ ``input'' points and $M-N$ ``output'' points. Then the value of the wave function at each output point 
is a linear combination of the values of the wave function at the input points. The coefficients of this transition 
matrix depend on the positions of the points and the divisor. \textbf{We show that the transition matrices arising in 
Theorem~\ref{lemma:vectors} admit arbitrary good approximation by transition matrices associated with Riemann spheres 
by choosing proper marked points.} Therefore we can construct a family of spectral curves with 
divisors such that the asymptotics of the wave function at double points for large $\xi$ is ruled by the algebraic
construction of the first part of the proof. 

Moreover, this approximation is corrected at each step so that we obtain exact matching of the vacuum wavefunction at the double points and control its value at the Darboux points. In particular, we prove that 
$\Psi^{(r)}(\zeta,\vec t)$ is positive if $\zeta>0$ for all $\vec t$. Thus, the zeroes of the vacuum 
wave function are located at the required positions.

\subsection{The algebraic part}

In this section $\RR_{+}^n$ denotes the positive octant in the space $\RR^n$: $x_j>0$, $j\in[n]$. Let us also remark that
the positive octant in $\RR^{n}$ is naturally isomorphic to a collection of $n+1$ positive numbers $\{\hat B_1,
\hat B_2,\ldots,\hat B_{n+1}\}$ such that 
$$
\hat B_1+\hat B_2+\ldots+\hat B_{n+1}=1.
$$
\begin{equation}
\label{eq:birat1}
\hat B_1=\frac{1}{1+x_1+\ldots+x_{n}}, \ \ \hat B_{j+1}=\frac{x_{j}}{1+x_1+\ldots+x_{n}}, \ \ j\in[n].
\end{equation}

Let us start from $N=1$.
In the case of $Gr^{\mbox{\tiny TP}} (1,M)$ we have a very simple characterization: the matrix $\hat A$ is totally positive if
and only if it may be represented as:
\begin{equation}
\label{eq:birat2}
\hat A = \hat B_1\cdot [1,0,0,\ldots,0] + \hat B_2\cdot [0,1,0,\ldots,0]+\ldots+ \hat B_M\cdot [0,0,0,\ldots,1],
\end{equation}
where all $\hat B_i$ are positive. We fix the unique representative in $[\hat A]$ by assuming:
\begin{equation}
\label{eq:birat3}
\hat B_1+\hat B_2+\ldots+\hat B_M=1.
\end{equation}
Formulas (\ref{eq:birat1})-(\ref{eq:birat3}) define a birational map $\RR_{+}^{(M-1)}\rightarrow Gr^{\mbox{\tiny TP}} (1,M)$.

In the next Lemma we express the first row of the band matrix $\hat A$ as a linear combination with positive coefficients of a minimal set of $M-N+1$ vectors whose non--negative entries depend only on the $(N-1)\times (N-1)$ minors of the band matrix obtained eliminating the first row and the first column of $\hat A$.

\begin{lemma}\label{lemma:PAL} (\textbf{Principal Algebraic Lemma}).
Assume that $\hat A$ is an $N\times M$  matrix in banded form such that $\hat A^i_i>0$, { $\hat A^i_j=0$, if and only if} $j<i$ or $j>M-N+i$ {and} $\sum_{j=1}^{M} \hat A^i_j=1$ for all $i\in [N]$. Assume also that $\hat A$ has the following property:
after removing the first row and the first column from $\hat A$  we obtain a matrix ${\hat A}^{(0)}$ such that $[{\hat A}^{(0)}]\in Gr^{\mbox{\tiny TP}}(N-1,M-1)$. 

Then $[\hat A] \in Gr^{\mbox{\tiny TP}}(N,M)$ if and only if there exist $\hat B_n>0$, $n\in {[M-N+1]}$, such that $\sum\limits_{n=1}^{M-N+1} \hat B_n=1$ and the first line of $\hat A$ can be represented in the following form
\begin{equation}
\label{eq:first-line}
[\hat A^1_1,\hat A^1_2,\ldots,{\hat A^1_{M-N+1},0,\dots,0}]=\sum_{n=1}^{M-N+1} \hat B_n {\hat E}^n,
\end{equation}
where 
${\hat E}^n$ denotes the following collection of vectors
\begin{equation}
\label{eq:calA}
{\hat E}^1=[1,0,0,\ldots,0],
\quad\quad
{\hat E}^n  \displaystyle =[0, E^n_2, E^n_3,\ldots, E^n_n, 0,\ldots 0], \quad n\in [2,M-N+1],
\end{equation}
\[
E^n_j = \frac{\Delta_{[j,n+1,\ldots,n+N-2]}}{ \sum_{s=2}^{n} \Delta_{[s, n+1,\dots,n+N-2]}}, \quad j\in [2,n].
\]
Moreover, in such case  
\begin{equation}\label{eq:alpha1}
\hat B_n = \left\{\begin{array}{ll} 
{\hat A}^1_{1},&\quad\quad n=1, \\[0.5ex]
\frac{\Delta_{[n,\dots,n+N-1]}\left( \sum_{s=2}^{n} \Delta_{[s, n+1,\dots,n+N-2]}\right)}{\Delta_{[n,\dots,n+N-2]} \Delta_{[n+1,\dots,n+N-1]} }, &\quad\quad n\in [2,M-N+1].
\end{array}
\right.
\end{equation} 
\end{lemma}

\begin{proof}
To start with, let us present an alternative definition of the vectors 
${\hat E}^n$. Consider 
the matrix $\hat A^{[2,3,\ldots,N]}_{[2,3,\ldots,n+N-2]}$ obtained from $\hat A$ by taking the consequent columns $2$, $3$, \ldots, $n+N-2$ and removing the first row. Using the same Gaussian elimination process as above we can transform $\hat A^{[2,3,\ldots,N]}_{[2,3,\ldots,n+N-2]}$ to the banded form. Denote by ${\mathcal A}^{'n}$ 
the vector obtained from the first row of the banded form of $\hat A^{[2,3,\ldots,N]}_{[2,3,\ldots,n+N-2]}$ by adding one zero at the left-hand side and $M-N-n+2$ zeroes at the right-hand side.

The Gaussian elimination does not affect $(n-1)$-order minors, formed by the last $n-1$ rows and arbitrary columns. Therefore we have the following formulas:
$$
({\mathcal A}^{'n})_j=\left\{\begin{array}{ll}
\frac{\Delta_{[j,n+1,\ldots,n+N-2]}}{\Delta_{[n+1,\ldots,n+N-2]}} & j\in [2,n],\\
0 & j=1 \ \ \mbox{or} \ \ j>n+1,
\end{array}\right.
$$
and ${\hat E}^n= \left({\mathcal A}^{'n}\right)\cdot
\left(\sum_{j=2}^{n} ({\mathcal A}^{'n})_j \right)^{-1}$.

Let us replace the first row of $\hat A$ by ${\hat E}^n$, and denote the new matrix by 
$\tilde A^{(n)}$. Denote by $\tilde A^{(n)}_{[j,j+1,\ldots,j+N-1]}$  the $N\times N$ submatrices 
of $\tilde A^{(n)}$ formed by $N$ consecutive columns starting from the column $j$, $j\ge 2$. 
If $j\le n-1$, the first row of $\tilde A^{(n)}_{[j,j+1,\ldots,j+N-1]}$ is a linear combination of the 
other rows, and $\det(\tilde A^{(n)}_{[j,j+1,\ldots,j+N-1]})=0$. If $j>n$, all elements of the first
row are equal to $0$ and $\det(\tilde A^{(n)}_{[j,j+1,\ldots,j+N-1]})=0$. If $j=n$, the matrix 
$\tilde A^{(n)}_{[j,j+1,\ldots,j+N-1]}$  is {lower}-triangular, therefore we have 
$$
\det(\tilde A^{(n)}_{[j,j+1,\ldots,j+N-1]})=\delta^{n}_j \cdot 
\frac{\Delta_{[n,n+1,\ldots,n+N-2]} \Delta_{[n+1,n+2,\ldots,n+N-1]}}{\left( \sum_{s=2}^{n} \Delta_{[s, n+1,\dots,n+N-2]}\right)}, 
$$
where $ n,j\in [2,M-N+1]$, and
$\delta^i_j$ denotes the standard Kronecker symbol.
As a corollary we immediately obtain that the vectors ${\hat E}^n$, $n\in [0,M-N]$ are 
linearly independent, therefore any vector with zero elements in the positions $M-N+2$, $M-N+3$, \ldots, $M$ can be uniquely represented as a linear combination of these vectors. 

Denote by $\tilde A$ the matrix, obtained from $\hat A$ by replacing the first row with the 
vector, defined by the formula (\ref{eq:first-line}). Then 
$$
\Delta_{[n,n+1,\ldots,n+N-1]}(\tilde A) = \hat B_n \cdot 
\frac{\Delta_{[n,n+1,\ldots,n+N-2]} \Delta_{[n+1,n+2,\ldots,n+N-1]}}{\left( \sum_{s=2}^{n} \Delta_{[s, n+1,\dots,n+N-2]}\right)}, 
$$
where $n>0$,
$$
\Delta_{[1,2,\ldots,N]}(\tilde A) = \Delta_{[1,2,\ldots,N]}.
$$
Therefore we have $\tilde A =\hat A$ if and only if $\hat B_n$ are defined by (\ref{eq:alpha1}) for $n>1$, $\hat B_1=\hat A^1_1$. 

If all minors of the matrix $\hat A$ are strictly positive, then all $\hat B_n>0$. Conversely, if all $\hat B_n>0$, then all $N$-order minors of $\hat A$ formed by consequent columns are strictly positive, and applying Lemma~\ref{lem:l2} we obtain that all $N$-order minors are strictly positive. It completes the proof. 
\end{proof}

\begin{corollary}
\label{cor:alg1}
Assume that the matrix $\hat A$ is the same as in the Principal Algebraic Lemma, and all $N\times N$ minors of $\hat A$ are strictly positive.  Denote by $\breve A$ the matrix, obtained from $\hat A$ by removing the last $s$ columns, $s<M-N$. If we apply to $\breve A$ the same procedure as in the Principal Algebraic Lemma, we obtain 
a collection of vectors $\breve {E}^n$, $n=1,2,\ldots,M-N-s+1$, where 
$\breve {E}^n$ is obtained from ${\hat E}^n$ by removing the last
$s$ zeroes.
\end{corollary}
\begin{proof}
The proof follows directly from the following property of the collection  
${\hat E}^n$: to define the element ${\hat E}^n$ it is sufficient 
to know the first $N+n-2$ columns of the matrix $\hat A$.
\end{proof}

In Lemma~\ref{lemma:PAL} the bases ${\hat E}^n$ are provided by explicit formulas. However, our construction of the vacuum wavefunction for a given spectral 
curve requires a recursive procedure for calculating the vectors ${\hat E}^n$, since such vectors rule the large $\xi$ 
asymptotic behavior of the vacuum wave function at all double points and Darboux points. Next Theorem is the main result of this subsection and it provides such recursive algebraic construction.

In Theorem~\ref{lemma:vectors} to each row of $\hat A$ we associate a unique collection of 
$M-N+1$ normalized vectors $\hat E^{(r)[j]}$, $j\in[M-N+1]$ with non--negative entries, such that this row is expressed as a linear combination 
of these vectors with positive coefficients $\hat B^{(r+1)}_l>0$, $r\in [0,N-1]$, $l\in [M-N+1]$. For each row we 
organize such a collection of vectors as an $(M-N+1)\times M$ matrix $\hat E^{(r)}$. We show that matrices 
$\hat E^{(r)}$ can be computed using a simple recursive procedure.

\begin{theorem}\label{lemma:vectors}(\textbf{The recursive algebraic construction}).
Let $N< M$ and let $\hat A$ be the totally non-negative $N\times M$ band matrix as in
Definition~\ref{def:bandmat} representing a given point $[\hat A]\in Gr^{\mbox{\tiny TP}} (N,M)$.

For any $r\in[0,N-1]$ let us define a normalized  $(M-N+1)\times M$  matrix  $\hat E^{(r)}$ with non-negative entries of 
the form:
\[
\hat E^{(r)}=\left[ \begin{array}{c} \hat E^{(r)[1]} \\  \hat E^{(r)[2]} \\ \vdots \\
\hat E^{(r)[M-N+1]}   
\end{array}\right],
\quad\quad\quad
\hat E^{(r)[j]}=[(\hat E^{(r)[j]})_1,(\hat E^{(r)[j]})_2,\ldots, (\hat E^{(r)[j]})_M],
\]
and normalization $\sum_{s=1}^{M} (\hat E^{(r)[j]})_s=1$, by the following formulas:
\begin{enumerate}
\item For $r=0$ the matrix $\hat E^{(0)}$ is defined by:
$(\hat E^{(0)[l]})_j = \delta_j^{N+l-1}.$
\item For $r\in[1,N-1]$ the matrix $\hat E^{(r)}$ is defined by:
$(\hat E^{(r)[1]})_j=\delta^{N-r}_j$,
and for $n\in[2,M-N+1]$ 
\begin{equation}\label{eq:Ern}
(\hat E^{(r)[n]})_j=\left\{\begin{array}{ll}  0, &\quad\quad\mbox{if} \ \ j \in [N-r] \\[0.5ex] 
\frac{\Delta_{[j;N-r+n,N-r+n+1,\ldots,N+n-2]}}
{\sum_{s=N-r+1}^{N-r+n-1} \Delta_{[s;N-r+n,N-r+n+1,\dots,N+n-2]}} ,	&\quad\quad \mbox{if}\ \  j	\in [N-r+1 , N-r+n-1] \\[0.5ex]
0, &\quad\quad \mbox{if} \ \ j \in[ N-r+n ,M].
\end{array}\right.
\end{equation} 
\end{enumerate}
For each $r\in[N]$ let us define a collection $\hat B^{(r)}_j$, $j\in[M-N+1]$ of positive coefficients as follows:
\begin{enumerate}
\item For $r=1$ the constants $\hat B^{(1)}_j$ are defined by:
$\hat B^{(1)}_j=\hat A^N_{N+j-1}$, $j\in[M-N+1].$
\item For $r\in[2,N]$ the constants $\hat B^{(r)}_j$ are defined by:
\begin{equation}\label{ex:explB}
\hat B^{(r)}_j = \left\{\begin{array}{ll}
\hat A^{N-r+1}_{N-r+1}, &\quad\quad \mbox{if} \ \ j=1 \\ 
\frac{\Delta_{[N-r+j,\dots,N+j-1]} \left( \sum_{s=N-r+2}^{N-r+j} \Delta_{[s;N-r+j+1,N-r+j+2,\dots,N+j-2]}
\right) }{\Delta_{[N-r+j+1,\dots,N+j-1]} \Delta_{[N-r+j,\dots,N+j-2]} }, &\quad\quad \mbox{if} \ \ j\in[2,M-N+1]
\end{array}
\right.
\end{equation}
\end{enumerate}
Then the matrices  $\hat E^{(r)}$ and the coefficients  $\hat B^{(r)}_j$ have the following properties:
\begin{enumerate}
\item The constants $\hat B^{(r)}_l$ are normalized:
$\sum_{s=1}^{M-N+1} \hat B^{(r)}_s=1$.
\item For each $r\in[1,N]$ we have:
\begin{equation}\label{eq:infinity0}
\hat A^{[N-r+1]}=\sum\limits_{j=1}^{M-N+1} \hat B^{(r)}_j\hat E^{(r-1)[j]},
\end{equation}
\item For each $r\in[1,N]$ we have:
{
\begin{equation}\label{eq:glue0}
\hat E^{(r)[2,M-N+1]} = {\mathcal B}^{(r)} \hat E^{(r-1)}, 
\end{equation}
}
where ${\mathcal B}^{(r)}$ denotes the following $(M-N)\times(M-N+1)$ matrix:
\begin{equation}
{\mathcal B}^{(r)} =\left[
\begin{array}{cccccc}
\frac{\hat B^{(r)}_1}{\hat B^{(r)}_1} & 0 & 0 & \ldots & 0 & 0 \\
\frac{\hat B^{(r)}_1}{\hat B^{(r)}_1+\hat B^{(r)}_2} & 
\frac{\hat B^{(r)}_2}{\hat B^{(r)}_1+\hat B^{(r)}_2} & 0 & \ldots & 0 & 0 \\
\frac{\hat B^{(r)}_1}{\hat B^{(r)}_1+\hat B^{(r)}_2+\hat B^{(r)}_3 } & 
\frac{\hat B^{(r)}_2}{\hat B^{(r)}_1+\hat B^{(r)}_2+ \hat B^{(r)}_3 } & 
\frac{\hat B^{(r)}_3}{\hat B^{(r)}_1+\hat B^{(r)}_2+ \hat B^{(r)}_3 } & 
\ldots & 0 & 0 \\
\vdots & \vdots & \vdots & \ddots & 0 & 0 \\
\frac{\hat B^{(r)}_1}{\sum_{l=1}^{M-N}\hat B^{(r)}_l } & 
\frac{\hat B^{(r)}_2}{\sum_{l=1}^{M-N}\hat B^{(r)}_l } & 
\frac{\hat B^{(r)}_3}{\sum_{l=1}^{M-N}\hat B^{(r)}_l } & 
\ldots & 
\frac{\hat B^{(r)}_{M-N}}{\sum_{l=1}^{M-N}\hat B^{(r)}_l } & 
0 \\
\end{array}
\right]
\end{equation}
\end{enumerate}
\end{theorem}

\begin{remark}\label{rem:ultime}
For any fixed $r\in [N]$, the elements of the matrix $\hat E^{(r)}$ and  the coefficients ${\hat B}^{(r)}_j$ are subtraction free rational expressions in the minors of $\hat A$ formed with its last $k$ rows for $k\le r$, so they may be expressed as subtraction free rational expressions in the elements $x_{l,s}$, $l\in [r]$, $s\in [M-N]$ of the FZ-basis following Corollary \ref{cor:ultime}.
As a consequence, all the identities in the above Theorem may be expressed in invariant form, that is they are associated 
to the given point in the Grassmannian and not to the representative matrix $\hat A$. 
\end{remark}

The proof of the Theorem is based on the following two Lemmas:

\begin{lemma}\label{lemma:poles}
Let $\hat A$ be a totally positive $N\times M$ matrix in banded as in Definition \ref{def:bandmat} and
let $s\in [N-1]$ be fixed. Let us define
\[
{\hat B}_j  = \left\{ \begin{array}{ll} \displaystyle {\hat A}^{s}_{s}, & \quad j=1		\\
\frac{\Delta_{[s+j-1,\dots,N+j-1]}\cdot
\left(\sum_{k=1}^{j-1} \Delta_{[s+k;s+j,\dots,N+j-2]}\right)}{\Delta_{[s+j-1,\dots,N+j-2]}\Delta_{[s+j,\dots,N+j-1]}} , &\quad j\in [2,M-N+1]{.}
\end{array}
\right.
\]
Then
\[
\sum\limits_{j=1}^k {\hat B}_j = \frac{\sum_{j=1}^k \Delta_{[s+j-1;s+k,\dots N+k-1]}}{\Delta_{[s+k,\dots,N+k-1]}}, \quad\quad k\in [M-N+1].
\]
In particular
\[
\sum_{j=1}^{M-N+1} {\hat B}_j = \frac{\sum_{j=1}^{M-N+1} \Delta_{[s+j-1;M-N+s+1,\dots M]}}{\Delta_{[M-N+s+1,\dots,M]}} = \sum_{j=s}^{M-N+s} {\hat A}^{s}_{j} \equiv 1.
\]
\end{lemma}

\begin{proof}
The proof is by induction in $l$ using the minors identity 
\[
\begin{array}{l}
\displaystyle
\Delta_{[s+j-1;s+k-1,\dots ,N+k-2 ]}\Delta_{[s+k,\dots ,N+k-1]} +\Delta_{[s+k-1,\dots ,N+k-1]}
\Delta_{[s+j-1; s+k,\dots ,N+k-2 ]} = \\[0.5ex]
\quad \displaystyle \Delta_{[s+j-1; s+k,\dots, N+k-1]} \Delta_{[s+k-1,\dots, N+k-2]}
\end{array}
\]
where $s< s+j-1< s+k-1 < N+k-1$.

Indeed for $l=1$ we just have
\[
{\hat B}_1 \equiv {\hat A}^{s}_{s} = \frac{\Delta_{[s\dots N]}}{\Delta_{[s+1,\dots,N]}}.
\]
Suppose the identity holds for $l=k-1$, then for $l=k$, using the minors identity, we immediately get
\[
\begin{array}{ll}
\displaystyle \sum_{l=1}^k {\hat B}_l &= \displaystyle\frac{\Delta_{[s+k-1,\dots,N+k-1]} \sum_{j=1}^{k-1} \Delta_{[s+j;s+k,\dots,N+k-2]}}{\Delta_{[s+k-1,\dots,N+k-2]}\Delta_{[s+k,\dots,N+k-1]}}+\frac{\sum_{j=1}^{k-1} \Delta_{[s+j-1;s+k-1,\dots, N+k-2]}}{\Delta_{[s+k-1,\dots,N+k-2]}}\\[0.5ex]
&=\displaystyle \frac{\Delta_{[s; s+k,\dots, N+k-1]} +\Delta_{[s+k-1,\dots ,N+k-1]} }{\Delta_{[s+k,\dots, N+k-1]}} +
\displaystyle \frac{\sum_{j=2}^{k-1} \Delta_{[s+j-1;s+k,\dots, N+k-1 ]} }{\Delta_{[s+k,\dots, N+k-1 ]}}
\end{array}
\]
During this calculation we used the following formula: due to the banded structure of $\hat A$ and 
$k\in[2,M-N+1]$:
$\frac{\Delta_{[s; s+k-1,\dots, N+k-2]}}{\Delta_{[s+k-1,\dots, N+k-2]}}
=\frac{\Delta_{[s; s+k,\dots, N+k-1]}}{\Delta_{[s+k,\dots, N+k-1]}}=\hat A^s_s$.
\end{proof}

\begin{lemma}\label{lemma:poles2}
Let $\hat A$ be a totally positive $N\times M$ matrix in banded form as in Definition \ref{def:bandmat}.
Let $r\in[N-1]$, $k\in[2,M-N+1]$ and $j\in [N-r+1,N-r+k-1]$.
Then we have the following identity
\begin{equation}
\label{eq:ident2}
\displaystyle
\sum\limits_{n=1}^{k-1} 
\frac{
\Delta_{[N-r+n,\dots,N+n-1]} \cdot \Delta_{[j;N-r+n+1,\dots,N+n-2]} }{\Delta_{[N-r+n+1,\dots,N+n-1]}\cdot\Delta_{[N-r+n,\dots,N+n-2]}}=
 \frac{\Delta_{[j;N-r+k,\dots,N+k-2]} }{\Delta_{[N-r+k,\dots,N+k-2]}}
\end{equation}
\end{lemma}
\begin{proof}
The proof is again by induction. For $k=2$ we have $j=N-r+1$, and the identity is trivial. Let $k>2$, 
and suppose that for all $2\le k'\le k-2$ the identity has been proven. Then for $j\in[N-r,N-r+k-2]$
we can write
\[
\begin{array}{l}
\sum_{n=1}^{k-1} \frac{
\Delta_{[N-r+n,\dots,N+n-1]} \cdot \Delta_{[j;N-r+n+1,\dots,N+n-2]} }{\Delta_{[N-r+n+1,\dots,N+n-1]}\cdot\Delta_{[N-r+n,\dots,N+n-2]}}=
\\[0.5ex]
\quad\quad  =\sum_{n=1}^{k-2} \frac{
\Delta_{[N-r+n,\dots,N+n-1]} \cdot \Delta_{[j;N-r+n+1,\dots,N+n-2]} }{\Delta_{[N-r+n+1,\dots,N+n-1]}\cdot\Delta_{[N-r+n,\dots,N+n-2]}}+ \frac{
\Delta_{[N-r+k-1,\dots,N+k-2]} \cdot \Delta_{[j;N-r+k,\dots,N+k-3]} }{\Delta_{[N-r+k,\dots,N+k-2]}\cdot\Delta_{[N-r+k-1,\dots,N+k-3]}}=
\\[0.5ex]
\quad\quad  =\frac{\Delta_{[j;N-r+k-1,\dots,N+k-3]} }{\Delta_{[N-r+k-1,\dots,N+k-3]}}+
\frac{
\Delta_{[N-r+k-1,\dots,N+k-2]} \cdot \Delta_{[j;N-r+k,\dots,N+k-3]} }{\Delta_{[N-r+k,\dots,N+k-2]}\cdot\Delta_{[N-r+k-1,\dots,N+k-3]}}=
\\[0.5ex]
\quad\quad  =\frac{\Delta_{[j;N-r+k-1,\dots,N+k-3]} \cdot \Delta_{[N-r+k,\dots,N+k-2]}+
\Delta_{[N-r+k-1,\dots,N+k-2]} \cdot \Delta_{[j;N-r+k,\dots,N+k-3]} }{\Delta_{[N-r+k,\dots,N+k-2]}\cdot\Delta_{[N-r+k-1,\dots,N+k-3]}}=
\end{array}
\]
Applying the minor identity to the numerator, we obtain
\[
\begin{array}{l}
\sum_{n=1}^{k-1} 
\frac{
\Delta_{[N-r+n,\dots,N+n-1]} \cdot \Delta_{[j;N-r+n+1,\dots,N+n-2]} }{\Delta_{[N-r+n+1,\dots,N+n-1]}\cdot\Delta_{[N-r+n,\dots,N+n-2]}}=\frac{\Delta_{[j;N-r+k,\dots,N+k-2]} \cdot \Delta_{[N-r+k-1,\dots,N+k-3]}}{\Delta_{[N-r+k,\dots,N+k-2]}\cdot\Delta_{[N-r+k-1,\dots,N+k-3]}}=
\\[0.5ex]
\quad\quad =\frac{\Delta_{[j;N-r+k,\dots,N+k-2]}}{\Delta_{[N-r+k,\dots,N+k-2]}}.
\end{array}
\]
Assume now that $j=N-r+k-1$. Then we have only one nonzero term in our sum:
$$
\frac{
\Delta_{[N-r+k-1,\dots,N+k-2]} \cdot \Delta_{[N-r+k-1;N-r+k,\dots,N+k-3]}}{\Delta_{[N-r+k,\dots,N+k-2]}\cdot\Delta_{[N-r+k-1,\dots,N+k-3]}}=
\frac{\Delta_{[N-r+k-1;N-r+k,\dots,N+k-2]} }{\Delta_{[N-r+k,\dots,N+k-2]}}.
$$
\end{proof}

Now we are ready to prove Theorem~\ref{lemma:vectors}.

\begin{proof}
The first item follows immediately from Lemma~\ref{lemma:poles}

The second statement is exactly the Principal Algebraic Lemma, applied to the matrix, obtained from $\hat A$
by removing the first $N-r$ rows and the first $N-r$ columns. Applying the formula (\ref{eq:calA}) we 
immediately notice, indeed, that the vector $\hat E^{(r-1)[l]}$ is obtained from ${\hat E}^l$ by adding $N-r$ zeroes 
from the left.

The last statement follows immediately from the Lemma~\ref{lemma:poles2}.
\end{proof}

\medskip

\subsection{The analytic part} We do the proof recursively starting from the last line.

Let us introduce the following notation.

\begin{equation}\label{eq:psismall}
\psi^{(r)}_n (\vec t) =\sum_{j={1}}^{M} \hat E^{(r)[n]}_{j}
e^{\theta_{j}(\vec t)}
 =\frac{\sum_{j= N-r+1 }^{ N-r+n-1} \Delta_{[j;N-r+n,N-r+n+1,\ldots,N+n-2]} e^{\theta_{j} (\vec t)}} 
{\sum_{s=N-r+1}^{N-r+n-1} \Delta_{[s;N-r+n,N-r+n+1,\dots,N+n-2]}},
\end{equation}
where $\hat E^{(r)[n]}_{j}$ are as in Theorem \ref{lemma:vectors}. In Proposition \ref{prop:1}
we compute the vacuum divisor points on the component $\Gamma_1$, and we check that, at leading order in $\xi$, 
the value of the vacuum wave function at each double point coincides with that prescribed by 
Theorem~\ref{lemma:vectors}, i.e. we check that $\psi^{(0)}_n (\vec t)$ and  $\psi^{(1)}_n (\vec t)$ respectively are the leading order
behavior at $\lambda^{(1)}_n$ and $\alpha^{(1)}_n$ when $\xi\gg1$ .

\begin{prop}\label{prop:1}
Let  $\xi\gg1$ be fixed, and let $\Gamma=\Gamma(\xi)$ be as in Definition~\ref{def:gamma_xi}, with
 $\lambda_l^{(1)}$, $l\in [M-N+1]$, $\alpha_s^{(1)}$, $s\in [2,M-N+1]$, as in (\ref{eq:lambdas}) and (\ref{eq:alphas}). Let 
\[
\Psi^{(1)} (\zeta, \vec t) = \sum_{l=1}^{M-N+1} {\mathring B^{(1)}_l} \frac{\prod_{k\not = l}^{M-N+1} (\zeta -\lambda_k^{(1)}) }{\prod_{k=1}^{M-N} (\zeta -b^{(1)}_k )}
e^{\theta_{N+l-1}(\vec t)},
\]
where $\mathring B^{(1)}_l$, $b^{(1)}_k$ are unknown parameters. 
Then the requirements:
\begin{enumerate}
\item\label{eq:caseN}
$\displaystyle \lim_{\zeta\to \infty} \Psi^{(1)} (\zeta, \vec t) = \sum_{j=N}^{M}
{\hat A}^N_j e^{\theta_{j}(\vec t)},$
\item\label{eq:N-N+1}
$\displaystyle
\Psi^{(1)} (\lambda_l^{(1)}, \vec t) = e^{\theta_{N+l-1}(\vec t)}>0$, $l\in[M-N+1]$
\end{enumerate}
uniquely define the coefficients $\mathring B^{(1)}_l = {\hat A}^N_{N+l-1}$, $l\in [M-N+1]$ and the divisor points $b^{(1)}_k= b^{(1)}_k(\xi)\in ]\lambda^{(1)}_{k+1},\lambda^{(1)}_k[$, $k\in[M-N]$. Moreover $\Psi^{(1)}$ has the following properties:
\begin{enumerate}
\item $\Psi^{(1)} (\zeta, \vec t) >0$, $\forall \zeta>0,$ and for all \textbf{real} $\vec t$. In particular, all $\alpha^{(1)}_n>0$,
therefore $\Psi^{(1)} (\alpha_{n}^{(1)}, \vec t)>0$ for all $n$ and for all real $\vec t$. Moreover, we have the following 
estimate: for $\alpha^{(1)}_n$ as in (\ref{eq:alphas}), $n\in [2,M-N+1]$,
\begin{equation}\label{eq:estN}
\Psi^{(1)} (\alpha_{n}^{(1)}, \vec t)  = \displaystyle \sum_{{j=1}}^{M} E^{(1)[{{n}}]}_{j}(\xi)
e^{\theta_{j}(\vec t)} = \left(\psi^{(1)}_n(\vec t)  + \frac{\sum_{j=N+{{n}}-1}^{M} {\hat A}^{N}_{j}e^{\theta_{j}(\vec t)}
\xi^{-2(j-{{n}}-N+1)-1} }{\sum_{s=N}^{N+{{n}}-2} {\hat A}^{N}_{s}}  \right) (1+O(\xi^{-1})),
\end{equation}
where
$$
\psi^{(1)}_n(\vec t)=\sum_{j=N}^{N+{{n}}-2} {\hat A}^{N}_{j}e^{\theta_{j}(\vec t)},
$$
and for all $n\in [2,M-N+1]$ the coefficients $E^{(1)[{{n}}]}_{j}(\xi)$ defined by (\ref{eq:estN}) 
have the following properties:
\begin{enumerate}
\item $E^{(1)[n]}_{j}(\xi) \equiv 0$,  if  $j\in [N-1]$;
\item $E^{(1)[n]}_{j}(\xi)$ are rational functions in $\xi$ for all $j\in[N,M]$;
\item $\displaystyle \lim_{\xi \to \infty} E^{(1)[n]}_{j}(\xi) =\hat E^{(1)[n]}_{j}$, where $\hat E^{(1)}$ is as in 
Theorem~\ref{lemma:vectors} for $r=1$.
\end{enumerate}
\item The elementary symmetric functions in the $b^{(1)}_k(\xi)$,
\[
\Pi^{(1)}_s (\xi) = \sum_{1<j_1<j_2<\cdots<j_s\le M-N} \left(\prod_{l=1}^s b^{(1)}_{j_l}\right), \quad\quad s\in [M-N]
\] 
are rational functions in $\xi$ with coefficients depending only on $\hat A^N_{N+l}$, $l\in [0,M-N]$; 

\item For $\xi\gg1$, we have the following explicit asymptotic estimates for the divisor points $b^{(1)}_k$ on $\Gamma_1$:
\begin{equation}\label{eq:bk1}
b^{(1)}_k =-\displaystyle \frac{\sum_{l=0}^{k-1} \hat A^N_{N+l}}{\sum_{l=0}^{k} \hat A^N_{N+l}}\xi^{2(k-1)} (1+O(\xi^{-1)})),\quad\quad k\in [M-N];
\end{equation}
\end{enumerate}
\end{prop}

\begin{proof}
Requirement (\ref{eq:caseN}) is clearly equivalent to ${\mathring B^{(1)}_l} = {\hat A}^N_{N+l-1}$, $l\in [M-N+1]$. 
Requirement (\ref{eq:N-N+1}) is equivalent to:
\[
{\hat A}^{N}_{N+l-1} \frac{\prod_{k\not = j}^{M-N+1} (\lambda_j^{(1)} -\lambda_k^{(1)}) }{\prod_{k=1}^{M-N} (\lambda_j^{(1)} -b^{(1)}_k )} = \delta^j_l, \quad\quad j,l\in [M-N+1].
\]
Then, using Lemma \ref{lemma:C} in Appendix~\ref{sec:lemmas} with $c_l = {\hat A}^N_{N+l-1}$, $l\in[M-N+1]$, the coefficients $b^{(1)}_k$, $k\in [M-N]$ are uniquely defined and satisfy the required asymptotics (\ref{eq:bk1}).
If $\zeta>0$, then $\Psi^{(1)}(\zeta, \vec t)$ is a finite sum of  positive terms.
Finally the asymptotic behavior of $\Psi^{(1)} (\alpha_{n}^{(1)}, \vec t)$, $n\in [2,M-N+1]$,  easily follows 
from Lemma \ref{lemma:C}.
\end{proof}

Let us now present the construction and the properties of the vacuum wavefunction on each $\Gamma_r$, for $r\in[2,N]$.

\begin{theorem}\label{theo:main0}
Let $\xi>1$ fixed and sufficiently big, $\lambda_l^{(r)}$, $l\in [M-N+1]$, $\alpha_s^{(r)}$, $s\in [2,M-N+1]$, $r\in [N]$, as in (\ref{eq:lambdas}) and (\ref{eq:alphas}), $f^{(i)} (\vec t)$, $i\in [N]$ as in (\ref{eq:heatsol2}) and let $\Psi^{(1)} (\zeta, \vec t)$ satisfy Proposition \ref{prop:1}. 

For $r\in [2,N]$ let $\Psi^{(r)} (\zeta, \vec t)$ as in (\ref{eq:psi0}) with $V^{(r)}_l (\vec t)$ as in (\ref{eq:V0}), $l\in [M-N+1]$. Then for  $r\in[2,N]$  properties
\begin{equation}\label{eq:infinity002}
\lim_{\zeta\to\infty} \Psi^{(r)} (\zeta, \vec t) = \frac{ f^{(r)}(\vec t) + \sum_{j=1}^{r-1} \epsilon_{j}^{(r)} f^{(j)}(\vec t)}{ 1+\sum_{j=1}^{r-1} \epsilon_{{j}}^{(r)}},
\end{equation}
\begin{equation}\label{eq:comp0}
\displaystyle \Psi^{(r)} (\lambda_{{n}}^{(r)}, \vec t) =\left\{ \begin{array}{ll} \displaystyle 
\displaystyle e^{\theta_{N-r+1}(\vec t)}, & \quad {{n}}=1\\[0.5ex]
\displaystyle \displaystyle \Psi^{(r-1)}  (\alpha_{{n}}^{(r-1)}, \vec t),
&\quad {{n}}\in [2,M-N+1],
\end{array}
\right.
\end{equation}
uniquely define the coefficients $\mathring B^{(r)}_l\equiv\displaystyle \frac{B^{(r)}_l}{ 1+\sum_{k=i}^{r-1} \epsilon_k^{(r)}}$, the parameters $\epsilon_j^{(r)}$ and the divisor points $b^{(r)}_k\in ]\lambda^{(r)}_{k+1},\lambda^{(r)}_k[$. Moreover 
\begin{enumerate}
\item $B^{(r)}_l= B^{(r)}_l(\xi)$, $\epsilon_j^{(r)}=\epsilon_j^{(r)}(\xi)$ and the elementary symmetric functions in the divisor points $b^{(r)}_k=b^{(r)}_k(\xi)$,
\[
\Pi^{(r)}_s (\xi) = \sum_{1<j_1<j_2<\cdots<j_s\le M-N} \left(\prod_{l=1}^s b^{(r)}_{j_l}\right), \quad\quad s\in [M-N]
\] 
are all rational function in $\xi$ with coefficients depending only on $x_{l,s}$, $l\in [r]$, $s\in [M-N]$;
\item For $\xi\gg1$, $B^{(r)}_{{l}}>0$, $\epsilon^{(r)}_{j}>0$ and:
\begin{equation}\label{eq:Bs}
B^{(r)}_{{l}} = \left\{\begin{array}{ll} 
{\hat A}^{N-r+1}_{N-r+1}, & {l}=1, \\[0.5ex]
{\hat B}^{(r)}_{l} \cdot(1 + O(\xi^{-1})), & l\in [2,M-N+1],
\end{array}
\right.
\end{equation}
where ${\hat B}^{(r)}_{l}$ are defined by (\ref{ex:explB}); 
\begin{equation}\label{eq:eps}
\epsilon^{(r)}_{j} = \eta^{(r)}_{j} \xi^{-{{j}}} (1 + O(\xi^{-1} )), \quad \quad j\in [1,r-1],
\end{equation}
where the  positive constants $\eta^{(r)}_{j}$ are as in (\ref{eq:epsilon}) and may be explicitly computed using 
(\ref{eq:psialpha2}) in Lemma \ref{lemma:beps} and (\ref{eq:sigmarec}) in Lemma \ref{lemma:psiasym};
\item For $\xi\gg1$, the poles have the following asymptotics:
\begin{equation}\label{eq:bs}
b^{(r)}_k = -\xi^{2(k-1)} \left( \frac{\sum_{l=1}^k \hat B^{(r)}_l}{\sum_{l=1}^{k+1} \hat B^{(r)}_l} \right)(1+O(\xi^{-1})), \quad\quad k\in[M-N];
\end{equation}
\item $\Psi^{(r)} (\zeta, \vec t)>0$, for all $\zeta >0$ and for all $\vec t$, $r\in [2,N]$. In particular, 
all $\alpha^{(r)}_n>0$, therefore $\Psi^{(r)} (\alpha_{n}^{(r)}, \vec t)>0$ for all $n$ and for all $\vec t$. 
Moreover, for any $n\in [2,M-N+1]$ and for all $\vec t$, 
\begin{equation}\label{eq:psialpha}
\begin{array}{l}   \Psi^{(r)} (\alpha_n^{(r)} ,\vec t) = \sum_{j={1}}^{M} E^{(r)[n]}_{j}(\xi)
e^{\theta_{j}(\vec t)} =\left\{\psi^{(r)}_n(\vec t)+ \frac{\sum_{j=N-r+n}^{N+n-2} \sigma^{(r)}_{n,j}e^{\theta_{j}} \xi^{-j+N-r+n-1}}
{\sum_{s=N-r+1}^{N-r+n-1} \Delta_{[s;N-r+n,N-r+n+1,\dots,N+n-2]}}+ \right.
\\[1.5ex]
+\left.\frac{\sum_{j=N+n-1}^{M} \Delta_{[N-r+n,N-r+n+1,\ldots,N+n-2;j]} e^{\theta_{j}} \xi^{-r-2(j-N-n+1)}}
{\sum_{s=N-r+1}^{N-r+n-1} \Delta_{[s;N-r+n,N-r+n+1,\dots,N+n-2]}}
  \right\}\times {\displaystyle\left( 1 + O(\xi^{-1})\right)},
\end{array}
\end{equation}
where the constants $\sigma^{(r)}_{n,j}>0$ are recursively computed using Lemmas \ref{lemma:psiasym} and \ref{lemma:sigma} and depend only on the affine coordinates $x_{l,k}$ on $Gr^{\mbox{\tiny TNN}} (N,M)$, $l\in [r]$, $s\in [M-N]$, defined 
in Proposition~\ref{prop:xcoor} in Appendix~\ref{sec:totpos}. Moreover, for all $n\in [2,M-N+1]$:
\begin{enumerate}
\item $E^{(r)[n]}_{j}(\xi) \equiv 0$ if  $j\in [N-r]$;
\item $E^{(r)[n]}_{j}(\xi)$ are rational functions in $\xi$, for $j\in[N-r+1,M]$
\item $\displaystyle \lim\limits_{\xi \to \infty} E^{(r)[n]}_{j}(\xi) =\hat E^{(r)[n]}_{j}$, where $\hat E^{(r)}$ is as in 
Theorem~\ref{lemma:vectors}.
\end{enumerate}
\end{enumerate}
\end{theorem}

\begin{proof} {\it of Theorem \ref{theo:main0}.}
The proof goes through several steps and by induction. 

Step 1: Direct proof for $r=2$. Let $\Psi^{(1)}(\zeta, \vec t)$ be as in Proposition \ref{prop:1}
and let 
\[
\begin{array}{ll}
\displaystyle \Psi^{(2)} (\zeta , \vec t) &= \frac{\mathring B^{(2)}_{1}\prod_{k\not = 1}^{M-N+1} (\zeta - \lambda_k^{(2)})}{\prod_{k = 1}^{M-N} (\zeta - b^{(2)}_k)} e^{\theta_{{N-1}} (\vec t)} +
\sum\limits_{j=2}^{M-N+1} 
\frac{\mathring B^{(2)}_j\prod_{k\not = j}^{M-N+1} (\zeta - \lambda_k^{(2)})}{\prod_{k = 1}^{M-N}(\zeta - b^{({2})}_k)} \Psi^{(1)} (\alpha_j^{(1)} , \vec t).
\end{array}
\]
with all of the $\displaystyle \mathring B^{(2)}_l = \frac{B^{(2)}_l}{1+ \epsilon^{(2)}_1}$, for $l\in [M-N+1]$, $\epsilon^{(2)}_1$ and $b^{(2)}_k$ to be determined. Then:

a) We have 
\begin{equation}\label{eq:limN-1}
\lim_{\zeta \to \infty} \Psi^{(2)} (\zeta , \vec t) =\frac{ f^{(2)} (\vec t) +\epsilon^{(2)}_1 f^{(1)} (\vec t)}{1+ \epsilon^{(2)}_1}
\end{equation}
if and only if $B^{(2)}_1 ={\hat A}^{N-1}_{N-1}$ and the remaining coefficients satisfy the linear system
\[
\sum_{l=2}^{M-N+1} B^{(2)}_l E^{(1)[l]}_{s}(\xi)  = {\hat A}^N_{s} \epsilon^{(2)}_{{1}} + {\hat A}^{N-1}_{s}, \quad \quad s \in [N, M],
\]
where the coefficients $E^{(1)[l]}_{s}(\xi)$ are the rational functions in $\xi$ defined in (\ref{eq:estN}).
Then it is straightforward to check both the uniqueness of the solution of the above system for almost all $\xi>1$ and the regularity properties in $\xi$ of the coefficients $B^{(2)}_l$ and $\epsilon^{(2)}_{1}$. Moreover, using Lemma \ref{lemma:beps} for $r=2$ with 
$\sigma^{(1)}_{n,l} = {\hat A}^N_{N+l-1}$, $l\in [1,M-N+1]$, $n\in [2,M-N+1]$, (\ref{eq:estN}), the Principal Algebraic Lemma \ref{lemma:PAL} and Lemma \ref{lemma:sigma}, we immediately get the required estimates for the coefficients ($l\in[2,M-N+1]$),
\[
B^{({2})}_{{l}} = \frac{\Delta_{[N+l-2,N+l-1]}\left( \sum_{j=1}^{l-1} {\hat A}^N_{N+j-1}\right)}{{\hat A}^N_{N+l-2}{\hat A}^N_{N+l-2} } \left( 1 + O ( \xi^{-1}) \right),
\quad\quad
\epsilon^{(2)}_{{1}} = \frac{{\hat A}^{N-1}_{M-1} }{\xi{\hat A}^N_{M-1} } \left( 1 + O ( \xi^{-1}) \right).
\]
Thanks to the positivity property of the matrix ${\hat A}$ we immediately have $B^{(2)}_j >0$, $j\in [M-N+1]$ and $\epsilon^{(2)}_{1}>0$ for all $\xi\gg1$. Moreover, inserting $\vec t = \vec 0$ in (\ref{eq:limN-1}), we conclude
$\sum_{j=1}^{M-N+1} {\mathring B}^{(2)}_j (\xi)\equiv 1.$ Finally all of the quantities my be expressed in invariant form using the FZ-basis $x_{l,k}$, $l=1,2$, $k\in [M-N]$.

\smallskip

b) The set of conditions 
\[
\Psi^{(2)} (\lambda_1^{(2)},\vec t ) = e^{\theta_{N-1} (\vec t)}, \quad  \Psi^{(2)} (\lambda_{n}^{(2)},\vec t ) =\Psi^{(1)} (\alpha_{n}^{(1)},\vec t ), \; n\in [2,M-N+1],
\]
is equivalent to 
\[
{\mathring B}^{(N-1)}_{n} \prod\limits_{k\not = n}^{M-N+1} (\lambda_{n}-\lambda_k) = \prod\limits_{k=1}^{M-N} (\lambda_{n}-b^{(2)}_{k}) , \quad\quad n\in [1,M-N+1].
\]
Using Lemmas \ref{lemma:poles} and \ref{lemma:C} in the case $r=2$, with $c_j ={\mathring B}^{(2)}_j$, we immediately obtain the required conditions for the regularity in $\xi$, the position and the leading order expansion of the poles ($k\in [M-N]$),
\[
b^{(2)}_{{k}} = - {\xi^{2(k-1)}} \frac{{\hat A}^N_{N+k} \left(\sum_{l=1}^{k} \Delta_{[N-2+l,N-1+k]}\right)}{{\hat A}^N_{N+{{k}}-1}\left(\sum_{l=1}^{k+1} \Delta_{[N-2+l,N+k]}\right)} (1 + O(\xi^{-1})).
\]
c) If $\zeta>0$, then $\Psi^{(2)}(\zeta, \vec t)$ is a finite sum of  positive terms.
Finally, using Lemmas \ref{lemma:psiasym}, \ref{lemma:sigma} and Theorem \ref{lemma:vectors}, we get that the coefficients $E^{(2)[n]}_{j}(\xi)$ have the required regularity properties in $\xi$ and that they satisfy the asymptotic expansion (\ref{eq:psialpha})
\[
\begin{array}{rl} 
\Psi^{({2})} (\alpha_{{n}}^{(2)}, \vec t) &=\sum\limits_{{j}=1}^{{M}} E^{({2}){{[n]}}}_{{{j}}}(\xi)
e^{\theta_{{j}}(\vec t)}
= \left( \sum_{j=N-1}^{N+n-3} \Delta_{[j,N+n-2]} e^{\theta_{{j}}}+\frac{({\hat A}^N_{N+n-2})^2 \Delta_{[N+n-3,N+n-1]}}{{\hat A}^N_{N+{n-3}} {\hat A}^N_{N+n-1}}  \frac{ e^{\theta_{j}}}{\xi}  +\right.
\\[0.5ex]
& \left. +\sum_{j=N+n-1}^{M} \frac{\Delta_{[N+n-2,j]}e^{\theta_{j}}}{  \xi^{2(j-N-n+2)}}  \right)\times\frac{(1 +O (\xi^{-1}))}{\sum_{k =N-1}^{N+n-3} \Delta_{[k,N+n-2]}} ;
\end{array}
\]

\smallskip

Step 2: The induction procedure. Let $r\in [3,N]$ and let us suppose we proved the Theorem for
$i=2,\dots,r-1$, and let us prove it for $i=r$. Let us denote $\mathring B^{(r)}_l = \frac{B^{(r)}_l}{1+ \sum_{j=1}^{r-1} \epsilon^{(r)}_{j}}$, for $l\in [M-N+1]$ and define 
\[
\displaystyle \Psi^{(r)} (\zeta , \vec t) = \frac{ B^{(r)}_{1}e^{\theta_{{N-r+1}} (\vec t)}\prod_{j\not = 1}^{M-N+1} (\zeta - \lambda_{j}^{(r)})}{(1 + \sum_{j=1}^{r-1} \epsilon^{(r)}_{j}) \prod_{k = 1}^{M-N}(\zeta - b^{(r)}_{k})} + \sum_{n=2}^{M-N+1} 
\frac{B^{(r)}_{n}\Psi^{(r-1)} (\alpha_{n}^{(r-1)} , \vec t)\prod_{j\not = n}^{M-N+1} (\zeta - \lambda_{j}^{(r)})}{(1 + \sum_{j=1}^{r-1} \epsilon^{(r)}_{j}) \prod_{k = 1}^{M-N}(\zeta - b^{(r)}_{k})} ,
\]
with coefficients $B^{(r)}_{n}$, $\epsilon^{(r)}_{j}$ and
$b^{(r)}_{k}$ to be determined.

a) We have
\begin{equation}\label{eq:limN-l}
\lim_{\zeta \to \infty} \Psi^{(r)} (\zeta , \vec t) =\frac{ f^{(r)} (\vec t) +\sum_{j=1}^{{r-1}} \epsilon^{(r)}_{j} f^{(j)} (\vec t)}{1+ \sum_{j=1}^{{r-1}} \epsilon^{(r)}_{j}}
\end{equation}
if and only if ${B^{(r)}_1 ={\hat A}^{N-r+1}_{N-r+1}}$ and the remaining coefficients satisfy the linear system
\[
{\sum\limits_{l=2}^{M-N+1} B^{(r)}_l E^{(r-1)[l]}_{s}(\xi)  =\sum\limits_{j=1}^{r-1} {\hat A}^{N-r+j+1}_{s} \epsilon^{(r)}_{j} + {\hat A}^{N-r+1}_{s}, \quad  s \in [N-r+2, M],}
\]
where the coefficients $E^{(r-1)[l]}_{s}(\xi)$ are rational functions in $\xi$. Due to the compatibility of the above linear system for almost all $\xi>1$ and the regularity properties of the coefficients, there immediately follow  both the uniqueness for almost all $\xi>1$ and the regularity properties in $\xi$ for $B^{(r)}_l$ and $\epsilon^{(r)}_{j}$. Again, using Lemmas \ref{lemma:PAL}, \ref{lemma:beps} and  \ref{lemma:sigma}, we immediately get the required asymptotic estimates for the coefficients $B^{(r)}_l$ ($l\in[2,M-N+1]$) and $\epsilon^{(r)}_j$, ($j\in [r-1]$) as in (\ref{eq:Bs}) and (\ref{eq:eps}), respectively, when $\xi\gg1$. 
In particular, substituting $\vec t =\vec 0$ in (\ref{eq:limN-l}), we have  $\sum_{{l}=1}^{M-N+1}	{\mathring B}^{(r)}_{{l}}=1$.

b) The set of conditions 
\[
\Psi^{({r})} (\lambda_1^{(r)},\vec t ) = e^{\theta_{N-r+1} (\vec t)}, \quad  \Psi^{(r)} (\lambda_{n}^{(r)},\vec t ) =\Psi^{({r-1})} (\alpha_{n}^{(r-1)},\vec t ), \; {n}\in [2,M-N+1],
\]
is equivalent to 
\[
{\mathring B}^{(r)}_{n} \prod\limits_{j\not = n}^{M-N+1} (\lambda_{n}^{(r)}-\lambda_j^{(r)}) = \prod\limits_{k=1}^{M-N} (\lambda_{n}^{(r)}-b^{(r)}_{k}), \quad\quad n\in [1,M-N+1].
\]
Again, using Lemma \ref{lemma:C} in the case $c_j ={\mathring B}^{(r)}_j$, and Lemma \ref{lemma:poles}, we immediately obtain the required estimates for the regularity and for position of the poles $b^{(r)}_{k}$ ($k\in [M-N]$) as well as the leading order estimates for $(\xi\gg1$) as in (\ref{eq:bs}).

c) Finally, from Lemmas~\ref{lemma:psiasym}, \ref{lemma:sigma} and Theorem \ref{lemma:vectors}, we get the required estimates for the regularity, the sign and the leading order term expansion of $\Psi^{(r)} (\alpha_{n}, \vec t)$, for  $n\in [2,M-N+1]$.
\end{proof}

In Theorem~\ref{theo:main0} we provide the explicit construction of the normalized vacuum wave function, with 
the properties required in Theorem~\ref{theo:main}. This remark completes the proof of Theorem~\ref{theo:main}.

\medskip
\section{Analytic properties of the effective divisor}\label{sec:divest}

\subsection{Zero divisor for the vacuum wave function}
In Theorem~\ref{theo:main0} we provide an explicit estimate of the position of the vacuum pole divisor (\ref{eq:bs}). 
Let us now present some estimates for the position of vacuum zero divisor for sufficiently small times. We do not use
directly these formulas in our paper, but they may be useful in future inverstigations.

\begin{corollary}\label{cor:zerodiv} 
Let $\xi\gg1$ be fixed, $\Gamma =\Gamma(\xi)$ be as in Definition \ref{def:gamma} and $\Psi(\zeta, \vec t)$  be the vacuum wavefunction of Theorem \ref{theo:main}. 
Then in each finite oval $\Omega_{r,n}$, ($r\in [N]$, $n\in [M-N]$), $\Psi(\zeta, \vec t)$ possesses exactly  one simple pole $b^{(r)}_n (\xi)$, whose position is independent of $\vec t$, and exactly one simple zero $\chi^{(r)}_n (\xi;\vec t)$. In particular
\begin{enumerate}
\item $b^{(r)}_n (\xi)\in ]\lambda^{(r)}_{n+1}, \lambda^{(r)}_{n} [ \subset\Gamma_r\cap \Omega_{r,n}$;
\item $\chi^{(r)}_n (\xi;\vec 0) =b^{(r)}_n (\xi)$;
\item $\chi^{(r)}_n (\xi;\vec t) \in ]\lambda^{(r)}_{n+1}, \lambda^{(r)}_{n} [\subset\Gamma_r\cap \Omega_{r,n}$,  for all $\vec t$;  
\item
Assume that only a finite number of times is different from zero: $t_j=0$ for $j>j_0$, and all times $t_1$, $t_2$, 
\ldots, $t_{j_0}$ lie in a compact domain $K_0$. Then we have the following asymptotic expansion in $\xi\gg1$:
\begin{equation}\label{eq:zeroes0}
\resizebox{\textwidth}{!}{$
 \chi^{(r)}_n (\xi;\vec t)  = - \,\frac{\sum\limits_{l=1}^n {\hat B}^{(r)}_l V^{(r)}_l (\vec t)}{\sum\limits_{l=1}^{n+1} {\hat B}^{(r)}_l V^{(r)}_l (\vec t)} \xi^{2(j-1)} (1+O(\xi^{-1}))
 = -\, \frac{\sum\limits_{j=N-r+1}^{N-r+n} \frac{\Delta_{[j;N-r+n+1,\dots,N+n-1]}}{\Delta_{[N-r+n+1,\dots,N+n-1]} } e^{\theta_{j}}}{\sum\limits_{j=N+r-1}^{N-r+n+1}\frac{ \Delta_{[j;N-r+n+2,\dots,N+n]}}{\Delta_{[N-r+n+2,\dots,N+n]}} e^{\theta_{j}}} \xi^{2(j-1)}(1+O(\xi^{-1})).$}
\end{equation}
\end{enumerate}
\end{corollary}

\begin{proof}
The number of poles $b^{(r)}_k$ is equal to the number of ovals and their position is computed in Proposition \ref{prop:1} and in Theorem \ref{theo:main0}. Therefore the number of zeroes of $\Psi(\zeta, \vec t)$ is equal to the number of ovals. 
For $\vec t=(0,0,\ldots)$, by definition, the zeroes coincide with the divisor points, and their positions continuously depend on $\vec t$.  A zero could leave a real oval only if it collides with another zero coming from another oval. 
That is impossible since $\Psi^{(r)} (\lambda^{(r)}_j, \vec t)>0$, for all $\vec t$, with
$r\in [N]$, $j\in [M-N+1]$ (see Theorem~\ref{theo:main0}).
It means that for all times $\vec t$ each zero remains in the same open interval 
$]\lambda^{(r)}_n, \lambda^{(r)}_{n+1} [$.

All terms $V^{(r)}_l (\vec t)$ are of order 1 in $\xi$  for $(t_1,\dots, t_{j_0}) \in K_0$.  Let us write the function  $\Psi^{(r)}(\zeta, \vec t)$ as a sum of simple fractions: 
\[
\Psi^{(r)}(\zeta, \vec t)= f_{r,\xi} (\vec t)+\sum\limits_{k=1}^{M-N}\frac{\psi^{(r)}_k (\vec t)}{\zeta-b^{(r)}_k},
\]
where $f_{r,\xi} (\vec t)$ is as in (\ref{eq:heatxi}) and the positions of the poles are given by (\ref{eq:bs}). Therefore 
\[
\psi^{(r)}_k (\vec t)=\xi^{2k-2}\cdot \left(\frac{ \sum_{j=1}^{k} \hat B^{(r)}_j   V^{(r)}_j (\vec t)   - \left(   \sum_{j=1}^{k} \hat B^{(r)}_j  \right)  V^{(r)}_{k+1} (\vec t) }{
\left(\sum_{j=1}^{k+1} \hat B^{(r)}_j \right)^2} \hat B^{(r)}_{k+1}   \right) (1+O(\xi^{-1})),
\]
and inside the interval $[\lambda^{(r)}_{k+1},\lambda^{(r)}_{k}]$ we have 
\begin{equation}
\label{eq:psi_int}
\Psi^{(r)}(\zeta, \vec t)=\xi^{2k-2}\cdot \left(\frac{  \sum_{j=1}^{k+1} \hat B^{(r)}_j   V^{(r)}_j (\vec t) }
{  \sum_{j=1}^{k+1} \hat B^{(r)}_j } -\frac{\psi^{(r)}_k (\vec t)}{\zeta-b^{(r)}_k}, \right) (1+O(\xi^{-1})).
\end{equation}
We complete the proof solving the equation $\Psi^{(r)}(\zeta, \vec t)=0$ and using the approximation (\ref{eq:psi_int}).
\end{proof}

We remark that the condition that each zero of $\Psi(P,\vec t)$ lies in a well-defined  open interval $]\lambda_{j+1}^{(r)}, \lambda_{j}^{(r)}[$ for all $\vec t$, is natural since $\Psi(P,\vec t)$ represents a {vacuum} wave function: no collision is possible for the zero divisor in this case!

\subsection{Zero divisor for the normalized Darboux transformed wave function}

We now provide some estimates on the zero and pole divisors of the normalized wavefunction ${\tilde\Psi} (\zeta, \vec t)$. The pole (effective) divisor has been characterized in Theorem \ref{theo:divisor}. For the zero divisor ${\mathcal D}_{\xi} (	\vec t)$ we adopt the following notation. By construction, for any fixed $\vec t$, its restriction to $\Gamma_0$ is
 ${\mathcal D}^{(0)}(\vec t) =\{ \gamma^{(0)}_k (\vec t), \, k \in [N]\}$ and it coincides with the set of solutions to (\ref{eq:Satodiv})
\[
(\gamma^{(0)}_l)^N (\vec t)- w_1 (\vec t) (\gamma^{(0)}_l)^{N-1} (\vec t)- \cdots - w_{N-1} (\vec t) \gamma^{(0)}_l-w_N (\vec t) = 0 ,\quad\quad l\in [N].
\]
The restriction of the zero divisor to $\Gamma_r$ is  ${\mathcal D}^{(r)}_{\xi} (\vec t)= \{ \gamma^{(r)}_l (\vec t), \, l \in [M-N-1]\}$ and it coincides with the set of solutions to
\[
D^{(k)} \Psi^{(r)} (\zeta, \vec t) =0, \quad r\in [N].
\]

In the following theorem we estimate the position of the zero divisor in the case in which only a finite number of times in $\vec t$ may be different from zero, and moreover they vary in a neighborhood of $\vec 0=(0,\dots0)$. In particular we give the explicit estimate
for the position of the effective divisor $\mathcal{D}_{\xi}$.

\begin{theorem} (\textbf{Estimate of the position of divisor $\mathcal{D}_{\xi}$})\label{theo:t0}
Let $\xi\gg1$ and let $D\Psi^{(r)} (\zeta, \vec t)$, $r\in [N]$, as above. Assume that only a finite number of times may be different from zero: $t_j=0$ for $j>j_0$, and all times $t_1$, $t_2$, 
\ldots, $t_{j_0}$ lie in a compact domain $K_0$ containing the point $(t_1,\dots, t_{j_0}) = (0,\dots, 0)$.
Then for $\xi\gg1$, generically
the following asymptotic expansion holds for the zeroes of $D\Psi^{(r)} (\zeta, \vec t)$, $\vec t \in K_0$:
\begin{equation}\label{eq:zeroesD}
\resizebox{\textwidth}{!}{$
\gamma^{(r)}_n (\vec t) = -\frac{ \sum_{l=1}^n {\hat B}^{(r)}_l DV^{(r)}_l (\vec t)}{ \sum_{l=1}^{n+1} {\hat B}^{(r)}_l DV^{(r)}_l (\vec t)} \xi^{2(n-1)} (1+O(\xi^{-1}))
= -\frac{\sum_{j=N-r+1}^{N-r+n} \frac{ \Delta_{[j;N-r+n+1,\dots,N+n-1]}}{ \Delta_{[N-r+n+1,\dots,N+n-1]} } P^{(0)} (\kappa_j) e^{\theta_j}}{\sum\limits_{j=N+r-1}^{N-r+n+1}\frac{ \Delta_{[j;N-r+n+2,\dots,N+n]}}{ \Delta_{[N-r+n+2,\dots,N+n]}} P^{(0)} (\kappa_j)e^{\theta_j}}\xi^{2(n-1)} (1+O(\xi^{-1})),
$}
\end{equation}
where $P^{(0)} (\kappa_j) = \prod\limits_{l=1}^N ( \kappa_j -\gamma^{(0)} (\vec t))$, $j\in [M]$.
In particular, for $(t_1,\dots,t_{j_0})= (0,\dots,0)\in K_0$ we have the estimate of the effective divisor $\mathcal{D}_{\xi}$.
\end{theorem}

In Corollaries \ref{cor:zeros1} and \ref{cor:zeropos} we control the position of the both the pole and the zero divisor on each sheet. Indeed, during the time evolution the  divisor points can pass through the double points $X\in \Gamma_{r_1} \cap \Gamma_{r_2}$ only in pairs and coming from different sheets ($r_1\not = r_2$), because of the properties of $\tilde \Psi (P, \vec t)$ settled in Theorem \ref{theo:eff_div} (see Figure \ref{fig:7}).

\begin{figure}[!tbp]
  \centering
  {\includegraphics[scale=0.15,angle=0,width=0.56\textwidth]{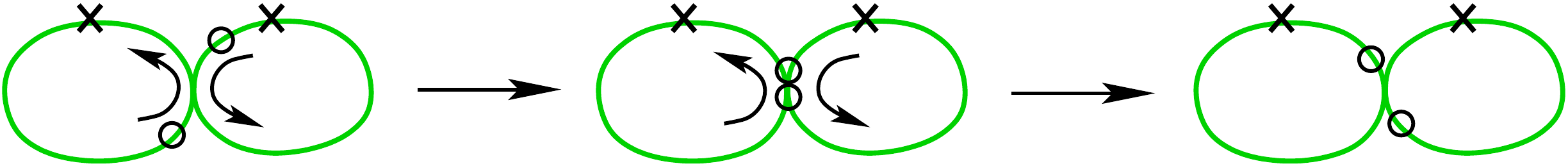}}
\caption{\small{\sl A pair of divisor points passes through a double point.}}\label{fig:7}        
\end{figure}

In Theorem~\ref{theo:divisor} we have shown that any finite oval contains exactly one point of the zero divisor  
$\mathcal{D}_{\xi} (\vec t)$. Let us provide additional information about the positions of the zero divisor points
for any fixed real $\vec t$.

\begin{corollary}\label{cor:zeros1} {\bf Characterization of the zero divisor  $\mathcal{D}_{\xi} (\vec t)$}
For any fixed $\xi\gg1$ and for any real $\vec t$, we have
\begin{enumerate}
\item\label{it:g01} $\mathcal{ D}^{(0)}{ (\vec t)}\subset ]\kappa_1,\kappa_M[$ and $\# \left( \mathcal{ D}^{(0)}{ (\vec t)}\cap ]\kappa_1,\kappa_M[ \right) =N$;
\item\label{it:g02} There is at most one divisor point in each interval $[\kappa_j, \kappa_{j+1}]$, $j\in [M-1]$;
\item For any $r\in [N]$, $\mathcal{ D}^{(r)}_{\xi} { (\vec t)}\subset ]\lambda^{(r)}_{M-N+1}, \alpha^{(r)}_{M-N+1}[$  and 

$\# \left( \mathcal{ D}^{(r)}_{\xi}{ (\vec t)}\cap ]\lambda^{(r)}_{M-N+1}, \alpha^{(r)}_{M-N+1}[ \right) =M-N-1$;
\item For any $r\in [N]$, there is at most one divisor point in each interval $[\lambda^{(r)}_{j+1}, \lambda^{(r)}_{j}]$, $j\in [M-N]$;
\item For any $r\in [N]$, there is at most one divisor point in each interval $[\alpha^{(r)}_{j}, \alpha^{(r)}_{j+1}]$, $j\in [M-N]$, where $\alpha^{(r)}_1 =\lambda^{(r)}_1$;
\item For any $r\in [N]$, 
\[
s^{(r)} { (\vec t)}\equiv \# \left( \mathcal{ D}^{(r)}_{\xi} { (\vec t)}\cap [\lambda^{(r)}_1, \alpha^{(r)}_{M-N+1}[ \right)\le 
\min \{ N-r, M-N-r\}.
\]
\end{enumerate}
\end{corollary}

For any fixed real $\vec t$, we have the complete control of the position of the divisor $\mathcal{D}^{(r)}_{\xi} (\vec t)$ and in particular it is possible to estimate the position of the effective divisor.

\begin{corollary}\label{cor:zeropos} \textbf{(Counting the number of positive poles and zeros on each sheet)}
For any fixed $\vec t$, the number of negative and positive divisor points in $\mathcal{D}^{(r)}_{\xi} (\vec t)$ is uniquely determined for all $r\in [N]$ from $\mathcal{ D}^{(0)} (\vec t)$.
Indeed let $\vec t$ be fixed and define
\[
\begin{array}{lll}
\displaystyle s^{(0)}\equiv s^{(0)} ({\vec t}) &\displaystyle = \# \left( \mathcal{D}^{(0)} (\vec t) \cap [\kappa_N,\kappa_M{[}\right); &\\[0.5ex]
\displaystyle s^{(r)}\equiv s^{(r)} ({\vec t}) &\displaystyle = \# \left( \mathcal{D}^{(r)}_{\xi} (\vec t) \cap [\lambda^{(r)}_1,\alpha^{(r)}_{M-N+1}{[}\right); & \quad r\in [N].
\end{array}
\]
\begin{enumerate}
\item $s^{(r)}$ is a decreasing function of $r$, $r\in [0,N]$ and $s^{(N)} =0$;
\item $s^{(r)} \le \min \{ N-r, M-N-r\} $, for all $r\in [0,N]$;
\item If $s^{(0)}=1$, then $s^{(r)} =0$ for any $r\in [N]$.
\item If $s^{(0)} >1$ and 
$\# \left( \mathcal{D}^{(0)} \cap [\kappa_{N-1}, \kappa_N[ \right) = 1$, then
$s^{(1)}=s^{(0)}$; otherwise $s^{(1)} =s^{(0)}-1$;
\item Let $r\in [2,N]$ be fixed and suppose that $s^{(r-1)} \ge 1$. Then:

if $\# \left( \mathcal{D}^{(0)} \cap [\kappa_{N-r+1}, \kappa_{N-r+2}[ \right) = 1$, then
$s^{(r)}=s^{(r-1)}$; 

otherwise $s^{(r)} =s^{(r-1)}-1$.
\end{enumerate}
\end{corollary}

\begin{corollary}\label{cor:sign}
Under the hypotheses of the Theorem \ref{theo:divisor} and for any fixed $\vec t$, the following holds true:
\begin{enumerate}
\item $D\Psi^{(0)} (\kappa_M, \vec t) >0$, $(-1)^N D\Psi^{(0)} (\kappa_1 ,\vec t)>0$;
\item $(-1)^r D\Psi^{(r)} (\alpha^{(r)}_{M-N+1}, \vec t) >0$, for all $r\in [N-1]$;
\item $\tilde \Psi^{(0)} (\kappa_M, \vec t) >0$, $\tilde \Psi^{(0)} (\kappa_1, \vec t) >0$;
\item $\tilde \Psi^{(r)} (\alpha^{(r)}_{M-N+1}, \vec t) >0$, for all $r\in [N]$.
\end{enumerate}
\end{corollary}

\begin{remark}
In \cite{Mal}, Malanyuk states  that if $A$ is an element of $Gr^{\mbox{\tiny TNN}} (N,M)$ then, for $j\in [N]$, $\gamma^{(0)}_j (\vec 0)$are real, distinct and lie in $[\kappa_1,\kappa_M]$. 
Our estimates improve such result in the case $Gr^{\mbox{\tiny TP}} (N,M)$ and are optimal.
\end{remark}

\section{$\Gamma(\xi)$ and the vacuum divisor for soliton data in $Gr^{\mbox{\tiny TP}} (2,4)$}\label{sec:example}

In this section we construct the rational curve and the vacuum pole divisor associated to generic soliton data in $Gr^{\mbox{\tiny TP}} (2,4)$.
The degenerate curve $\Gamma(\xi)$ in Definition (\ref{def:gamma_xi}) is the partial normalization of the nodal plane curve in (\ref{eq:curveGr24}), which is the rational degeneration of the genus 4 $\mathtt M$--curve in (\ref{eq:curveGr24_pert}). Let us recall that
generic curves of sufficiently high genus can not be represented as plane curves without self-intersections \cite{GrH}, \cite{ACG}, 
therefore the partial normalization is generically unavoidable. We plot both the topological model and the partial normalization 
for this example in Figure \ref{fig:gr24_top}.

\subsection{$\Gamma(\xi)$ and its desingularization for generic soliton data in $Gr^{\mbox{\tiny TP}} (2,4)$}

\begin{figure}[!tbp]
  \centering
  {\includegraphics[width=0.32\textwidth]{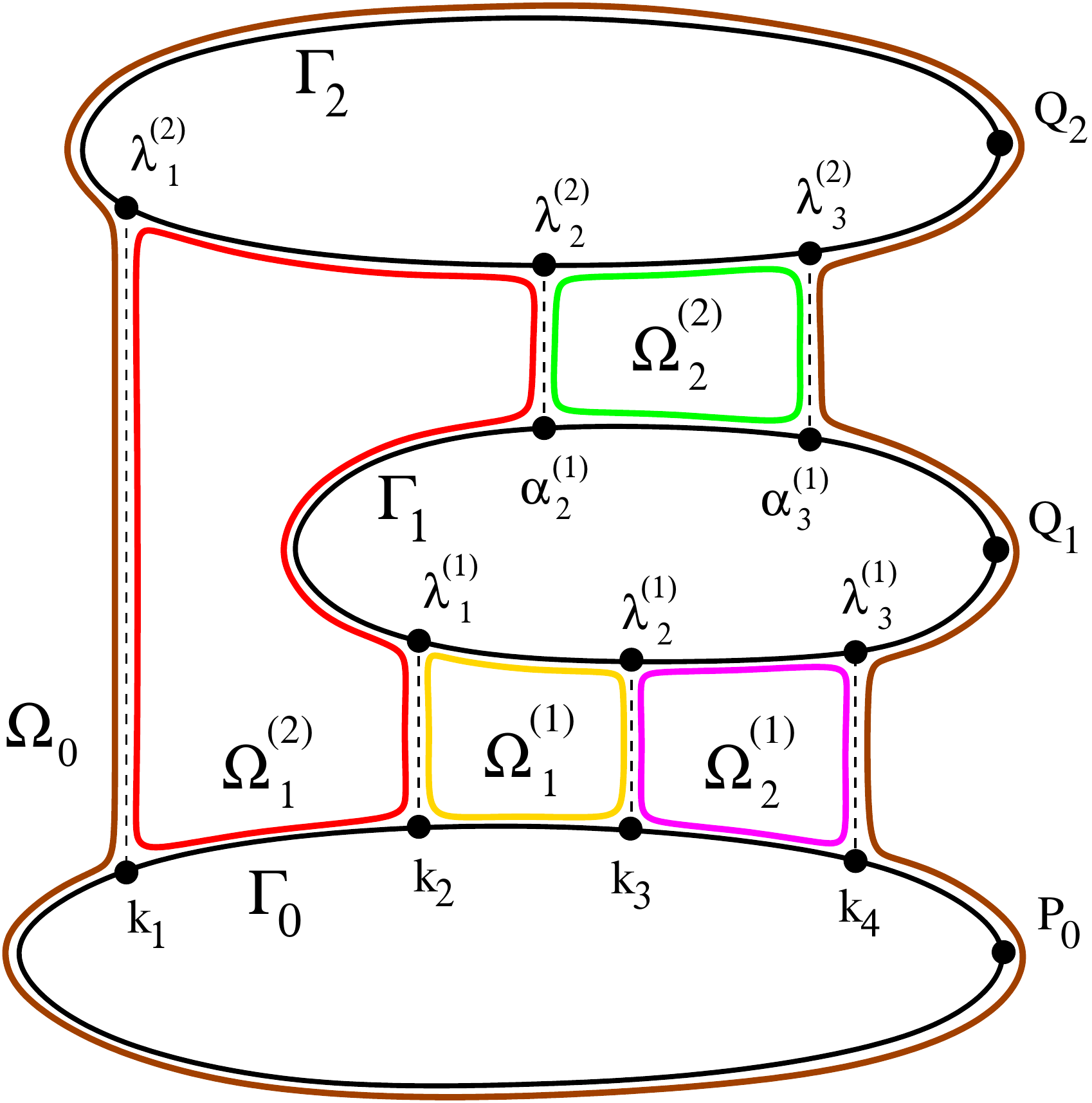}}
  \hspace{.6 truecm}
  {\includegraphics[,width=0.42\textwidth]{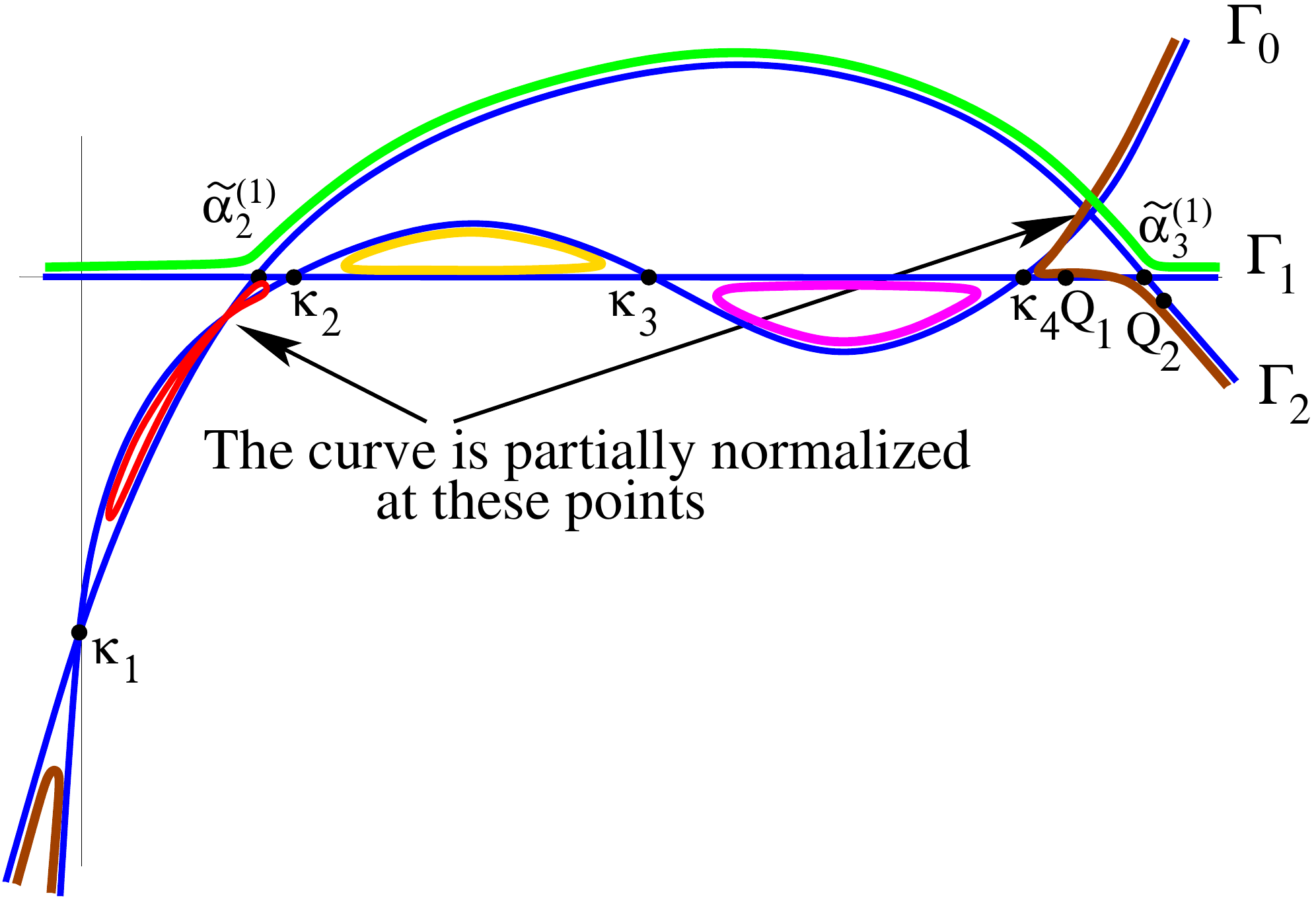}}
\caption{\small{\sl The topological scheme of spectral curve for soliton data $Gr^{\mbox{\tiny TP}}(2,4)$, $\Gamma(\xi)$ (left) is the partial normalization the plane algebraic curve (right), which is the rational degeneration of the genus 4 $\mathtt M$--curve in (\ref{eq:curveGr24_pert}). The ovals in the nodal plane curve are labeled as in the real part of its partial normalization.}}\label{fig:gr24_top}        
\end{figure}

According to Definition \ref{def:gamma_xi}, the degenerate curve $\Gamma(\xi)$ is obtained gluing three copies of $\mathbb{CP}^1$: $\Gamma_0$, $\Gamma_1$ and $\Gamma_2$. 

On $\Gamma_0$ we have 5 real marked points $\kappa_1<\kappa_2<\kappa_3<\kappa_4$ and $P_0$. Let ${\tilde \zeta}\equiv \zeta$ the local coordinate on $\Gamma_0$ such that $\tilde \zeta^{-1} (P_0)=0$ and let us denote with the same symbol the phases and their local coordinates, ${\tilde \zeta} (\kappa_j) =\kappa_j$, $j\in [4]$. To simplify the expressions below and without loss of generality, we take ${\tilde \zeta}(\kappa_1)=0$.

On $\Gamma_1$, we have 6 real ordered marked points, with $\zeta$--coordinates: $\zeta(\alpha^{(1)}_3)=\xi>\zeta(\alpha^{(1)}_2)=\xi^{-1}>\zeta(\lambda^{(1)}_1)=0>\zeta(\lambda^{(1)}_2)=-1>\zeta(\lambda^{(1)}_3)=-\xi^2$ and $\zeta^{-1}(Q_1)=0$. 
The fractional linear change of coordinates $\zeta\mapsto {\tilde \zeta}$ is uniquely defined by the conditions that the $\tilde \zeta$ coordinates of the marked points $\lambda^{(1)}_j\in \Gamma_1$, $j\in [3]$, coincide with the $\tilde \zeta$ coordinates of $\kappa_{j+1}\in \Gamma_0$:
${\tilde \zeta} (\lambda^{(1)}_1)=\kappa_2$, ${\tilde \zeta} (\lambda^{(1)}_2)=\kappa_3$ and
${\tilde \zeta} (\lambda^{(1)}_3)=\kappa_4$:
\begin{equation}\label{eq:coorG1}
{\tilde \zeta} = \frac{\xi^2\kappa_2(\kappa_4-\kappa_3)-\zeta [\xi^2 \kappa_4 (\kappa_3-\kappa_2)-\kappa_3(\kappa_4-\kappa_2)]}{\xi^2(\kappa_4-\kappa_3)-\zeta [\xi^2 (\kappa_3-\kappa_2)-(\kappa_4-\kappa_2)]}
\end{equation}
Then, it is easy to verify that
\begin{equation}\label{eq:marG1}
\begin{array}{l}
{\tilde \zeta} (\alpha^{(1)}_2) = \frac{\xi^2\kappa_2 (\kappa_4-\kappa_3)-(\xi-1)\kappa_3 (\kappa_4-\kappa_2) }{\xi^2 (\kappa_4-\kappa_3)-(\xi-1)(\kappa_4-\kappa_2)}=\kappa_2 -\frac{(\kappa_3-\kappa_2)(\kappa_4-\kappa_2)}{\xi(\kappa_4-\kappa_3)}+ O\left(\xi^{-2}\right)\\
{\tilde \zeta} (\alpha^{(1)}_3) = \frac{\xi\kappa_4 (\kappa_3-\kappa_2)-\kappa_3 (\kappa_4-\kappa_2) }{\xi (\kappa_3-\kappa_2)-(\kappa_4-\kappa_2)}=\kappa_4 +\frac{(\kappa_4-\kappa_3)(\kappa_4-\kappa_2)}{\xi(\kappa_3-\kappa_2)}+ O\left(\xi^{-2}\right)\\
{\tilde \zeta} (Q_1) = \frac{\xi^2\kappa_4 (\kappa_3-\kappa_2)-\kappa_3 (\kappa_4-\kappa_2) }{\xi^2 (\kappa_3-\kappa_2)-(\kappa_4-\kappa_2)}=\kappa_4 +\frac{(\kappa_4-\kappa_3)(\kappa_4-\kappa_2)}{\xi^2(\kappa_3-\kappa_2)}+ O\left(\xi^{-4}\right).
\end{array}
\end{equation}

On $\Gamma_2$, we have 4 real ordered marked points, with $\zeta$--coordinates: $\zeta(\lambda^{(2)}_1)=0>\zeta(\lambda^{(2)}_2)=-1>\zeta(\lambda^{(2)}_3)=-\xi^2$ and $\zeta^{-1}(Q_2)=0$. 
The fractional linear change of coordinates $\zeta\mapsto {\tilde \zeta}$ is uniquely defined by the condition that the $\tilde \zeta$ coordinates at the double points $\lambda^{(2)}_j\in \Gamma_2$, $j\in [3]$ respectively coincide with the $\tilde \zeta$ coordinates of $\kappa_1\in \Gamma_0$, $\alpha^{(1)}_j\in \Gamma_1$:
${\tilde \zeta} (\lambda^{(2)}_1)=0$, ${\tilde \zeta} (\lambda^{(2)}_2)={\tilde \zeta} (\alpha^{(1)}_2)$ and ${\tilde \zeta} (\lambda^{(2)}_3)=\alpha^{(1)}_2$:
\begin{equation}\label{eq:coorG2}
{\tilde \zeta} = \frac{(\xi^2-1){\tilde \zeta} (\alpha^{(1)}_2){\tilde \zeta} (\alpha^{(1)}_3)\zeta}{\xi^2({\tilde \zeta} (\alpha^{(1)}_2)-{\tilde \zeta} (\alpha^{(1)}_3))-\zeta ({\tilde \zeta} (\alpha^{(1)}_3)-\xi^2{\tilde \zeta} (\alpha^{(1)}_2))},
\end{equation}
with ${\tilde \zeta} (\alpha^{(1)}_2),{\tilde \zeta} (\alpha^{(1)}_3)$ as in (\ref{eq:marG1}). Then, it is easy to verify that
\begin{equation}\label{eq:marG2}
{\tilde \zeta} (Q_2) = \frac{(\xi^2-1){\tilde \zeta} (\alpha^{(1)}_3) {\tilde \zeta} (\alpha^{(1)}_2)}{\xi^2 {\tilde \zeta} (\alpha^{(1)}_2)-{\tilde \zeta} (\alpha^{(1)}_3)}={\tilde \zeta} (\alpha^{(1)}_3) +\frac{{\tilde \zeta} (\alpha^{(1)}_3)({\tilde \zeta} (\alpha^{(1)}_3)-{\tilde \zeta} (\alpha^{(1)}_2))}{\xi^2{\tilde \zeta} (\alpha^{(1)}_2)}+ O\left(\xi^{-4}\right) 
\end{equation}
and 
\[
{\tilde \zeta}(Q_2)-{\tilde \zeta}(\alpha^{(1)}_{3})=\frac{(\kappa_4-\kappa_2)\kappa_4}{\xi^2\kappa_2} + O(\xi^{-3}).
\]
Then, in the local coordinates $(\tilde \mu,\tilde \zeta)$, the reducible real algebraic curve $\Gamma(\xi)$ corresponding to the gluing rules of Definition \ref{def:gamma_xi} is the following plane nodal curve
\begin{equation}\label{eq:curveGr24}
\quad\quad \tilde\mu\cdot (\tilde\mu-p_0(\tilde\zeta))\cdot(\tilde\mu+p_2(\tilde\zeta))=0,
\end{equation}
where $p_0(\tilde\zeta) =\prod_{j=1}^3(\tilde \zeta -\kappa_j)$, $p_2(\tilde \zeta) =\frac{\kappa_2\kappa_3\kappa_4}{\tilde \zeta(\alpha^{(1)}_2)\tilde \zeta(\alpha^{(1)}_3)}(\tilde \zeta-\tilde \zeta(\alpha^{(1)}_2))(\tilde \zeta -\tilde \zeta(\alpha^{(1)}_3))$. (\ref{eq:curveGr24}) is the rational degeneration of the $\mathtt M$--curve, which is the normalization 
of the following nodal curve:
\begin{equation}\label{eq:curveGr24_pert}
\Gamma(\xi)_{\delta} \; : \quad\quad \tilde\mu\cdot (\tilde\mu-p_0(\tilde\zeta))\cdot(\tilde\mu+p_2(\tilde\zeta)) -\delta^2
\left( \frac{p_0(\tilde \zeta)+p_2(\tilde \zeta)}{\tilde \zeta} \right)^2 (\kappa_4 +\frac{1}{2\xi}-\tilde \zeta)=0, \ \ 
|\delta|\ll 1.
\end{equation}
It is easy to verify that the curve (\ref{eq:curveGr24_pert}) has genus 4 for generic values of $\kappa_j$s and $\xi\gg1$. By construction it also possesses the maximum number of ovals.

\subsection{The leading order coefficients and vectors of the vacuum wavefunction}
For completeness we compute the coefficients and the vectors of Theorem \ref{lemma:vectors} for points in $Gr^{\mbox{\tiny TP}} (2,4)$. The totally positive matrix in banded form defined in (\ref{eq:our_form1}) expressed in function of the FZ-basis $x_{l,s}$ is
\[
A = \left( \begin{array}{cccc}
1 & \displaystyle \frac{x_{2,2} + x_{1,2}x_{2,1}}{x_{1,1}x_{1,2}}& \displaystyle \frac{x_{2,2}}{x_{1,2}}& 0 \\
0 &1 & x_{1,1} & x_{1,2} \\
\end{array}
\right),
\]
and $\hat A$ is the corresponding normalized totally positive matrix, obtained dividing each $A^i_j$ by the sum of the elements on the $i$-th row (see Remark \ref{rem:norm}).
Then the bases $E^{(0)}$, $E^{(1)}$ and the coefficients  $B^{(1)}_j$, $B^{(2)}_j$, $j\in[3]$ from the 
Principal Algebraic Lemma and Theorem~\ref{lemma:vectors} can be easily calculated:
\[
\hat E^{(0)} = \left( \begin{array}{cccc}
0& 1&0&0\\
0&0&1&0\\
0&0&0&1
\end{array}
\right),
\quad
\hat E^{(1)} = \left( \begin{array}{cccc}
1&0&0 &0\\
0&1&0&0\\
0&\frac{ 1}{1+x_{1,1}} &\frac{x_{1,1}}{1+x_{1,1}}&0
\end{array}
\right),
\]
\[
\hat B^{(1)}_1= \frac{1}{1+x_{1,1}+x_{1,2}}, \quad \hat B^{(1)}_2= \frac{x_{1,1}}{1+x_{1,1}+x_{1,2}}, 
\quad \hat B^{(1)}_3= \frac{x_{1,2}}{1+x_{1,1}+x_{1,2}}
\]
\begin{equation}\label{eq:leading}
\hat B^{(2)}_1 = \hat A^1_1 = \frac{x_{1, 2}x_{1, 1}}{x^{(2)}_D},
\quad
\hat B^{(2)}_2 = \frac{\Delta_{[2,3]}}{\hat A^2_3}= \frac{x_{2, 1}x_{1, 2}}{x^{(2)}_D},\quad
\hat B^{(2)}_3 = \frac{\Delta_{[3,4]} (\hat A^2_2+ \hat A^2_3)}{\hat A^2_3\hat A^2_4} =
\frac{x_{1, 1} x_{2, 2}+x_{2, 2}}{x^{(2)}_D}
,
\end{equation}
where $\Delta_{[2,3]} = \hat A^1_2 \hat A^2_3-\hat A^1_3 \hat A^2_2$, $\Delta_{[3,4]} = \hat A^1_3 \hat A^2_4$ and $x^{(2)}_D = x_{1, 1}(x_{1, 2}+ x_{2, 2})+x_{1, 2} x_{2, 1}+x_{2, 2}$.

\subsection{The vacuum divisor for soliton data in $Gr^{\mbox{\tiny TP}} (2,4)$}
Let us calculate now the vacuum wave function from Theorem~\ref{theo:main0} for this example. In the local coordinate $\zeta$, the vacuum wavefunction takes the form
\[
\begin{array}{ll}
\Psi^{(0)} (\zeta, \vec t ) = \exp (\theta(\zeta, \vec t)), &\zeta \in \Gamma_0,\\
\Psi^{(1)} (\zeta, \vec t) = \frac{(\xi^2+\zeta)(\zeta+1)e^{\theta_2(\vec t)}+(\xi^2+\zeta)\zeta x_{1,1}e^{\theta_3(\vec t)}+\zeta(\zeta+1)x_{1,2}e^{\theta_4(\vec t)}}{ \zeta^2 (1+x_{1,1}+x_{1,2}) + \zeta (\xi^2 [1+x_{1,1}]+x_{1,2}+1) + \xi^2}, &\zeta\in \Gamma_1,\\
\Psi^{(2)} (\zeta, \vec t) = \frac{(\zeta+1)(\zeta+\xi^2)e^{\theta_1(\vec t)} + B^{(2)}_2 \zeta(\zeta+\xi^2) V^{(2)}_2 (\xi,\vec t) +B^{(2)}_3 \zeta(\zeta+1) V^{(2)}_3 (\xi,\vec t) }{C^{(2)}(\zeta^2 +\nu^{(2)}_1\zeta +\nu^{(2)}_2)},&\zeta\in \Gamma_2
\end{array}
\]
with
\[
\begin{array}{l}
V^{(2)}_2 (\xi,\vec t) = \Psi^{(1)} (\xi^{-1}, \vec t)=\frac{(1+\xi^3) e^{\theta_2(\vec t)}+x_{1,1}(\xi^2-\xi +1) e^{\theta_3(\vec t)}+x_{1,2} e^{\theta_4(\vec t)}}{\xi^3+\xi^2 x_{1,1}-\xi x_{1,1}+x_{1,1}+x_{1,2}+1} \\
\quad\quad\quad\quad = \left(e^{\theta_2(\vec t)} + \frac{x_{1,1}}{\xi}e^{\theta_3(\vec t)} + \frac{x_{1,2}}{\xi^3}e^{\theta_4(\vec t)}\right)(1+O (\xi^{-1})),
\\
V^{(2)}_3 (\xi,\vec t) = \Psi^{(1)} (\xi, \vec t)=\frac{ (\xi+1)e^{\theta_2(\vec t)}+\xi x_{1,1} e^{\theta_3(\vec t)}+x_{1,2} e^{\theta_4(\vec t)}}{\xi (x_{1,1}+1) +x_{1,2}+1)} = \left(\frac{e^{\theta_2(\vec t)}+x_{1,1} e^{\theta_3(\vec t)}}{x_{1,1}+1}+ \frac{x_{12}e^{\theta_4(\vec t)}}{(1+x_{11})\xi}\right)(1+O (\xi^{-1})),
\end{array}
\]
\[
\begin{array}{ll}
B^{(2)}_2 &=\frac{(\xi x_{1,2}x_{2,1}-x_{1,2}x_{2,1}-x_{2,2})(\xi^3+\xi^2 x_{1,1}-\xi x_{1,1}+x_{1,1}+x_{1,2}+1)}{(\xi-1)^2\xi^2 x_{1,1}x_{1,2}}= \frac{x_{2,1}}{x_{1,1}} + O(\xi^{-1}),
\\
B^{(2)}_3 &=\frac{(\xi^2 x_{2,2}-\xi (x_{1,2}x_{2,1}+ x_{2,2})+x_{1,2} x_{2,1}+x_{2,2})(\xi (x_{1,1}+1)+ x_{1,2}+1)}{\xi x_{1,1}x_{1,2}(\xi -1)^2}
= \frac{x_{2,2}(1+x_{1,1})}{x_{1,1}x_{1,2}} + O(\xi^{-1}),
\\
\epsilon^{(2)}_1 &= \frac{\xi^2x_{2,2}-\xi x_{1,2}x_{2,1}+x_{1,2}x_{2,1}+x_{2,2}}{\xi^2x_{1,1}x_{1,2}(\xi-1)} = \frac{x_{2,2}}{x_{1,1}x_{1,2} \ \xi} + O(\xi^{-2}),\\
C^{(2)} &= 1+B^{(2)}_2+B^{(2)}_3,\quad \quad
\nu^{(2)}_1 =\frac{\xi^2 }{C^{(2)}},\quad\quad
\nu^{(2)}_2 =  \frac{\xi^2 (1+B^{(2)}_2) + 1+ B^{(2)}_3}{C^{(2)}},
\end{array}
\]
so that, in agreement with the estimates in Theorem \ref{theo:main0}, we have
\[
\frac{1}{C^{(2)}} = \hat B^{(2)}_1 + O(\xi^{-1}),\frac{B^{(2)}_2}{C^{(2)}} = \hat B^{(2)}_2 + O(\xi^{-1}), \quad\quad \frac{B^{(2)}_3}{C^{(2)}} = \hat B^{(2)}_3 + O(\xi^{-1}),
\]
with ${\hat B}^{(2)}_j$, $j\in [3]$, as in (\ref{eq:leading}).

The vacuum poles in $\Gamma_1$ in the local coordinate $\zeta$ take the form
\[
\begin{array}{l}
\zeta(b^{(1)}_1) = - \frac{1}{1+x_{1,1}} + \frac{x_{1,1}x_{1,2}}{\xi^2(1+x_{1,1})^3} +O(\xi^{-4}),\in ]-1,0[,\\
\zeta(b^{(1)}_2) = - \frac{1+x_{1,1}}{1+x_{1,1}+x_{1,2}} \xi^2 
 -\frac{x_{1,1}x_{1,2}}{(1+x_{1,1})(1+x_{1,1}+x_{1,2})} -\frac{x_{1,1}x_{1,2}}{\xi^2(1+x_{1,1})^3} +O(\xi^{-4})\in ]-\xi^2,-1[,
\end{array}
\]
and, using (\ref{eq:coorG1}), in the local coordinates $\tilde \zeta$ we have
\[
\begin{array}{l}
{\tilde \zeta}(b^{(1)}_1) = - \frac{\kappa_2x_{1,1} (\kappa_4-\kappa_3)+\kappa_3(\kappa_4-\kappa_2)}{x_{1,1}(\kappa_4-\kappa_3)+\kappa_4-\kappa_2} + \frac{x_{1,1}(1+x_{1,1}+x_{1,2})(\kappa_4-\kappa_3)(\kappa_4-\kappa_2)(\kappa_3-\kappa_2)}{\xi^2(1+x_{1,1})(x_{1,1}(\kappa_4-\kappa_3)+\kappa_4-\kappa_2)^2} +O(\xi^{-4}), \in ]\kappa_2,\kappa_3[\\
{\tilde \zeta}(b^{(1)}_2) = \kappa_4 -\frac{(\kappa_4-\kappa_2)(\kappa_4-\kappa_3)x_{1,2}}{\xi^2(1+x_{1,1})(\kappa_3-\kappa_2)} +O(\xi^{-4})\in ]\kappa_3,\kappa_4[.
\end{array}
\]
Similarly, the vacuum poles in $\Gamma_2$ in the local coordinate $\zeta$ take the form
\[ 
\begin{array}{ll}
\zeta(b^{(2)}_1) &= - \frac{x_{1,1}}{x_{1,1}+x_{2,1}}+\frac{x_{1,1} (x_{1,1}x_{1,2}x_{2,1}+x_{1,2}x_{2,1}-x_{2,2}}{\xi x_{1,2} (x_{1,1}+x_{2,1})^2}+ O(\xi^{-2})\in ]-1,0[,\\
\zeta(b^{(2)}_2) &= - \frac{x_{1,2} (x_{1,1}+x_{2,1})}{  x_{1,2} (x_{1,1}+x_{2,1})+x_{2,2} (1+x_{1,1})} \xi^2  +c^{(2)}_1 \xi + O(1)\in ]-\xi^2,-1[,
\end{array}
\]
with 
\[
\begin{array}{l}c^{(2)}_1 =  \frac{x_{1,1}^2(x_{1,2}^2x_{2,1}+x_{1,2}x_{2,1}x_{2,2}-x_{1,2}x_{2,2})+x_{1,2}^2(x_{1,1}x_{2,1}^2+x_{1,1}x_{2,1}-x_{1,1}x_{2,2}+x_{2,1}^2-x_{2,1}x_{2,2})}{(x_{1,1}(x_{1,2}+x_{2,2})+x_{1,2}x_{2,1}+x_{2,2})^2}+\\
\quad\quad+\frac{x_{1,1}x_{1,2}x_{2,1}x_{2,2}-2x_{1,1}x_{1,2}x_{2,2}-x_{1,1}x_{2,2}^2-x_{1,2}x_{2,1}x_{2,2}-x_{2,2}^2}{(x_{1,1}(x_{1,2}+x_{2,2})+x_{1,2}x_{2,1}+x_{2,2})^2},
\end{array}
\]
so that, using (\ref{eq:coorG2}), in the local coordinates $\tilde \zeta$ we have
\[
\begin{array}{l}
{\tilde \zeta}(b^{(2)}_1) =\frac{\kappa_2\kappa_4x_{1,1}}{\kappa_4 x_{11}+x_{21}(\kappa_4-\kappa_2)}- \frac{(\kappa_4-\kappa_3)(\kappa_4-\kappa_2)}{(\kappa_3-\kappa_2)\xi}, \in ]0,{\tilde \zeta}(\alpha^{(1)}_2)[,\\
{\tilde \zeta}(b^{(2)}_2 )= {\tilde \zeta}(\alpha^{(1)}_3) -\frac{x_{2,2}(1+x_{11})\kappa_4(\kappa_4-\kappa_2)}{\xi^2x_{1,2}\kappa_2(x_{1,1}+x_{2,1})} + O(\xi^{-3})\in ]{\tilde \zeta}(\alpha^{(1)}_2),{\tilde \zeta}(\alpha^{(1)}_3)[.
\end{array}
\]
The effective divisor may be easily computed applying the Darboux transformation generated by $f^{(1)} (\vec t) = e^{\theta_2}+ x_{1,1} e^{\theta_3} + x_{1,2}e^{\theta_4}$ and $f^{(2)} (\vec t)= e^{\theta_1} + \frac{x_{2,2} + x_{1,2}x_{2,1}}{x_{1,1}x_{1,2}}e^{\theta_2}+\frac{x_{2,2}}{x_{1,2}}e^{\theta_3}$.

\appendix

\section{Points in $Gr^{\mbox{\tiny TP}}(N,M)$ and totally positive matrices in classical sense}
\label{sec:totpos}

In this Section we provide the characterization of the positivity properties of the banded matrices introduced in 
Definition~\ref{def:bandmat}. We show that such matrices are totally positive in classical sense (see Proposition \ref{prop:pos}) and we parametrize $Gr^{\mbox{\tiny TP}} (N,M)$ using the local coordinates $x_{r,s}$, $r\in [N]$, $s\in [M-N]$, which are naturally associated to our construction (see Proposition \ref{prop:xcoor}). We believe that all results presented in this Appendix are known 
to experts, but we failed to find an appropriate reference for some of them.

Points in the totally positive Grassmannian $Gr^{\mbox{\tiny TP}}(N,M)$ may be represented by real $N	\times M$ matrices with all maximal ($N\times N$) minors strictly positive. 
$Gr^{\mbox{\tiny TP}}(N,M)$ is the top cell in the sense of Postnikov's decomposition \cite{Pos} of the totally non-negative Grassmannian $Gr^{\mbox{\tiny TP}}(N,M)$.

Let us discuss the relations between total positivity of matrices in the classical sense and the property of total positivity in 
the Grassmannian.  

\begin{definition}
We recall that a matrix $B$ is called totally positive (respectively strictly totally positive) if all minors of all orders of $B$ are non-negative (respectively positive) \cite{Pinkus}.
\end{definition}

It is easy to establish the following natural connection between points of  
$Gr^{\mbox{\tiny TP}}(N,M)$ and  $N\times (M-N)$ strictly  totally positive matrices (see \cite{Pos}): let 
the $N\times M$ matrix $A$, represent a point in $Gr^{\mbox{\tiny TP}}(N,M)$. Then, using the standard elementary operations on rows it can be uniquely transformed to reduced row echelon form:
\begin{equation}\label{A_RRE}
A^{\mbox{\tiny RRE}}=\left[\begin{array}{cccccccccc}
1             & \cdots  & 0             & 0            & 0 & | & \pm b_{N \, 1} & \pm b_{N \, 2} & \cdots & \pm b_{N \, M-N} \\
\hphantom{0}  & \ddots  & \hphantom{0}  & \hphantom{0}  & \hphantom{0} & |              &                & \cdots &                \\
0 & \cdots  & 1 & 0 & 0 & | & b_{3\, 1} &  b_{3 \, 2} & \cdots & b_{3 \, M-N} \\
0 & \cdots  & 0 & 1 & 0 & | & -b_{2 \, 1} &  -b_{2 \, 2} & \cdots & -b_{2 \, M-N} \\
0 & \cdots  & 0 & 0 & 1 & | & b_{1 \, 1} &  b_{1 \, 2} & \cdots & b_{1 \, M-N} 
\end{array}\right]
\end{equation}
where the matrix $B$
\begin{equation}\label{eq:BTP}
{B}=\left[\begin{array}{cccc}
b_{1 \, 1} &  b_{1 \, 2} & \cdots & b_{1 \, M-N} \\
b_{2 \, 1} &  b_{2 \, 2} & \cdots & b_{2 \, M-N} \\
b_{3\, 1} &  b_{3 \, 2} & \cdots & b_{3 \, M-N} \\
                & \cdots &     &           \\
b_{N \, 1} & b_{N \, 2} & \cdots & b_{N \, M-N} \\
\end{array}\right]
\end{equation}
is strictly totally positive.

A convenient characterization of strictly totally positive matrices is the following.

\begin{theorem}(Theorem 2.3 page 39, \cite{Pinkus})\label{theo:STP39}
$B$ is strictly totally positive if and only if all $k$-th order minors of $B$ composed by the first $k$ rows and $k$ consecutive columns, and also all $k$-th order minors of $B$ composed by the first $k$ columns and $k$ consecutive rows are strictly positive for $k=1,\dots, \min\{N,M{-N}\}$. 
\end{theorem}

The number of such minors is $N\times( M-N)$ and they form a basis of coordinates for strictly totally positive $N\times (M-N)$ matrices, since  all of the other minors of $B$ may be expressed in terms of subtraction-free rational functions of such coordinates.  Since any  maximal minor of $A^{\mbox{\tiny RRE}}$ is expressed as a minor of $B$, also all of the maximal minors of $A^{\mbox{\tiny RRE}}$ are expressed as subtraction free rational functions of such coordinates, that is
they form a totally positive basis in Fomin--Zelevinsky sense \cite{FZ1}.

In \cite{Tal}, Talaska studies the problem of reconstructing an element $A \in Gr^{\mbox{\tiny TNN}} (N,M)$ from a subset of its Pl\"ucker coordinates $\Delta_I (A)$. For each cell in the Gelfand--Serganova decomposition of $Gr^{\mbox{\tiny TNN}} (N,M)$ (see for instance \cite{Pos} for necessary definitions), she characterizes a minimal set of Pl\"ucker coordinates $T(L)$ sufficient to reconstruct the corresponding element using Postnikov boundary measurement map and Le--diagrams \cite{Pos}. 
In this way, she constructs a totally positive basis in  Fomin--Zelevinsky sense $T(L)$ associated to the Le--diagram \cite{Pos} of any point in $Gr^{\mbox{\tiny TNN}} (N,M)$. 

It is straightforward to check that for points in $Gr^{\mbox{\tiny TP}} (N,M)$, Talaska's basis of minors concides with the basis in Theorem \ref{theo:STP39}. 

For our purposes it is convenient to transform the matrix $A^{(RRE)}$ to the 
{\bf banded form}:

\begin{equation}
\label{eq:our_form1}
{A}=\left[ \begin{array}{cccccccccccc}  
1 & A_2^1 &  A_3^1 & A_4^1 &\ldots & A_{M-N+1}^1& 0 & 0 & \ldots  & 0 & 0 & 0\\
0 & 1 &  A_3^2 & A_4^2 & \ldots &   A_{M-N+1}^2 & A_{M-N+2}^2  & 0 & \ldots & 0 & 0 & 0 \\
0 & 0 &  1 & A_4^3 &  \ldots  & A_{M-N+1}^3 & A_{M-N+2}^3 &  A_{M-N+3}^3  &  \ldots & 0 & 0 & 0  \\
 & & & & & \cdots & & & & & \\
0 & 0 &  0 & \ldots & 0 &  1 & \ldots & \ldots & \ldots &  A_{M-2}^{N-1} & A_{M-1}^{N-1} & 0  \\
0 & 0 & 0 &  0 & \ldots & 0 &  1 & \ldots & \ldots &  A_{M-2}^{N}  &  A_{M-1}^{N} & A_{M}^{N}  \\
\end{array}\right]
\end{equation}

Here all elements ${A}^i_j$ with $j<i$ or $j>M-N+i$ are 0. 

This transformation can be achieved by applying the Gauss elimination process starting from the last column.\footnote{
We observe that this transformation from the reduced row echelon form to the banded form corresponds to left multiplication by a $N\times N$ upper triangular matrix with unit determinant, therefore it preserves the point of the Grassmannian.}  

In Proposition \ref{prop:pos} we show that $A$ as in (\ref{eq:our_form1}) is a totally positive matrix in classical sense. For the proof we need the following result

\begin{lemma}\label{lem:TP-1}
Let $A$ be a $N\times M$ matrix in the banded form (\ref{eq:our_form1}) representing a point of $Gr^{\mbox{\tiny TP}}(N,M)$.
Consider all $n\times n$ submatrices consisting of consecutive rows and arbitrary columns in increasing order:
$A^{[i,i+1,\ldots,i+n-1]}_{[j_1,j_2,\ldots,j_n]}$ for all $n\in[N]$, $i\in[N-n+1]$. Then their determinants are non-negative:
$$
\Delta^{[i,i+1,\ldots,i+n-1]}_{[j_1,j_2,\ldots,j_n]}\ge 0.
$$
Moreover, this determinant is positive if and only if the corresponding submatrix has no zero columns or rows. 
In particular, all elements $A^i_j$ with $i\le j\le M-N+i$ are positive. 
\end{lemma}
\begin{proof}
If the submatrix has a zero column, then the associated minor is zero. The condition that the matrix 
has no zero columns means exactly that $j_1\ge i$, $j_n \le M-N+i+n-1$. Then
$$
\Delta^{[i,i+1,\ldots,i+n-1]}_{[j_1,j_2,\ldots,j_n]}=
\Delta_{[1,2,\ldots,i-1,j_1,j_2,\ldots,j_n,M-N+i+n,\ldots,M]}>0.
$$
In particular, for  $i\le j\le M-N+i$,
$$
A^i_j=\Delta_{[1,2,\ldots,i-1,j,M-N+i+1,\ldots,M]}>0.
$$
(Condition $i\le j\le M-N+i$ guarantees that this minor has no repeating columns and the columns are 
in increasing order).
\end{proof}

\begin{prop}\label{prop:pos}
Let $A$ be a $N\times M$ matrix in the banded form (\ref{eq:our_form1}) representing a point of $Gr^{\mbox{\tiny TP}}(N,M)$. Then 
the matrix $A$ is totally positive. 
\end{prop}
\begin{proof}
We know already that all maximal ($N\times N$) minors are strictly positive. By Theorem~2.13 in \cite{Pinkus}, page 56,  
Lemma~\ref{lem:TP-1} implies that the matrix $A$ is totally positive.
\end{proof}

The following proposition shows that the representation of a point in $Gr^{\mbox{\tiny TP}}(N,M)$ through a totally positive matrix in banded form as in (\ref{eq:our_form1}) is naturally linked to the strictly totally positive $N\times(M-N)$ matrix $B$ defined in (\ref{eq:BTP}) and gives another criterion to check the total positivity property. 

\begin{prop}\label{prop:xcoor}
Let $A$ be a matrix in banded form with 
$A^i_i =1$, $i\in N$, and $A^i_j =0$ if and only if $j<i$ or $j>M-N+i$, with $i\in [N]$, $j\in [M]$. Let
\begin{equation}\label{simo_x}
x_{r,s} = \Delta^{[N-r+1,\dots,N]}_{[N-r+1+s,\dots,N+s]}(A),\quad\quad r\in [N],\quad s\in [M-N].
\end{equation} 
Then $A$ represents a point of $Gr^{\mbox{\tiny TP}}(N,M)$ if and only if $x_{r,s}>0$, $\forall
r\in [N]$, $\forall s\in [M-N]$.
Moreover in such a case
\begin{equation}\label{eq:ciccc}
x_{r,s} = \left\{ \begin{array}{ll} \Delta^{[1\dots r]}_{[s..\dots s+r-1]} (B), & \quad r\le s\le  M-N-k+1,\\[0.5ex]
\Delta^{[r\dots r+s-1]}_{[1\dots s]} (B), & \quad s< r \le N-s,
\end{array}
\right.
\end{equation}
with $B$ as in (\ref{eq:BTP}). 
\end{prop}

\begin{proof}
The minors $x_{r,s}$ may be transformed to maximal $N\times N$ minors of $A$, so they have to be all positive.
Let now $A$ be in banded form with all $x_{r,s}>0$ as defined in (\ref{simo_x}) and put it in RRE form. By definition it takes the form as in (\ref{A_RRE}) with pivot set $\{1,\dots,N\}$. Let $B$ be the associated matrix as in (\ref{eq:BTP}). Then the minors of $B$ formed by the first $r$ rows and consecutive columns and the minors formed by the first $r$ columns and $r$ consecutive rows, by construction, are just the $x_{r,s}$ minors of {the} matrix $A$ ($r\in [N]$, $s\in [M-N]$) as in (\ref{eq:ciccc}).

Then by Theorem \ref{theo:STP39}, $B$ is strictly totally positive if and only if $x_{r,s}$ are all positive and, in such case, $A$ represents a point $Gr^{\mbox{\tiny TP}}(N,M)$.
\end{proof}

\begin{remark}
The coordinates $x_{r,s}$ coincide with the basis of minors in Lemma \ref{lem:TP-1} and with a totally positive basis in Fomin--Zelevinsky sense, and we shall refer to them simply as the {\bf FZ-basis}. 
\end{remark}

The following Corollary is the key observation which allows to express the 
recursive construction of the $\mathtt M$--curve and of the wavefunction in invariant form.

\begin{corollary}\label{cor:ultime}
Let $A$ be the banded totally positive matrix defined above and representing a given point in $Gr^{\mbox{\tiny TP}} (N,M)$.
Then all of its minors of any order are either zero because they contain a zero row or a zero column, or they are subtraction--free rational espressions in the FZ--basis $x_{r,s}$.
In particular the minors of $A$ formed by the last $r$ rows and $r$ columns are subtraction free rational expressions of the elements $x_{l,s}$, $l\in [r]$, $s\in [M-N]$ of the FZ-basis.
\end{corollary}

We also require the following version of Fekete's Lemma (see \cite{Pinkus}, page 37), adapted to our setting:

\begin{lemma}
\label{lem:l2}
Let ${N}\le M$ and assume $A$ to be a $N\times M$ banded matrix in the form (\ref{eq:our_form1}) with the following properties:
\begin{enumerate}
\item Consider the submatrix $\hat A$ obtained from $A$ by removing the first row and the first column. 
All $N-1$-order minors of $\hat A$ are positive.
\item All $N$-order minors of $A$ composed from consecutive columns are also positive. 
\end{enumerate}
Then all $N$-order minors of $A$ are positive.
\end{lemma}

\section{Lemmas for the proof of Theorem~\ref{theo:main0}}
\label{sec:lemmas}

The proof of Theorem \ref{theo:main0} requires a series of analytic estimates which we provide in this Appendix. We assume here that
$\xi\gg1$.

In the next Lemma for a fixed $r\in[N]$ we assume that the vacuum wave function at infinity is a linear 
combination of its values at the double points $\lambda_j$ with positive coefficients, and we show that there exists an unique 
collection of corresponding divisor points $b_k$ located in proper intervals, we provide the estimates on the positions of the
points $b_k$, and for all marked points $\alpha_n$ we compute  the vacuum wave function at leading order in  $\xi$ for all phases. 
The points $\alpha_n$ are the double points attached to the next component of the spectral curve.

\begin{lemma}\label{lemma:C}
Let $c_n>0$, $n\in [M-N+1]$ and such that $\displaystyle \sum\limits_{n=1}^{M-N+1} c_n=1.$
Let $\lambda_1=0$, $\lambda_k =-\xi^{2(k-2)}$, $k=2,\dots,M-N+1$ and define
\[
C_n (\zeta) = c_n \frac{ \prod_{j\not = n} (\zeta-\lambda_j)}{ \prod_{k=1}^{M-N} (\zeta-b_k)}, \quad\quad n\in[M-N+1].
\]
Then 
\[
C_n (\lambda_j) = \delta^n_j \quad \quad \forall j,n\in [M-N+1],
\]
for uniquely defined poles  $b_k{=b_k(\xi)}\in ] \lambda_{k+1}, \lambda_{k}[$, $k\in [M-N]$, such that {for $\xi \gg1$},
\[
b_k ({\xi})= - \frac{\sum_{j=1}^k c_j }{\sum_{j=1}^{k+1} c_j} \xi^{2(k-1)} (1+O(\xi^{-1)})).
\]
Moreover, in such case $\forall \zeta\in {\mathbb C}$, $\sum_{n=1}^{M-N+1} C_n(\zeta)=1$ and
\begin{equation}\label{eq:Caysm}
C_j (\pm \xi^{2s-5}) = \left\{ \begin{array}{ll} \frac{c_j}{\sum_{l=1}^{s-1} c_l} {\cdot(1+O(\xi^{-1}))}& \quad j\in [{2},s-1],\\[0.5ex]
\pm \frac{c_j }{\left(\sum_{l=1}^{s-1} c_l\right)} {\cdot\frac{(1+O(\xi^{-1}))}{\xi^{2(j-s)+1}}}&\quad j\in [s,M-N+1].
\end{array}\right.
\end{equation}
\end{lemma}

\begin{proof}
Let $P(\zeta) = \prod_{k=1}^{M-N} (\zeta-b_{k})$. Then
$C_j(\lambda_j)=1$ if and only if $P(\lambda_j)= c_j \prod_{k\not = j} (\lambda_j -\lambda_k)$, $j\in [M-N+1]$. Thanks to the positivity of the coefficients $c_j$,  $P(\lambda_j)$ and $P(\lambda_{j+1})$ have opposite signs $j\in[M-N]$ {so that} poles $b_{{k}} \in]\lambda_{{k}+1},\lambda_{{k}}[$, ${k}\in [M-N]$.  

By construction $Q(\zeta) = \sum\limits_{j=1}^{M-N+1} C_j(\zeta)$ is a rational function of degree
 less than or equal to $M-N$ and takes the value $1$ in $M-N+1$ points, from which we conclude that it is constant to 1 everywhere.

The estimate for the leading order expansion of $b_{k}$, $k\in [M-N]$, for $\xi\gg1$, follows from
the fact that, for any $l\in [M-N]$, the $l$-th symmetric product in $b_{k}$s is a linear combination of the $l$-th symmetric products in $\lambda_l$ for $l\not = j$, $j\in [M-N+1]$, that is
\[
\begin{array}{l}
{\hat \pi}_l (b_1,\dots, b_{M-N})  \equiv
\sum_{1\le j_1 <j_2 <\cdots < j_l\le M-N} \left(\prod_{s=1}^l b_{j_s}\right)
=\sum_{j=1}^{M-N+1} c_j {\hat \pi}_l (\lambda_1,\dots,
{\hat \lambda}_j, \dots,\lambda_{M-N+1}) \\[0.5ex]
\sum_{j=1}^{M-N+1} c_j \left( \sum^{\prime}_{1\le j_1 <j_2 <\cdots < j_l\le M-N+1} \left(\prod_{s=1}^l \lambda_{j_s}\right) \right)
 =\left( \sum_{j=1}^{M-N+1-l} c_j \right) \xi^{p(l)} + l.o.t., 
\end{array}
\]
where 
\[
p(l) = 2\sum_{j=M-N-l}^{M-N-1} j = l(2M-2N-1-l),
\]
from which we easily get the assertion on the leading order behavior of the poles.

Finally the estimate on the asymptotic behavior of $C_j (\alpha_s)$ , $s\in [2,M-N]$ easily follows taking into account of the leading orders of $\lambda_j$s  and $b_{k}$s.
\end{proof}

In the next Lemma we use the fact that the vacuum wave function at infinity $\Psi^{(r)}_{\infty} (\vec t)$ is a linear combination of 
its known values at the double points $\alpha_n^{(r-1)}$  with unknown coefficients $B^{(r)}_n$. We impose that  
$\Psi^{(r)}_{\infty} (\vec t)$ is equal to the heat hierarchy solution $f^{r}(\vec t)$ associated to the $N-r+1$-th row of the banded 
matrix $\hat A$ plus small correction, which is a linear combination of heat hierarchy solutions associated to rows below. Then 
the resulting linear system uniquely defines both the coefficients $B^{(r)}_n$ and the small correcting constants $\epsilon^{(r)}_k$.
Moreover, if $\xi$ is sufficiently big, the coefficients $B^{(r)}_n$, $\epsilon^{(r)}_k$ are positive and may be explicitly 
estimated (see Formulas~(\ref{ex:Br})-(\ref{eq:epsilon})). Let us point out that the proof of positivity of $B^{(r)}_n$ 
is essentially based on the Principal Algebraic Lemma.

\begin{lemma}\label{lemma:beps}
Let $r\in [2,N]$ be fixed and $\xi>1$. Let $\alpha_n^{(r-1)}$ ($n\in [2,M-N+1]$) as in (\ref{eq:alphas}), and
\[
\Psi^{(r)}_{\infty} (\vec t) = {\hat A}^{N-r+1}_{N-r+1} e^{\theta_{N-r+1}} + \sum_{n=2}^{M-N+1} B^{(r)}_n 
\Psi^{(r-1)} (\alpha_n^{(r-1)}, \vec t),
\]
for some $B^{(r)}_n\in {\mathbb R}$ and 
\begin{equation}\label{eq:psialpha2}
\begin{array}{r} \Psi^{(r-1)} (\alpha_n^{(r-1)} ,\vec t) = \sum_{j=1}^{M} E^{(r-1)[n]}_{j}(\xi)
e^{\theta_{j}(\vec t)} =\left\{ 
\sum_{j= N-r +2 }^{ {N-r+n}}\sigma^{(r-1)}_{n,j} e^{\theta_{j}}+ \sum_{j=N-r+n+1}^{N+n-2} \frac{\sigma^{(r-1)}_{n,j}e^{\theta_j}}{\xi^{j-N+r-n -1}} 
+ \right. 
\\[0.5ex]
 +\left.\sum_{j=N+n-1}^{M} \frac{\sigma^{(r-1)}_{n,j} e^{\theta_{j}}}{\xi^{r-1+2(j-N-n+1)}} 
  \right\}\times \frac{\left( 1 + O(\xi^{-1})\right)}{\sum_{s= N-r {+2} }^{ {N-r+n}}\sigma^{(r-1)}_{n,s}},
\end{array}
\end{equation}
where for all  $n\in [2,M-N+1]$, $j\in [N-r+2,M]$, $\sigma^{(r-1)}_{n,j}>0$ are constants independent of $\xi$, and, moreover:
\begin{equation}\label{eq:sigmaNM}
\sigma^{(r-1)}_{n,j}=\left\{\begin{array}{ll}
\Delta_{[j;N-r+n {+1},N-r+n {+2},\ldots,N+n-2]} & \mbox{if} \ \ j\in [N-r {+2}, {N-r+n} ] \\[0.5ex]
\Delta_{[N-r+n+1,N-r+n+2,\ldots,N+n-2;j]} & \mbox{if} \ \ j\in[N+n-1,M].
\end{array}\right.
\end{equation}
Then the requirement
\begin{equation}\label{eq:inftycond}
\Psi^{(r)}_{\infty} (\vec t) = \sum_{j=N-r+1}^{M} \left( {\hat A}^{N-r+1}_{j} + \sum_{k=1}^{r-1}
{\hat A}^{N-r+k+1}_{j} \epsilon^{(r)}_k \right) e^{\theta_{j}},
\end{equation}
uniquely defines $B^{(r)}_n=B^{(r)}_n(\xi)$, $n\in [2,M-N+1]$, and $\epsilon^{(r)}_k= \epsilon^{(r)}_k(\xi)$, $k\in [r-1]$, which are rational in $\xi$ and strictly positive for all $\xi\gg1$.
Moreover the following estimates hold true
\begin{equation}\label{ex:Br}
B^{(r)}_n = \frac{\Delta_{[N-r+n,\dots,N+n]} \left( \sum_{s=N-r+2}^{N-r+n} \Delta_{[s; N-r+n+1,\dots,N+n-1]} \right) }{\Delta_{[N+r+n,\dots,N+n-1]} \Delta_{[N-r+n+1,\dots,N+n]} }  (1 + O(\xi^{-1}));
\end{equation}
\begin{equation}\label{eq:epsilon}
\epsilon^{(r)}_k =\frac{\sigma^{(r-1)}_{M-N+1,M-r+k+1} \cdot {\hat A}^{N-r+1}_{M-r+1}}{\sigma^{(r-1)}_{M-N+1,M-r+1} \cdot {\hat A}^{N-r+k+1}_{M-r+k+1}} \cdot \frac{1}{\xi^k} \left(1+O(\xi^{-1}) \right). 
\end{equation}
\end{lemma}

\begin{proof}
The proof is straightforward since the linear system associated to (\ref{eq:inftycond}) in $B^{(r)}_j$, $\epsilon^{(r)}_k$ is clearly compatible for $\xi\gg1$, the coefficients are rational functions in $\xi$.
Let us define
\begin{equation}\label{eq:hatsigma}
\hat \sigma^{(r-1)}_{n,j} = \frac{\sigma^{(r-1)}_{n,j}}{\sum_{s= N-r {+2} }^{ {N-r+n}}\sigma^{(r-1)}_{n,s}}, \quad \forall n\in [2,M-N+1], j\in [N-r+2,M],
\end{equation}
then, for $\xi \gg1$, the linear system  may be expressed as
\[
\displaystyle \sum_{\hat j=\hat n}^{M-N}  \hat \sigma^{(r-1)}_{\hat j+1,N-r+\hat n +1} B^{(r)}_{\hat j+1}-\sum_{k=1}^{r-1}\epsilon^{(r)}_k {\hat A}^{N-r+k+1}_{N-r+\hat n +1} = {\hat A}^{N-r+1}_{N-r+\hat n +1} + O( \xi^{-1}),
\quad \hat n\in [M-N],
\]
\[
\displaystyle \sum_{j=s}^{r-1} \frac{\hat\sigma^{(r-1)}_{M-N+1+s-j,M-N+s}}{\xi^{j}} B^{(r)}_{M-N+1+s-j} (1+O(\xi^{-1})) -\sum_{l=s}^{r-1}\epsilon^{(r)}_l {\hat A}^{N-r+1+l}_{M-r+1+s} = 0,
\quad s\in [r-1].
\]
Using the Principal Algebraic Lemma and Theorem \ref{lemma:vectors}, we easily conclude that, at leading order in $\xi$ the above system is equivalent to the linear system
\[
\hat 	\Omega \hat c = \hat p,
\]
in the unknowns $\hat c = [B^{(r)}_{2},\cdots , B^{(r)}_{M-N+1},\epsilon^{(r)}_1,\dots,\epsilon^{(r)}_{r-1} ]^T$, where
$\hat p = [\hat A^{(N-r+1)}_{N-r+2},\cdots , \hat A^{(N-r+1)}_{M-r+1},0,\dots,0 ]^T$
and $\hat \Omega$ is the $(M-N+r-1)\times (M-N+r-1)$ matrix, such that for $\hat n \in [M-N]$:
\[
\hat \Omega^{\hat n}_{\hat j} = \left\{ \begin{array}{ll} 
 \hat\sigma^{(r-1)}_{\hat j+1,N-r+\hat n +1}, & \quad \hat j \in [\hat{n},M-N]\\[0.5ex]
0 &  \quad \hat j \in [\hat n-1],\\[0.5ex]
{\hat A}^{M+r-1-\hat j}_{N-r+\hat{n}+1}, & \quad \hat j=\in [M-N+1,M-N+r-1],
\end{array}\right.
\]
and for $\hat n \in [M-N+1, M-N+r-1]$
\[
\hat \Omega^{\hat n}_{\hat j} = \left\{ \begin{array}{ll} 0 &  \hat j \in M-N-1]\\[0.5ex]
\frac{\hat \sigma^{(r-1)}_{M-N+1,N-r+\hat n +1}}{\xi^{\hat M-N-n}}, & \quad \hat j =M-N\\[0.5ex]
{\hat A}^{M+r-1-\hat j}_{N-r+\hat{n}+1}, & \quad \hat j=\in [M-N+1,M-N+r-1],
\end{array}\right.
\]
that is
\[
\resizebox{\textwidth}{!}{$
\scriptstyle
{\hat \Omega} = \left[ \scriptstyle\begin{array}{ccccccccc}
 \frac{\sigma^{(r-1)}_{2,N-r+2}}{\sigma^{(r-1)}_{2,N-r+2}} &  \frac{\sigma^{(r-1)}_{3,N-r+2}}{\sigma^{(r-1)}_{3,N-r+2}+\sigma^{(r-1)}_{3,N-r+3}} &\scriptstyle\cdots  &\scriptstyle\cdots  &\frac{\sigma^{(r-1)}_{M-N+1,N-r+2}}{\sum_{j=N-r+2}^{M-r+1} \sigma^{(r-1)}_{M-N+1,j}} & \scriptstyle {\hat A}^{N-r+2}_{N-r+2} &\scriptstyle 0&\scriptstyle\cdots &\scriptstyle 0\\[0.5ex]
\scriptstyle 0 &  \frac{\sigma^{(r-1)}_{3,N-r+3}}{\sigma^{(r-1)}_{3,N-r+2}+\sigma^{(r-1)}_{3,N-r+3}}  &\scriptstyle\cdots&\scriptstyle \cdots & \frac{\sigma^{(r-1)}_{M-N+1,N-r+3}}{\sum_{j=N-r+2}^{M-r+1} \sigma^{(r-1)}_{M-N+1,j}} & \scriptstyle {\hat A}^{N-r+2}_{N-r+3} &\scriptstyle {\hat A}^{N-r+3}_{N-r+3}&\scriptstyle 0\cdots &\scriptstyle 0 \\[0.5ex]
\scriptstyle\vdots &\scriptstyle \ddots &\scriptstyle  \ddots &\scriptstyle \ddots &\scriptstyle \vdots &\scriptstyle \vdots &\scriptstyle \vdots&\scriptstyle\ddots&\scriptstyle\vdots\\[0.5ex]
\scriptstyle 0  &\scriptstyle  \cdots \;\; 0 & \frac{\sigma^{(r-1)}_{s,N-r+s}}{\sum_{j=N-r+2}^{N-r+s} \sigma^{(r-1)}_{s,j}} &\scriptstyle \cdots & \frac{\sigma^{(r-1)}_{M-N+1,N-r+s}}{\sum_{j=N-r+2}^{M-r+1} \sigma^{(r-1)}_{s,j}} &\scriptstyle {\hat A}^{N-r+2}_{N-r+s} &\scriptstyle {\hat A}^{N-r+3}_{N-r+s}&\scriptstyle \cdots &\scriptstyle {\hat A}^{N}_{N-r+s}\\[0.5ex]
\scriptstyle\vdots &\scriptstyle \ddots & \scriptstyle \ddots & \scriptstyle\ddots &\scriptstyle \vdots &\scriptstyle \vdots &\scriptstyle \vdots &\scriptstyle \vdots &\scriptstyle \vdots\\[0.5ex]
\scriptstyle 0  &\scriptstyle \cdots & \scriptstyle \cdots   & \scriptstyle 0 &\frac{\sigma^{(r-1)}_{M-N+1,M-r+1}}{\sum_{j=N-r+2}^{M-r+1} \sigma^{(r-1)}_{M-N+1,j}} &\scriptstyle {\hat A}^{N-r+2}_{M-r+1} &\scriptstyle {\hat A}^{N-r+3}_{M-r+1}&\scriptstyle \cdots &\scriptstyle {\hat A}^{N}_{M-r+1}\\[0.5ex]
\scriptstyle 0 &\scriptstyle 0 &\scriptstyle \cdots &\scriptstyle 0 &   \frac{\sigma^{(r-1)}_{M-N+1,M-r+2}}{\xi \left(\sum_{j=N-r+2}^{M-r+1} \sigma^{(r-1)}_{M-N+1,j}\right)} &\scriptstyle {\hat A}^{N-r+2}_{M-r+2} &\scriptstyle {\hat A}^{N-r+3}_{M-r+2}&\scriptstyle \cdots &\scriptstyle {\hat A}^{N}_{M-r+2}\\[0.5ex]
\scriptstyle 0 & \scriptstyle 0 &\scriptstyle \cdots &\scriptstyle 0 &   \frac{\sigma^{(r-1)}_{M-N+1,M-r+2}}{\xi^2 \left(\sum_{j=N-r+2}^{M-r+1} \sigma^{(r-1)}_{M-N+1,j}\right)} &\scriptstyle 0 &\scriptstyle {\hat A}^{N-r+3}_{M-r+1}&\scriptstyle \cdots &\scriptstyle {\hat A}^{N}_{M-r+1}\\[0.5ex]
\scriptstyle\vdots &\scriptstyle \vdots &\scriptstyle  \vdots &\scriptstyle \vdots &\scriptstyle \vdots &\scriptstyle \vdots  &\scriptstyle \ddots &\scriptstyle \ddots &\scriptstyle \vdots\\[0.5ex]
\scriptstyle 0 &\scriptstyle 0 &\scriptstyle \cdots &\scriptstyle 0 &   \frac{\sigma^{(r-1)}_{M-N+1,M}}{\xi^{r-1} \left(\sum_{j=N-r+2}^{M-r+1} \sigma^{(r-1)}_{M-N+1,j}\right)} & \scriptstyle 0 &\scriptstyle \cdots &\scriptstyle 0 &\scriptstyle {\hat A}^{N}_{M}\\[0.5ex]
\end{array}
\right].
$}
\]
Then the coefficients 
\[
B^{(r)}_n(\xi) = {\hat B}^{(r)}_{n} \left( 1 + O (\xi^{-1}) \right), \quad\quad n\in [2,M-N+1]
\] 
where ${\hat B}^{(r)}_{n}$ are as in Theorem \ref{lemma:vectors}, while
and $\epsilon^{(r)}_k = O(\xi^{-k})$, $r\in[N-i]$
and at leading order are as in (\ref{ex:Br}) and (\ref{eq:epsilon}). In particular,
 if $\sigma^{(r-1)}_{n,j}$ are all positive then also  ${\hat B}^{(r)}_{l}(\xi)>0$, $l\in [2,M-N+1]$ and $\epsilon^{(r)}_k(\xi)>0$, $k\in[r-1]$ for all $\xi\gg1$.
\end{proof}

In the next Lemma we refine the results of Theorem \ref{lemma:vectors}, namely at double points $\alpha_n^{(r)}$ we estimate 
the vacuum wave function at the leading order in $\xi$ for all phases. We recall that in Theorem \ref{lemma:vectors} only the 
coefficients in front of the dominant phases were computed. 

\begin{lemma}\label{lemma:psiasym}
{Let $r\in [2,N]$ be fixed. Let $\lambda_j$ ($j\in [M-N+1]$) as in (\ref{eq:lambdas}), $\alpha_n^{(r-1)}$ ($n\in [2,M-N+1]$) as in (\ref{eq:alphas}),
\[
\Psi^{(r)} (\zeta,\vec t) = C_1 (\zeta) e^{\theta_{N-r+1}} + \sum_{n=2}^{M-N+1} C_n (\zeta)
\Psi^{(r-1)} (\alpha_n^{(r-1)}, \vec t),
\]
with $\Psi^{(r-1)} (\alpha_n^{(r-1)}, \vec t)$ as in (\ref{eq:psialpha2}),
\[
C_n (\zeta) = {\mathring B}^{(r)}_n\frac{\prod_{j\not =n}^{M-N+1} (\zeta-\lambda_j) }{\prod_{k=1}^{M-N} (\zeta -b^{(r)}_k )}, \quad\quad n\in [M-N+1],
\]
with 
\[
\mathring B^{(r)}_n(\xi) =\left\{ \begin{array}{ll}
\hat A^{N-r+1}_{N-r+1} &\quad n=1\\[0.5ex]
\frac{B^{(r)}_n(\xi)}{1+\sum_{k=1}^{r-1} \epsilon^{(r)}_k (\xi)}, &\quad n\in[2,M-N+1]
\end{array}\right.,
\]
with $B^{(r)}_n(\xi)$, $\epsilon^{(r)}_k(\xi)$ as in Lemma \ref{lemma:beps}, and
$b^{(r)}_k(\xi)$, ($k\in [M-N]$) as in Lemma \ref{lemma:C} with $c_n=\mathring B^{(r)}_n(\xi)$.
Let ${\hat \sigma}^{(r-1)}_{n,j}$ as in (\ref{eq:hatsigma}), with $\sigma^{(r-1}_{n,j}>0$ as in Lemma \ref{lemma:beps}.
 
Then, for $\alpha_n^{(r)}$ ($n\in [2, M-N+1]$) as in (\ref{eq:alphas}), we have 
\begin{equation}
\begin{array}{ll} \Psi^{(r)} (\alpha_n^{(r)} , t) = & 
\left( \sum_{j=N-r+1}^{N-r+n-1} {\hat \sigma}^{(r)}_{n,r}
e^{\theta_{j}} + \sum_{j=N-r+n}^{N+n-2} \frac{ {\hat \sigma}^{(r)}_{n, j}}{\xi^{j-N+r-n+1}}
e^{\theta_{j}} \right.
\\[0.5ex]
&\left. + \sum_{j=n+N-2}^{M} \frac{ {\hat \sigma}^{(r)}_{n,j}}{\xi^{2(j-N-n+2)+r}}
e^{\theta_{j}} \right) \left( 1+ O(\xi^{-1}) \right),
\end{array}
\end{equation}
for uniquely defined positive constants $\sigma^{(r)}_{n,s}$ such that, for any $n\in [2,M-N+1]$, 
\begin{equation}\label{eq:sigmarec}
{\hat \sigma}^{(r)}_{n,s} = \left\{ \begin{array}{ll} 
\frac{{\hat B}^{(r)}_1}{\sum_{{\hat j}=1}^{n-1} {\hat B}^{(r)}_{{\hat j}}} , & s=N-r+1,\\[0.5ex]
\frac{\sum_{j=s+r-N}^{n-1} {\hat B}^{(r)}_j\cdot {\hat \sigma}^{(r-1)}_{j,s}}{\sum_{i=1}^{n-1} {\hat B}^{(r)}_{i} }  , &
s\in [N-r+2,N-r+n-1],\\[0.5ex]
\frac{{\hat B}^{(r)}_{n-1}\cdot {\hat \sigma}^{(r-1)}_{n-1,s}+{\hat B}^{(r)}_n\cdot {\hat \sigma}^{(r-1)}_{n,s} }{\sum_{i=1}^{n-1} {\hat B}^{(r)}_{i} }    , &s\in [N-r+n, N+n-2],\\[0.5ex]
\frac{\sum_{j=n}^{s-N+1} {\hat B}^{(r)}_j\cdot {\hat \sigma}^{(r-1)}_{j,s}}{\sum_{i=1}^{n-1} {\hat B}^{(r)}_{i} } , &s\in [N+n-1, M].
\end{array}
\right.
\end{equation}
Finally, by construction,
\[
\sum_{k=N-r+1}^{N-r+n-1} {\hat \sigma}^{(r)}_{n,k} =1,\quad\quad \forall n\in[2,M-N+1].
\]
}
\end{lemma}

The proof of the above Lemma is straightforward and follows by direct inspection of the leading order in $\xi$ for each phase $\theta_s$, $s\ge N-r+1$, using the definition of $\Psi^{(r)} (\zeta, \vec t)$, and the asymptotic expansions of $C^{(r)} (\alpha_s^{(r)})$, as in (\ref{eq:Caysm}), with $c_j=\mathring B^{(r)}_j = \hat B^{(r)}_j \left( 1+ O(\xi^{-1}) \right)$ and of $\Psi^{(r-1)} (\alpha_s^{(r-1)}, \vec t)$ as in (\ref{eq:psialpha2}).

\begin{remark}
Lemmas \ref{lemma:beps} and \ref{lemma:psiasym} allow to compute the coefficients $B^{(r)}_n$, $\epsilon^{(r)}_k$ and $\sigma^{(r)}_{n,s}$ recursively in $r\in [N]$, starting from the case $r=1$ computed directly in Proposition \ref{prop:1}.

The coefficients $\hat B^{(1)}_n$, $\sigma^{(1)}_{n,k}$ are all positive for $n\in [2,M-N+1]$, $s\in [N-r+2,M]$ by the same Proposition \ref{prop:1}. Moreover $\epsilon^{(r)}_k$ and $\sigma^{(r)}_{n,k}$ respectively in (\ref{eq:epsilon}) and in (\ref{eq:sigmarec}) are subtraction free rational expressions in $\hat B^{(r)}_n$, $\sigma^{(r-1)}_{n,k}$ and the matrix entries of $\hat A$. 
The total positivity property of the matrix $\hat A$ ensures that $ \hat B^{(r)}_n>0$, thanks to  Theorem  \ref{lemma:vectors}. As a consequence we get that also all $\epsilon^{(r)}_k>0$ and $\sigma^{(r)}_{n,s}>0$.

In Theorem \ref{lemma:vectors}, we have computed $\sigma^{(r)}_{n,s}$ for $s\in [N-r+1,N-r+n]$, $n\in [2,M-N+1]$ (see (\ref{eq:Ern}).
In the next Lemma we compute explicitly these coefficients also for $s\in [N+n-1,M]$, $n\in [2,M-N+1]$.
\end{remark}

\begin{lemma}\label{lemma:sigma}
Let $r\in [2,N]$ and suppose that $\sigma^{(r-1)}_{n,s}$, $\hat B^{(r)}_n$ are as in (\ref{eq:sigmaNM}) and (\ref{ex:explB}), respectively.
Then, for any $n\in [2,M_N+1]$, we have
\begin{equation}\label{eq:ttt}
\sigma^{(r)}_{n,j}=\left\{\begin{array}{ll}
\Delta_{[j;N-r+n ,N-r+n {+1},\ldots,N+n-2]}, & \mbox{if} \ \ j\in [N-r {+1}, {N-r+n-1} ] \\[0.5ex]
\Delta_{[N-r+n,N-r+n+1,\ldots,N+n-2;j]} , & \mbox{if} \ \ j\in[N+n-1,M].
\end{array}\right.
\end{equation}
\end{lemma}

\begin{proof}
The case  $j\in [N-r {+1}, {N-r+n-1}]$, $n\in [2,M-N+1]$ is just (\ref{eq:Ern}) which is proven using Lemmas  \ref{lemma:poles} and \ref{lemma:poles2}. The case $j\in [N+n-1,M]$, $n\in [2,M-N+1]$ follows in a similar way using the identity
\begin{equation}
\displaystyle
\sum\limits_{n=k}^{M-N+1} 
\frac{
\Delta_{[N-r+n,\dots,N+n-1]} \cdot \Delta_{[N-r+n+1,\dots,N+n-2;j]} }{\Delta_{[N-r+n+1,\dots,N+n-1]}\cdot\Delta_{[N-r+n,\dots,N+n-2]}}=
 \frac{\Delta_{[N-r+k,\dots,N+k-2;j]} }{\Delta_{[N-r+k,\dots,N+k-2]}},
\end{equation}
for $r\in [N-1]$, $k\in [2,M-N+1]$, $j\in [N+k-1,M]$,
which may be proven recursively along the same lines as for (\ref{eq:ident2}).
\end{proof}

\end{document}